\documentclass[a4paper,onecolumn,11pt,accepted=2026-11-27]{quantumarticle}
\pdfoutput=1
\usepackage[utf8]{inputenc}
\usepackage[english]{babel}
\usepackage[T1]{fontenc}
\usepackage{bbold}
\usepackage{bm}
\usepackage{amsmath}
\usepackage{amssymb}
\usepackage{thmtools}
\usepackage{hyperref}
\usepackage{cleveref}
\usepackage{mathtools}
\usepackage[numbers]{natbib}
\usepackage{float}
\usepackage{appendix}
\usepackage{url}
\usepackage{todonotes}
\usepackage{subcaption}
\usepackage{xcolor}

\usepackage{tikz}
\usetikzlibrary{decorations.pathreplacing}
\usepackage{lipsum}
\usetikzlibrary{decorations.pathreplacing}

\newtheorem{theorem}{Theorem}
\newtheorem{lemma}[theorem]{Lemma}
\newtheorem{corollary}[theorem]{Corollary}

\DeclarePairedDelimiter\bra{\langle}{\rvert}
\DeclarePairedDelimiter\ket{\lvert}{\rangle}
\DeclarePairedDelimiterX\braket[2]{\langle}{\rangle}{#1 \delimsize\vert #2}
\DeclareMathOperator{\Tr}{Tr}

\crefformat{equation}{Eq.~(#2#1#3)}
\crefrangeformat{equation}{Eqs.~(#3#1#4)--(#5#2#6)}
\crefmultiformat{equation}{Eqs.~(#2#1#3)}{ and~(#2#1#3)}{, (#2#1#3)}{ and~(#2#1#3)}
\crefformat{figure}{Fig.~#2#1#3}
\crefname{definition}{Definition}{Definitions}
\crefname{section}{Section}{Sections}
\crefname{corollary}{Corollary}{Corollaries}
\crefname{theorem}{Theorem}{Theorems}
\crefname{lemma}{Lemma}{Lemmas}
\crefname{appsec}{Appendix}{Appendices}
\Crefname{appsec}{Appendix}{Appendices}

\begin{document}

\title{\mbox{Lower Bounds on Spectral Gaps of Parent} Hamiltonians via Tensor Networks}
 
\author[1,2]{Milán Ádám Rozmán}
\orcid{0009-0007-8795-0031}
\email{milan.adam.rozman@univie.ac.at}
\author[3]{András Molnár}
\orcid{0000-0001-5144-438X}
\email{andras.molnar@univie.ac.at}
\author[1,3]{Norbert Schuch}
\orcid{0000-0001-6494-8616}
\email{norbert.schuch@univie.ac.at}
\affil[1]{Faculty of Physics, University of Vienna, Boltzmanngasse 5, 1090 Vienna, Austria}
\affil[2]{Vienna Doctoral School of Physics, Boltzmanngasse 5, 1090 Vienna, Austria}
\affil[3]{Faculty of Mathematics, University of Vienna, Oskar-Morgenstern-Platz 1, 1090 Vienna, Austria}
\maketitle

\hypersetup{pdftitle={Lower Bounds on Spectral Gaps of Parent Hamiltonians via Tensor Networks}}

\begin{abstract}
Proving spectral gaps above the ground space is a central problem in quantum
many-body physics. Yet, finding lower bounds on the gap is notoriously difficult. We
revisit the method originally developed by Fannes, Nachtergaele, and
Werner~\cite{fannes_finitely_1992,nachtergaele_spectral_1996} to prove the
existence of spectral gaps of parent Hamiltonians of Matrix Product States
(MPS), and provide improvements to the different steps of the method. Most importantly, 
we devise a new technique which allows to compute the key quantity in the method
-- the overlap of local ground spaces -- exactly \emph{and} efficiently. This
enables a clear improvement of the method, allowing it to outperform other
existing techniques to lower bound gaps, which we demonstrate by benchmarking
on several models. Remarkably, our -- numerically motivated -- approach at the
same time also yields a significantly simplified proof of the fact that parent
Hamiltonians of any (well-behaved) MPS are always gapped.
\end{abstract}

\section{Introduction}
In quantum many-body systems, the existence of a gap -- that is, a uniform lower
bound on the spectral gap above the ground state manifold, independent of the
system size -- has far-reaching consequences. It implies that correlations decay
exponentially~\cite{hastings_area_2007,hastings_spectral_2006}, and global
properties such as topological order remain stable under
perturbations~\cite{bravyi_topological_2010,michalakis_stability_2013}. It also
underlies quasi-adiabatic evolution~\cite{hastings_quasi-adiabatic_2005}, which
establishes that ground states of quantum many-body Hamiltonians connected by a
gapped path can be transformed into each other by a local evolution, leaving
global entanglement patterns intact.  Together, this provides the theoretical
underpinning for the central concept of gapped quantum phases of matter.

Yet, rigorously proving the existence of a spectral gap for interacting
quantum many-body Hamiltonians is notoriously difficult. A breakthrough
result was the construction of the Affleck-Kennedy-Lieb-Tasaki (AKLT)
model~\cite{affleck_valence_1988}. The key novelty of its approach is to first explicitly construct a quantum
many-body state on a spin chain, from which a ``parent Hamiltonian'' is
derived, which has this state as its exact ground state. For this parent
Hamiltonian, the existence of a spectral gap can then be rigorously
proven.  The AKLT construction was subsequently generalized by Fannes,
Nachtergaele, and Werner (FNW)~\cite{fannes_finitely_1992}, and later again by
Nachtergaele~\cite{nachtergaele_spectral_1996}, by introducing the class of
Finitely Correlated States (FCS): These states appear as exact ground states of
local parent Hamiltonians, and those parent Hamiltonians can be proven to be
gapped. In modern terms, FCS correspond to translation invariant Matrix Product
States (MPS)~\cite{perez-garcia_matrix_2007,cirac_review} which satisfy some natural and
generic properties (which we will define later on).  A key feature shared by
these models is frustration-freeness: The ground state of the Hamiltonian
$H=\sum h_i$ is also a ground state of each $h_i$ individually.  

Beyond proving the mere existence of a gap, obtaining explicit lower
bounds on the gap is of central importance, as it correspondingly implies
quantitative bounds on stability, correlation decay, or the complexity of a
quasi-adiabatic path.  Not long after the original AKLT work, Knabe
devised a method to obtain quantitative bounds on the gap, by relating the
gap of the finite-size AKLT Hamiltonian (which can -- at least in
principle -- be computed numerically) to its infinite-size gap~\cite{knabe_energy_1988}. Later
generalized by Gosset and Mozgunov~\cite{gosset_local_2016}, these approaches are now referred to
as finite-size criteria.  They apply to general frustration-free
Hamiltonians, but therefore also don't leverage a possible MPS structure
of the ground space.  At the same time, the work of FNW and Nachtergaele~\cite{fannes_finitely_1992,nachtergaele_spectral_1996}, the latter's work being
frequently referred to as the ``martingale method'', not only proved the
existence of a gap for parent Hamiltonians of MPS, but also yielded quantitative
lower bounds on its value. The key idea of their method is to first relate the
original Hamiltonian to a blocked Hamiltonian, and subsequently bound the gap in
terms of the angle between the ground state spaces of the terms of the blocked
Hamiltonian and their individual spectral gap.  Based on this method,
bounds on spectral gaps have also been derived for a number of other models,
including the AKLT Hamiltonian in two dimensions~\cite{abdul-rahman_class_2020,guo_nonzero_2021,pomata_demonstrating_2020}.  

In this paper, we revisit the method of FNW, with the goal of
turning it into a tool which gives strong quantitative lower bounds on the
gap. To this end, we improve their method in two places. First, we refine
and enlarge the blocking schemes used, which gives both stronger bounds on
the gap and additional controls to optimize the bound. Our key
contribution is then the second refinement: We show how to compute the exact
value of the overlap of the ground spaces -- the key quantity in the FNW method -- efficiently, at a cost which only scales linearly
with the block size. This contrasts with computing the exact value
based on exact diagonalization (which scales exponentially), and in fact
comes at the same computational cost as the evaluation of FNW's original
bound on the overlap, which however yields significantly weaker bounds on the gap.
The key ingredient is a transformation which relates the corresponding
overlap to an eigenvalue problem of an operator which is defined on a space of
fixed dimension $D^4$  (with $D$ the bond dimension of the MPS)
and which is efficiently computable.  Remarkably,
this approach not only yields improved numerical gap bounds as
compared to FNW: It also allows us to give a significantly easier and unified  proof of the fact that parent
Hamiltonians of (well-behaved) MPS are always gapped. This is enabled by
the fact that the overlap of ground spaces, which enters the gap bound, 
is now an eigenvalue problem on a space of fixed and finite dimension,
where convergence can be easily analyzed.

We illustrate our method through the study of a number of examples with both
unique and degenerate ground spaces, including the AKLT
model~\cite{affleck_valence_1988}, the $\mathbb Z_3$-XY model (a family
constructed from deformations of the $\mathbb Z_3$-Potts
model)~\cite{wouters_interrelations_2021} and two models constructed from
$\mathrm{SU}(3)$ valence bond states~\cite{greiter_valence_2007}, and show that
our method is able to outperform not only the FNW method, but also the
finite-size criteria by Knabe and Gosset-Mozgunov. Depending on the use case,
our method not only provides quantitatively better bounds on the gap, but 
can also give explicit quantitative bounds on the gap where the other
methods fail to return any bound at all.

The paper aims at presenting the material in a self-contained style: In
particular, it includes a detailed introduction to MPS and their parent
Hamiltonians, as well as a summary of the method of FNW and Nachtergaele. We
start by defining MPS and introducing the relevant conditions on their
well-behavedness (normal, block-normal, injective, block-injective MPS,
canonical forms) in \cref{sec:mps}. In \cref{sec:parenthams}, we introduce
parent Hamiltonians and the conditions they need to satisfy to have a
well-behaved ground space. \cref{sec:gaps-existence-bounds} presents our key
result: The refined FNW method for computing accurate lower bounds on the
spectral gap of quantum spin chains. It also presents the new, simple proof that
any parent Hamiltonian of a well-behaved MPS is gapped which emerges from our
method. Finally, in \cref{sec:examples}, we apply our method to a number of
examples, both to illustrate its application and to demonstrate its power. For
completeness, we provide additional details, in particular a summary of the
original results of FNW and Nachtergaele, in the Appendix.

\section{Matrix Product States}
\label{sec1}
\label{sec:mps}

Let us start by introducing the formalism of Matrix Product States (MPS).
A \emph{(translation invariant) MPS} $\ket{\Psi_n (A, X)}$ on $n$ lattice sites of a chain with open boundaries, generated by a tensor $A \in \mathbb{C}^D \otimes (\mathbb{C}^D)^* \otimes \mathbb{C}^d$,
    \begin{equation}
    \label{tensordef}
        A = \sum_{\alpha, \beta, i} A^i_{\alpha \beta} \ket{\alpha}\bra{\beta} \otimes \ket{i} = \sum_i A^i \otimes \ket{i},
    \end{equation}
is given by
    \begin{equation}
    \label{mpsdef}
        \ket{\Psi_n(A, X)} = \sum_{i_1,\dots,i_n} \Tr (X A^{i_1} \cdots A^{i_n}) \ket{i_1\dots i_n},
    \end{equation} 
where $X \in \mathcal{M}_D$ is an arbitrary $D \times D$ matrix serving as the boundary condition. We call $D$ the \emph{bond dimension} of the MPS and $d$ the \emph{physical dimension}. Similarly, $\alpha$ and $\beta$ are called \emph{bond indices} and $i_1,\dots,i_n$ are called \emph{physical indices}.
For the choice $X = \mathbb{1}$ we call $\ket{\Psi_n(A,\mathbb{1})}$ the MPS with periodic boundary conditions, generated by $A$.

We now introduce a graphical notation~\cite{bridgeman_hand-waving_2017,cirac_review} to describe MPS and other related notions. The tensor $A$ is represented by a
dot, with bond indices drawn as horizontally oriented lines, and the physical
index as a line oriented vertically:
\begin{equation}
    A = \begin{tikzpicture}[baseline={([yshift=-.5ex]current bounding box.center)}]
    \draw (0,0) -- (0.8,0);
    \draw (0.4,0) -- (0.4,0.4);
    \fill (0.4,0) circle (2pt);
    \draw (0.4,0) node[anchor=north] {\footnotesize $A$};
    \draw (0.25,0.4) node {\footnotesize $i$};
    \draw (-0.15,0) node{\footnotesize $\alpha$};
    \draw (0.95,0) node{\footnotesize $\beta$};
    \end{tikzpicture} \ .
\end{equation}
We will also refer to these lines as \emph{legs}. Note that what is relevant in identifying the different legs is the orientation in which the legs are attached to the tensor (the dot), as the lines can be bent in diagrams. The MPS $\ket{\Psi_n(A,X)}$ defined in \cref{mpsdef} can then be graphically represented as a network of tensors by denoting contractions -- i.e., summation of double indices -- by connecting the respective legs (here, the bond indices), while the legs corresponding to the free (physical) indices remain open (unconnected):
\begin{equation}
\label{eq:Psi-n-M-X}
    \ket{\Psi_n(A, X)} = \begin{tikzpicture}[baseline={([yshift=-.5ex]current bounding box.center)}]
        \draw (0,0) -- (2.5,0);
        \draw (0,-0.8) -- (0,0);
        \draw (2.5,-0.8) -- (2.5,0);
        \draw (0,-0.8) -- (2.5,-0.8);
        \fill (0.4,0) circle (2pt);
        \draw (0.4,0) -- (0.4,0.4);
        \draw (0.4,0) node[anchor=north]{\footnotesize $A$};
        \fill (1,0) circle (2pt);
        \draw (1,0) -- (1,0.4);
        \draw (1,0) node[anchor=north]{\footnotesize $A$};
        \draw (2,0) -- (2,0.4);
        \fill (2,0) circle (2pt);
        \draw (2,0) node[anchor=north]{\footnotesize $A$};
        \draw (1.5,0) node[anchor=south]{\footnotesize $\dots$};
        \fill (1.25,-0.8) circle (2pt);
        \draw (1.25,-0.8) node[anchor=north]{\footnotesize $X$};
    \end{tikzpicture} \ . 
\end{equation}
In the remainder of this article, we will refer to such a graphical notation as the \emph{tensor network notation}~\cite{bridgeman_hand-waving_2017,cirac_review}.

For a given MPS tensor $A$ we introduce the four-index tensor
\begin{equation}
\label{fourleg}
\mathcal T:=\
    \begin{tikzpicture}[baseline={([yshift=-.5ex]current bounding box.center)}]
        \draw (0,0) -- (0,0.8);
        \draw (-0.4,0) -- (0.4,0);
        \draw (-0.4,0.8) -- (0.4,0.8);
        \fill (0,0) circle (2pt);
        \fill (0,0.8) circle (2pt);
        \draw (0,0) node[anchor=north]{\footnotesize $A$};
        \draw (0,0.8) node[anchor=south]{\footnotesize $\bar{A}$};
    \end{tikzpicture}\ ,
\end{equation}
where $\bar{A}$ denotes the complex conjugate of $A$. There exist two natural ways to interpret this object as a linear map: we can view it as a map from the right indices to the left indices, resulting in the \emph{horizontal transfer matrix} $\mathcal{T} = \sum_i A^{i} \otimes \overline{A^{i}}$ of the MPS tensor $A$. Alternatively, we can also define the \emph{vertical transfer matrix}, which we will denote by $T$, by considering the object as a linear map from the bottom to the top indices: $T: X \mapsto \sum_i \Tr (A^{i} X) \cdot \overline{A^{i}}$. 
To distinguish these two maps in the graphical notation, we will adopt the convention that the alignment of the open legs indicates the direction in which the map acts: If the legs are aligned horizontally, the map acts from right to left, as in Eq.~\eqref{fourleg}, which hence denotes the map $\mathcal T$. If the legs are aligned vertically, the map acts from bottom to top, and the map $T$ is thus graphically denoted as 
\begin{equation}
    T = \begin{tikzpicture}[baseline={([yshift=-.5ex]current bounding box.center)}]
        \draw (0,0) -- (0,0.8);
        \draw (-0.4,0) -- (0.4,0);
        \draw (-0.4,0.8) -- (0.4,0.8);
        \draw (-0.4,0) -- (-0.4,-0.4);
        \draw (0.4,0) -- (0.4,-0.4);
        \draw (-0.4,0.8) -- (-0.4,1.2);
        \draw (0.4,0.8) -- (0.4,1.2);
        \fill (0,0) circle (2pt);
        \fill (0,0.8) circle (2pt);
        \draw (0,0) node[anchor=north]{\footnotesize $A$};
        \draw (0,0.8) node[anchor=south]{\footnotesize $\bar{A}$};
    \end{tikzpicture} = \begin{tikzpicture}[baseline={([yshift=-.5ex]current bounding box.center)}]
        \draw (-0.2,-0.2) -- (0.2,-0.2);
        \draw (-0.2,-0.2) -- (-0.2,1);
        \draw (-0.2,1) -- (0.2,1);
        \draw (0.2,1) -- (0.2,-0.2);
        \draw (-0.4,0) -- (-0.2,0);
        \draw (0.2,0) -- (0.4,0);
        \draw (-0.4,0.8) -- (-0.2,0.8);
        \draw (0.2,0.8) -- (0.4,0.8);
        \draw (-0.4,0) -- (-0.4,-0.4);
        \draw (0.4,0) -- (0.4,-0.4);
        \draw (-0.4,0.8) -- (-0.4,1.2);
        \draw (0.4,0.8) -- (0.4,1.2);
        \draw (0,0.4) node[]{\small $\mathcal{T}$};
        \end{tikzpicture}\ .
\end{equation}
In the following, we will assume that all MPS tensors are normalized in such a way that the eigenvalue of $\mathcal{T}$ with the largest magnitude is 1 (as a completely positive map, there is always a positive eigenvalue of largest magnitude~\cite{perez-garcia_matrix_2007}). 

A tensor $B$ is said to be the $\ell$\emph{-fold blocking} of $A$ if
\begin{equation}
    B = \sum_{i_1, \dots, i_\ell} A^{i_1}\cdots A^{i_\ell} \otimes \ket{i_1\dots i_\ell}.
\end{equation}
As a consequence, we have $\ket{\Psi_{\ell n} (A,X)} = \ket{\Psi_n(B,X)}$ \cite{molnar_generalization_2018}. By construction, the horizontal transfer matrix of the $n$-fold blocking of $A$ is $\mathcal{T}^n$. We will denote by $T_n$ the vertical transfer matrix resulting from $\mathcal{T}^n$:
\begin{equation}
\label{vertrdef}
    T_n = \begin{tikzpicture}[baseline={([yshift=-.5ex]current bounding box.center)}]
        \draw (-0.75,0.25) -- (0.75,0.25);
        \draw (-0.75,0.25) -- (-0.75,0.75);
        \draw (-0.75,0.75) -- (0.75,0.75);
        \draw (0.75,0.75) -- (0.75,0.25);
        \draw (-0.5,0.25) -- (-0.5,0);
        \draw (0.5,0.25) -- (0.5,0);
        \draw (-0.5,0.75) -- (-0.5,1);
        \draw (0.5,0.75) -- (0.5,1);
        \draw (0,0.5) node[]{$T_n$};
        \end{tikzpicture} 
   := 
    \begin{tikzpicture}[baseline={([yshift=-.5ex]current bounding box.center)}]
        \draw (-0.3,0.25) -- (0.3,0.25);
        \draw (-0.3,0.25) -- (-0.3,0.75);
        \draw (-0.3,0.75) -- (0.3,0.75);
        \draw (0.3,0.75) -- (0.3,0.25);
        \draw (-0.5,0.375) -- (-0.3,0.375);
        \draw (0.3,0.375) -- (0.5,0.375);
        \draw (-0.5,0.625) -- (-0.3,0.625);
        \draw (0.3,0.625) -- (0.5,0.625);
        \draw (-0.5,0.375) -- (-0.5,0);
        \draw (0.5,0.375) -- (0.5,0);
        \draw (-0.5,0.625) -- (-0.5,1);
        \draw (0.5,0.625) -- (0.5,1);
        \draw (0,0.5) node[]{\small $\mathcal{T}^n$};
        \end{tikzpicture}
 = 
        \begin{tikzpicture}[baseline={([yshift=-.5ex]current bounding box.center)}]
        \draw (0,0) -- (2.5,0);
        \draw (0,-0.4) -- (0,0);
        \draw (2.5,-0.4) -- (2.5,0);
        \draw (0,0.8) -- (2.5,0.8);
        \draw (0,0.8) -- (0,1.2);
        \draw (2.5,0.8) -- (2.5,1.2);
        \fill (0.4,0) circle (2pt);
        \fill (0.4,0.8) circle (2pt);
        \draw (0.4,0) -- (0.4,0.8);
        \draw (0.4,0) node[anchor=north]{\footnotesize $A$};
        \draw (0.4,0.8) node[anchor=south]{\footnotesize $\bar{A}$};
        \fill (1,0) circle (2pt);
        \fill (1,0.8) circle (2pt);
        \draw (1,0) -- (1,0.8);
        \draw (1,0) node[anchor=north]{\footnotesize $A$};
        \draw (1,0.8) node[anchor=south]{\footnotesize $\bar{A}$};
        \draw (2,0) -- (2,0.8);
        \fill (2,0) circle (2pt);
        \fill (2,0.8) circle (2pt);
        \draw (2,0) node[anchor=north]{\footnotesize $A$};
        \draw (2,0.8) node[anchor=south]{\footnotesize $\bar{A}$};
        \draw (1.5,0) node[anchor=south]{\footnotesize $\dots$};
        \draw (1.5,0.8) node[anchor=north]{\footnotesize $\dots$};
    \end{tikzpicture} \ .
\end{equation}
This map can be constructed efficiently (i.e., at a polynomial -- specifically, linear -- cost in $n$) by taking the $n$'th power of the $D^2\times D^2$ matrix $\mathcal{T}$, and then regrouping its indices.

\subsubsection*{Normal and injective MPS}

Now, we define special subclasses of MPS and list their relevant properties. We call an MPS tensor \emph{injective} if its vertical transfer matrix $T$ has rank $D^2$. Equivalently, an MPS tensor $A$ is injective if
\begin{equation}
    \operatorname{span}\{A^1,\dots,A^d\} = \mathcal{M}_{D} \ .
\end{equation}
If there exists a natural number $\ell$ such that $T_\ell$ has rank $D^2$, or equivalently, the $\ell$-fold blocking of the tensor generating the MPS spans the space $\mathcal{M}_D$, we say that the MPS tensor is \emph{normal} \cite{schuch_peps_2010}. We call the smallest such $\ell$ the \emph{injectivity length} of the MPS. An MPS generated by an injective/normal tensor will be also called injective/normal.

A normal MPS tensor $A$ is said to be in the \emph{(left) canonical form}, if it satisfies
\begin{equation}
\label{canonical}
        \begin{tikzpicture}[baseline={([yshift=-.5ex]current bounding box.center)}]
        \draw (0,0) -- (0.8,0);
        \draw (0,0) -- (0,0.8);
        \draw (0,0.8) -- (0.8,0.8);
        \draw (0.4,0) -- (0.4,0.8);
        \fill (0.4,0) circle (2pt);
        \fill (0.4,0.8) circle (2pt);
        \draw (0.4,0) node[anchor=north]{\footnotesize $A$};
        \draw (0.4,0.8) node[anchor=south]{\footnotesize $\bar{A}$};
        \end{tikzpicture} = \begin{tikzpicture}[baseline={([yshift=-.5ex]current bounding box.center)}]
        \draw (0,0) -- (0.4,0);
        \draw (0,0) -- (0,0.8);
        \draw (0,0.8) -- (0.4,0.8);
        \end{tikzpicture}\ , \qquad \begin{tikzpicture}[baseline={([yshift=-.5ex]current bounding box.center)}]
        \draw (0,0) -- (0.8,0);
        \draw (0.8,0) -- (0.8,0.8);
        \draw (0,0.8) -- (0.8,0.8);
        \draw (0.4,0) -- (0.4,0.8);
        \fill (0.4,0) circle (2pt);
        \fill (0.4,0.8) circle (2pt);
        \fill (0.8,0.4) circle (2pt);
        \draw (0.4,0) node[anchor=north]{\footnotesize $A$};
        \draw (0.4,0.8) node[anchor=south]{\footnotesize $\bar{A}$};
        \draw (0.8,0.4) node[anchor=west] {$\rho$};
        \end{tikzpicture} = \begin{tikzpicture}[baseline={([yshift=-.5ex]current bounding box.center)}]
        \draw (0,0) -- (0.4,0);
        \draw (0.4,0) -- (0.4,0.8);
        \draw (0,0.8) -- (0.4,0.8);
        \fill (0.4,0.4) circle (2pt);
        \draw (0.4,0.4) node[anchor=west] {$\rho$};
        \end{tikzpicture},
\end{equation}
where $\rho$ is a positive definite matrix with $\Tr(\rho)=1$, namely the unique (right) eigenvector with eigenvalue $+1$ of the horizontal transfer matrix.
Any normal MPS can be brought into canonical form. More precisely, for any normal MPS tensor $\tilde{A}$ there exists a normal tensor $A$ satisfying \cref{canonical} such that $\ket{\Psi_n(\tilde{A},X)} = \ket{\Psi_n(A,YXY^{-1})}$ for some invertible $Y$ and for any $n\in \mathbb{N}$, see Ref.~\cite{perez-garcia_matrix_2007} for details. In addition, the horizontal transfer matrix $\mathcal{T}$ of a normal MPS in left canonical form satisfies \cite{fannes_finitely_1992}
\begin{equation}
\label{transferlimit}
    \lim_{n \to \infty}\mathcal{T}^n = \begin{tikzpicture}[baseline={([yshift=-.5ex]current bounding box.center)}]
    \draw (0,0) -- (0.4,0);
    \draw (0.4,0) -- (0.4,0.8);
    \draw (0,0.8) -- (0.4,0.8);
    \draw (0.8,0) -- (1.2,0);
    \draw (0.8,0) -- (0.8,0.8);
    \draw (0.8,0.8) -- (1.2,0.8);
    \fill (0.4,0.4) circle (2pt);
    \draw (0.4,0.4) node[anchor=east] {$\rho$};
    \end{tikzpicture} \eqcolon \mathcal{T}_{\infty}\ .
\end{equation}

\subsubsection*{Block-normal and block-injective MPS}

Finally, we introduce a family of MPS that is more general than normal MPS. Let $\{A_\alpha\}_{\alpha=1}^k$ be normal tensors that generate different states for large enough system sizes.\footnote{Otherwise, the tensors generating the same state would be related by a gauge transformation~\cite{molnar_generalization_2018}.} We define the MPS tensor $C$ by $C^{i} = \bigoplus_{\alpha=1}^{k} A^{i}_\alpha$. The MPS generated by $C$ then represents a linear combination of the states generated by the $A_\alpha$. In this case, we call $A_\alpha$ \emph{blocks} and the resulting tensor $C$ \emph{block-normal}. If each $A_\alpha$ is injective, then $C$ is called \emph{block-injective}.

Let us note that \emph{any} MPS with periodic boundaries can be generated by a tensor whose blocks satisfy \cref{canonical}, however these blocks are not necessarily normal~\cite{perez-garcia_matrix_2007}. 

\section{Hamiltonians with matrix product ground states}
\label{sec:parenthams}

\subsection{Weak parent Hamiltonians}\label{sec:weakparent}

Given an MPS generated by a tensor $A$, we can construct Hamiltonians which have this MPS as an exact ground state. To this end, we proceed as follows.  Define 
\begin{equation}
\label{mpsspace}
    \mathcal{S}_n := \{\ket{\Psi_n(A,X)}|X\in \mathcal{M}_D\}\ ,
\end{equation}
that is, $\mathcal S_n$ is the space spanned by the MPS on $n$ sites for all possible choices of boundary conditions. From now on, we restrict to MPS where $\mathcal{S}_n \neq \{0\}$ for all $n\in \mathbb{N}$.
Choose $h$ to be any positive semidefinite operator acting on $n$ lattice sites with the property
\begin{equation}
\label{locham}
    \ker h = \mathcal{S}_n\,.
\end{equation}
Then $h$ annihilates any state in $\mathcal{S}_n$, i.e., $h\ket{\Psi} = 0 \ \forall \ket{\Psi} \in \mathcal{S}_n$. We will say that $h$ is $n$\emph{-local}, and refer to $n$ as the \emph{interaction range} of the Hamiltonian.

We can now define a Hamiltonian $H_N$ on a system size $N\geq n$ with open boundaries by summing up translations of the $n$-local Hamiltonian $h$:
\begin{equation}
\label{globham}
    H_N = \sum_{i=1}^{N-n+1} h_i\,,
\end{equation}
with $h_i = \mathbb{1} \otimes \dots \otimes \mathbb{1} \otimes  h \otimes \mathbb{1} \otimes \dots \otimes \mathbb{1}$, where the $h$ is placed at sites  $(i,i+1,\dots,i+n-1)$. This Hamiltonian has the following properties:
\begin{itemize}
    \item $H_N$ is positive semidefinite.
    \item $H_N \ket{\Psi_N(A,X)} = 0 \ \forall X$, or equivalently, $\mathcal{S}_N \subseteq \ker H_N$.
    \item From this, it follows that the  ground state energy of $H_N$ is $0$, and $\ket{\Psi_N(A,X)}\ \forall\,X$ are ground states of $H_N$. 
    \item In this case, since the terms $h_i$ are positive semidefinite as well, any ground state also minimizes the energy of each $h_i$ individually; this property is called \emph{frustration-freeness}.
\end{itemize}
We call local Hamiltonians satisfying these properties \emph{weak parent Hamiltonians} of the MPS. Note that the value $0$ of the ground state energy of weak parent Hamiltonians is an arbitrary choice, and can be changed by adding a constant energy density $e$ to each $h$.
A common choice of a weak parent Hamiltonian, which we refer to as the \emph{canonical weak parent
Hamiltonian}, is obtained by choosing the $n$-local terms $h$ as
\begin{equation}
\label{localparent}
    h = \mathbb{1}-P_n\ ,
\end{equation}
where $P_n$ is defined to be the orthogonal projector onto $\mathcal{S}_n$.

\subsection{Parent Hamiltonians\label{subsec:parenthams}}

In this work, we are interested in weak parent Hamiltonians that satisfy a
stricter condition: Their ground space is spanned precisely by the MPS
$\ket{\Psi_N(A,X)}\ \forall\,X$ for large enough $N$, i.e., $\ker H_N =\mathcal
S_N$. We will call such Hamiltonians \emph{parent Hamiltonians} of the given MPS
(and correspondingly \emph{canonical parent Hamiltonians} if they are of the
form (\ref{localparent})). We will now discuss conditions which imply that a
weak parent Hamiltonian is, in fact, a parent Hamiltonian.

\subsubsection*{Normal MPS and parent Hamiltonians with unique ground states}

Let us first consider normal MPS. For normal MPS with injectivity length $\ell$,
a weak parent Hamiltonian $H_N$ is a parent Hamiltonian if $h$ has interaction range at least $\ell
+1$ \cite{schuch_peps_2010}. If the interaction range of $h$ is smaller, $H_N$
is still a parent Hamiltonian for $N\ge N_0$ if there is an $N_0\ge \ell+1$ for
which one can check that $\ker{H_{N_0}}=\mathcal S_{N_0}$, as this implies that
$H_N$ has the same ground state space as a weak parent Hamiltonian with interaction range  $N_0\ge \ell+1$.
Note that since $\mathcal S_{N_0}\subseteq \ker H_{N_0}$, it is sufficient to
check that $\dim \ker H_{N_0} \le \dim \mathcal{S}_{N_0} = D^2$ to establish
that $H_N$, $N\ge N_0$, is a parent Hamiltonian.  

For parent Hamiltonians defined on systems with periodic boundary conditions
(and system size $N\ge N_0\ge \ell+1$), the ground state becomes unique and is
given by $\ket{\Psi_N(A,\mathbb 1)}$ \cite{schuch_peps_2010,cirac_review}.
Moreover, for open boundaries, the ground states $\ket{\Psi_N(A,X)}$ only
differ at the boundary (as the dependence on the boundary conditions decays exponentially into the bulk~\cite{cirac_review}), and thus they all describe the same state in the bulk
as the system size is taken to infinity; in fact, on the infinite chain, $H_N$
has a unique ground state, given by the pure finitely correlated state (i.e., the
infinite-chain version of normal MPS) defined by $A$~\cite{fannes_finitely_1992}. In light of this, we will also apply the terminology ``$H_N$ has a unique
ground state'' to MPS parent Hamiltonians $H_N$ on open boundary condition
chains, or when boundary conditions are unspecified,  meaning the ground state
is unique in the bulk (and with periodic boundaries). 

\subsubsection*{Block-normal MPS and parent Hamiltonians with degenerate ground states}

Let us now turn to block-normal MPS. Here, the condition for a weak parent
Hamiltonian $H_N$ to be a parent Hamiltonian -- that is, for its ground space to
be given by $\mathcal S_N$ -- is as follows. Let $C$ be a block-normal MPS
tensor consisting of blocks $\{A_\alpha\}_{\alpha=1}^{k}$
with bond dimensions $\{D_\alpha\}$ and injectivity lengths $\{\ell_\alpha\}$.
Further, let  $P^{(\alpha)}_n$ denote the orthogonal projection onto $\mathcal
S^{(\alpha)}_n := \big\{\ket{\Psi_n(A_{\alpha},X_\alpha)}\vert X_\alpha \in \mathcal{M}_{D_\alpha} \big\}$. 
If the interaction range $n$ of $h$ is such that (i)
$n>\ell_\mathrm{max}:=\max_\alpha \ell_\alpha$, and (ii)
$\max\operatorname{spec}(P^{(\alpha)}_n P^{(\beta)}_n) < \frac{1}{(k-1)^2}$ for
all $\alpha\ne\beta$,\footnote{\label{footnote:orth}Note that this always holds for sufficiently
large $n$: For block-normal MPS, $\max\operatorname{spec}(P^{(\alpha)}_n
P^{(\beta)}_n)\to0$ for $n\to\infty$~\cite{molnar_generalization_2018}.
Also note that $\operatorname{spec}(P^{(\alpha)}_n P^{(\beta)}_n)\ge0$.}
then the ground space of $H_N$ for $N\ge n$ is given by the direct sum $\bigoplus_\alpha
\mathcal S_N^{(\alpha)}$, that is, it is spanned by the linearly independent MPS
$\ket{\Psi_N(A_\alpha,X)}$, and thus has dimension $\sum_{\alpha=1}^{k}
D_\alpha^2$ . As before, for a $h$ with a smaller interaction range to be a
parent Hamiltonian, it is sufficient to show that for some $N_0>\ell_\mathrm{max}$, $\ker H_{N_0}$ coincides with $\bigoplus_\alpha \mathcal S_{N_0}^{(\alpha)}$ (cf.\ Footnote~\ref{footnote:orth}).
Again, states with different boundary conditions $X_\alpha$ do not differ in the
bulk as $N\to \infty$, and the ground state is $k$-fold degenerate for periodic
boundaries and on infinite chains. We thus once again adopt the terminology ``$H_N$ has a $k$-fold degenerate ground state'' (that is, in the bulk) also for the open boundary
condition Hamiltonian $H_N$~\cite{nachtergaele_spectral_1996,schuch_peps_2010}.

\section{Spectral gaps of parent Hamiltonians: Existence and bounds}
\label{sec3}
\label{sec:gaps-existence-bounds}

This section contains the central results of our work: We present a method
(with different variants) which allows to compute rigorous lower bounds on
the spectral gap of parent Hamiltonians of normal and block-normal MPS. At the same time,
the technique we develop also yields a simple yet rigorous proof of the
fact that any parent Hamiltonian of a normal or block-normal MPS has a
spectral gap (even in cases where the technique fails to provide an actual
lower bound due to computational limitations).  Applications of these
results are illustrated through a number of examples in
\cref{sec:examples}.

The key idea of the method consists of two steps. The first step relates
the spectral gap of the original Hamiltonian to the gap of a blocked
Hamiltonian, which is obtained by blocking sites into super-sites after
rearranging terms of the original Hamiltonian. This blocking can be done
in different ways, covered in \cref{lem:block-obc} and
\cref{lem:block-pbc}. The reason for this blocking is two-fold, both
relevant to the second step: First, blocking gives rise to a nearest-neighbor
Hamiltonian, and second, increasing the block size leads to better bounds
on the gap.  The second step then gives a lower bound on the gap of
the blocked parent Hamiltonian which is independent of the system size, based
on a local condition involving the relative alignment of the ground spaces of
two adjacent terms of the blocked Hamiltonian. The optimal value for this local
condition can be found numerically, providing an explicit bound on the gap
(though it might be trivial). We prove that this bound converges to a non-zero
value as the blocking is increased, thus in particular proving the existence of
a gap for all parent Hamiltonians. In fact, for the blocking of
\cref{lem:block-pbc}, we show that the bound obtained converges to the
actual value of the gap for open boundaries when $N\to\infty$.

The technique we present is based on a method
developed originally by Fannes, Nachtergaele, and Werner (FNW) (see Ref.~\cite{fannes_finitely_1992}) for pure
finitely correlated states (FCS) -- those precisely amount to normal
translation invariant MPS in modern terminology. This method has later been
generalized by Nachtergaele (the ``martingale method'') to block-normal MPS (again in the language of
FCS)~\cite{nachtergaele_spectral_1996}. Our results extend the results of FNW and
Nachtergaele in several ways: They strengthen both the first and the second
step of the argument, and thus achieve better bounds on the gap.  At the
same time, the technique which we develop for the second step allows to
also significantly simplify the proof of the existence of a gap.  Finally,
our approach treats parent Hamiltonians of normal and block-normal MPS on
a unified footing -- in fact, all we require is that the ground space of
the Hamiltonian on any block length is parametrized by \emph{some} MPS
(which can depend on the block size), enabling the application of our
method to general frustration free Hamiltonians.

Throughout the section, we will clarify the relation of our results to
those of FNW and Nachtergaele; in addition, for reference we present their
original results, partly with proofs, in~\cref{appendix:fnw-n}, rephrased
in the modern terminology of MPS.  Let us stress that the main goal of FNW
and Nachtergaele was not to obtain quantitative bounds on the gap, but to
prove its existence, and it is thus not surprising that their quantitative
bounds can be improved; yet, remarkably, our approach also leads to a much
simplified proof of the existence of a gap as such.

In what follows, we will denote the gap of a Hamiltonian $H$ by
$\Delta(H)$.  Specifically,  $\Delta(H)$ is the difference between the
first and the second \emph{distinct} lowest eigenvalue. Since our
Hamiltonians are all frustration free with the ground state energy set to
zero, $\Delta(H)$ equals the smallest non-zero eigenvalue of $H$.

\subsection{Reshuffling and blocking of the Hamiltonian}

We start by discussing the first step of the method: Relating the spectral
gap of the parent Hamiltonian to the spectral gap of a blocked
Hamiltonian. All results in this subsection apply to any
frustration-free translation invariant Hamiltonian. Yet, we will 
consider MPS parent Hamiltonians, as we will require later on some statements which are 
specific to that setting.
\cref{lem:block-obc} applies independently of the boundary
conditions, while the stronger \cref{lem:block-pbc} requires periodic
boundaries. The construction of the Hamiltonian in the Lemmas and their proofs
are illustrated in \cref{fig:lemma-blocking} -- some readers might
find those more accessible than the index-heavy notation of the lemmas.

Recall that the original Hamiltonian is $H_N=\sum_{i=1}^{N-n+1} h_i$, with
interaction range $n$, and $h_i$ acting on sites $(i,i+1,\dots,i+n-1)$. We consider both the setting with OBC and PBC for $H_N$.

\begin{lemma}
\label{lem:block-obc}
Let $p\ge n-1$, $q\ge0$, and
$N=p(M+1)+qM$ for $M\in\mathbb N$ (for OBC) or
$N=(p+q)M$ for $M\ge 3$ (for PBC).
 We block sites in alternating blocks of size
$p$ and $q$. The blocked Hamiltonian $\tilde H_M := \sum_{\tilde\imath=0}^{M-1}
\tilde h_{\tilde\imath}$
has local terms $\tilde h_{\tilde\imath}$ with  support
on two consecutive $p$ blocks and the $q$ block inbetween (cf.\
\cref{fig:lemma-blocking}a), which are constructed as a linear combination of terms
$h_i$ in its support with weights $\lambda_\omega$, $\omega=1,\dots,2p+q-(n-1)$,
as
\begin{equation}
\label{eq:blockedham-obc}
\tilde h_{\tilde\imath} :=
\sum_{\omega=1}^{2p+q-(n-1)} \lambda_\omega h_{i_0+\omega}\ ,\qquad
\tilde\imath = 0,\dots,M-1\ ,
\quad i_0=\tilde\imath(p+q)\ ,
\end{equation}
where 
$\lambda_\omega>0$, 
\begin{equation}
\label{eq:lemma-obc-lambdai}
\lambda_\omega+\lambda_{\omega'}=1\mbox{\ for\ }
\omega=1,\dots,p-(n-1),\ \omega'=\omega+(p+q)\ ,
\end{equation}
and
$\lambda_\omega=1$ otherwise.  

Then, $\tilde H_M$ and $H_N$ have the same ground space, and
$\Delta(H_N)\ge\Delta(\tilde H_M)$.  Moreover, $\tilde H_M$ is frustration
free, and the ground space of its terms $\tilde h_{\tilde\imath}$ is
parametrized by the MPS from which the original parent Hamiltonian has been
constructed.
\end{lemma}
The entire construction of $\tilde h_{\tilde\imath}$ and $\tilde H_M$ is
illustrated in \cref{fig:lemma-blocking}.  The simplest (``canonical'') choice for the
$\lambda_\omega$ in \cref{eq:lemma-obc-lambdai} is $\lambda_\omega=\tfrac12$.
Furthermore, the Lemma still holds when relaxing
$\lambda_\omega+\lambda_{\omega'}\le 1$ and $\lambda_\omega\le 1$, respectively, but
equality will always result in the best bounds on the gap.  Specifically,
the choice $q=0$ and $\lambda_\omega=\tfrac12$ for \emph{all} $\omega$ is the choice made
by FNW and
Nachtergaele~\cite{fannes_finitely_1992,nachtergaele_spectral_1996}.

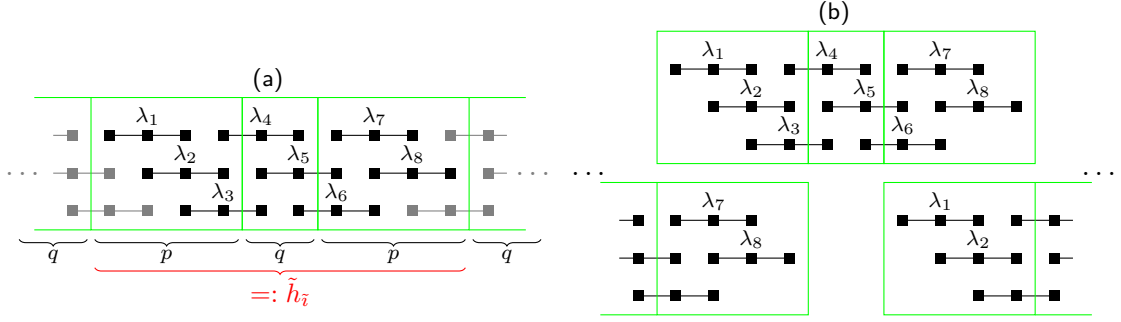
\begin{figure}[tb]
    \centering
    \hspace*{\fill}
\begin{subfigure}[b]{0.48\textwidth}
    \caption{}
    \centering
    \begin{tikzpicture}
        \foreach \x in {0,...,8}
            \node[rectangle,fill,minimum size=4pt,inner sep=0pt] at (0.5*\x,0) {};
        \foreach \x in {1,...,9}
            \node[rectangle,fill,minimum size=4pt,inner sep=0pt] at (0.5*\x,-0.5) {};
        \foreach \x in {2,...,7}
            \node[rectangle,fill,minimum size=4pt,inner sep=0pt] at (0.5*\x,-1) {};
        \foreach \x in {0,...,2}
           { \draw (1.5*\x,0) -- (1.5*\x+1,0);
            \draw (1.5*\x+0.5,-0.5) -- (1.5*\x+1.5,-0.5); }
        \draw (1,-1) -- (2,-1);
        \draw (2.5,-1) -- (3.5,-1);
        \foreach \y in {0,...,2}
            \foreach \x in {\y,...,2}
                \node[rectangle,fill=gray,minimum size=4pt,inner sep=0pt] at (-0.5*\x+0.5,0.5*\y-1) {};
        \node[rectangle,fill=gray,minimum size=4pt,inner sep=0pt] at (4.5,0) {};
        \node[rectangle,fill=gray,minimum size=4pt,inner sep=0pt] at (5,0) {};
        \node[rectangle,fill=gray,minimum size=4pt,inner sep=0pt] at (5,-0.5) {};
        \foreach \x in {0,...,2}
            \node[rectangle,fill=gray,minimum size=4pt,inner sep=0pt] at (0.5*\x+4,-1) {};
        \draw[gray] (-0.5,0) -- (-0.75,0);
        \draw[gray] (0,-0.5) -- (-0.75,-0.5);
        \draw[gray] (0.5,-1) -- (-0.5,-1);
        \draw[gray] (4.5,0) -- (5.25,0);
        \draw[gray] (5,-0.5) -- (5.25,-0.5);
        \draw[gray] (4,-1) -- (5,-1);
        \draw[gray] (-1.125,-0.5) node[] {$\dots$};
        \draw[gray] (5.625,-0.5) node[] {$\dots$};
        \draw[green] (-0.25,0.5) rectangle (1.75,-1.25);
        \draw[green] (1.75,-1.25) rectangle (2.75,0.5);
        \draw[green] (2.75,0.5) rectangle (4.75,-1.25);
        \draw[green] (-0.25,0.5) -- (-1,0.5);
        \draw[green] (-0.25,-1.25) -- (-1,-1.25);
        \draw[green] (4.75,-1.25) -- (5.5,-1.25);
        \draw[green] (4.75,0.5) -- (5.5,0.5);
        \draw (0.5,0) node[anchor=south] {\scriptsize $\lambda_1$};
        \draw (1,-0.5) node[anchor=south] {\scriptsize $\lambda_2$};
        \draw (1.5,-1) node[anchor=south] {\scriptsize $\lambda_3$};
        \draw (2,0) node[anchor=south] {\scriptsize $\lambda_4$};
        \draw (2.5,-0.5) node[anchor=south] {\scriptsize $\lambda_5$};
        \draw (3,-1) node[anchor=south] {\scriptsize $\lambda_6$};
        \draw (3.5,0) node[anchor=south] {\scriptsize $\lambda_7$};
        \draw (4,-0.5) node[anchor=south] {\scriptsize $\lambda_8$};
        \draw[decorate,decoration={brace,mirror}] (-0.2,-1.375) -- (1.7,-1.375);
        \draw[decorate,decoration={brace,mirror}] (1.8,-1.375) -- (2.7,-1.375);
        \draw[decorate,decoration={brace,mirror}] (2.8,-1.375) -- (4.7,-1.375);
        \draw[decorate,decoration={brace,mirror}] (-1.2,-1.375) -- (-0.3,-1.375);
        \draw[decorate,decoration={brace,mirror}] (4.8,-1.375) -- (5.7,-1.375);
        \draw (0.75,-1.375) node[anchor=north] {\scriptsize $p$};
        \draw (2.25,-1.375) node[anchor=north] {\scriptsize $q$};
        \draw (3.75,-1.375) node[anchor=north] {\scriptsize $p$};
        \draw (-0.75,-1.375) node[anchor=north] {\scriptsize $q$};
        \draw (5.25,-1.375) node[anchor=north] {\scriptsize $q$};
        \draw[red,decorate,decoration={brace,mirror}] (-0.2,-1.75) -- (4.7,-1.75);
        \draw[red] (2.25,-1.75) node[anchor=north] {\small $=: \tilde{h}_{\tilde\imath}$};
    \end{tikzpicture}
\end{subfigure}
\hspace*{\fill}
\begin{subfigure}[b]{0.48\textwidth}
    \caption{}
    \centering
    \begin{tikzpicture}
        \foreach \x in {0,...,8}
            \node[rectangle,fill,minimum size=4pt,inner sep=0pt] at (0.5*\x,0) {};
        \foreach \x in {1,...,9}
            \node[rectangle,fill,minimum size=4pt,inner sep=0pt] at (0.5*\x,-0.5) {};
        \foreach \x in {2,...,7}
            \node[rectangle,fill,minimum size=4pt,inner sep=0pt] at (0.5*\x,-1) {};
        \foreach \x in {0,...,2}
           { \draw (1.5*\x,0) -- (1.5*\x+1,0);
            \draw (1.5*\x+0.5,-0.5) -- (1.5*\x+1.5,-0.5); }
        \foreach \y in {0,...,2}
            \foreach \x in {\y,...,4}
                \node[rectangle,fill,minimum size=4pt,inner sep=0pt] at (3+0.5*\x,-0.5*\y-2) {};
        \foreach \x in {0,...,3}
            \node[rectangle,fill,minimum size=4pt,inner sep=0pt] at (0.5*\x-0.5,-2) {};
        \foreach \x in {0,...,4}
            \node[rectangle,fill,minimum size=4pt,inner sep=0pt] at (0.5*\x-0.5,-2.5) {};
        \foreach \x in {0,...,2}
            \node[rectangle,fill,minimum size=4pt,inner sep=0pt] at (0.5*\x-0.5,-3) {};
        \draw (1,-1) -- (2,-1);
        \draw (2.5,-1) -- (3.5,-1);
        \draw (0,-2) -- (1,-2);
        \draw (0.5,-2.5) -- (1.5,-2.5);
        \draw (-0.5,-3) -- (0.5,-3);
        \draw (-0.75,-2) -- (-0.5,-2);
        \draw (-0.75,-2.5) -- (0,-2.5);
        \draw (3,-2) -- (4,-2);
        \draw (4.5,-2) -- (5.25,-2);
        \draw (3.5,-2.5) -- (4.5,-2.5);
        \draw (5,-2.5) -- (5.25,-2.5);
        \draw (4,-3) -- (5,-3);
        \draw[] (-1.125,-1.375) node[] {$\dots$};
        \draw[] (5.625,-1.375) node[] {$\dots$};
        \draw[green] (-0.25,0.5) rectangle (1.75,-1.25);
        \draw[green] (-0.25,-1.5) rectangle (1.75,-3.25);
        \draw[green] (1.75,-1.25) rectangle (2.75,0.5);
        \draw[green] (2.75,0.5) rectangle (4.75,-1.25);
        \draw[green] (2.75,-1.5) rectangle (4.75,-3.25);
        \draw[green] (-0.25,-1.5) -- (-1,-1.5);
        \draw[green] (-0.25,-3.25) -- (-1,-3.25);
        \draw[green] (4.75,-1.5) -- (5.5,-1.5);
        \draw[green] (4.75,-3.25) -- (5.5,-3.25);
        \draw (0.5,0) node[anchor=south] {\scriptsize $\lambda_1$};
        \draw (0.5,-2) node[anchor=south] {\scriptsize $\lambda_7$};
        \draw (1,-0.5) node[anchor=south] {\scriptsize $\lambda_2$};
        \draw (1,-2.5) node[anchor=south] {\scriptsize $\lambda_8$};
        \draw (1.5,-1) node[anchor=south] {\scriptsize $\lambda_3$};
        \draw (2,0) node[anchor=south] {\scriptsize $\lambda_4$};
        \draw (2.5,-0.5) node[anchor=south] {\scriptsize $\lambda_5$};
        \draw (3,-1) node[anchor=south] {\scriptsize $\lambda_6$};
        \draw (3.5,0) node[anchor=south] {\scriptsize $\lambda_7$};
        \draw (3.5,-2) node[anchor=south] {\scriptsize $\lambda_1$};
        \draw (4,-0.5) node[anchor=south] {\scriptsize $\lambda_8$};
        \draw (4,-2.5) node[anchor=south] {\scriptsize $\lambda_2$};
    \end{tikzpicture}
\end{subfigure}
\hspace*{\fill}
\caption{
\label{fig:lemma-blocking}
Illustration of~\cref{lem:block-obc} and \cref{lem:block-pbc} for
$n=3,\ p=4,\ q=2$. \textbf{(a)} The spin chain is blocked into
alternating blocks of $p$ and $q$ sites, respectively. The (translation
invariant) blocked Hamiltonian $\tilde h_{\tilde\imath}$ acts on two
consecutive $p$-blocks and the $q$-block inbetween, and contains all original
$n$-body Hamiltonian terms $h_i$ of $H_N=\sum h_i$ which lie inside that region
with a strictly positive weight $\lambda_\omega$ (with $\omega$ the offset).
\textbf{(b)} \cref{lem:block-obc}: Summation of all blocked terms, $\tilde
H_M=\tilde h_{\tilde\imath}$, must give the original Hamiltonian $H_N$ (or
something smaller). This implies that the weights $\lambda_\omega$ of the $h_i$
which are contained in overlapping $\tilde h_{\tilde\imath}$ must sum up to $1$
(or, again, something smaller). In the concrete example in the figure, this means
$\lambda_1+\lambda_7=1$, $\lambda_2+\lambda_8=1$, and
$\lambda_3=\dots=\lambda_6=1$, cf.\ \cref{eq:lemma-obc-lambdai}.} 
\end{figure}

\emph{Proof.---}The proof is straightforward and can be directly read off
\cref{fig:lemma-blocking}b: We can expand $\tilde H_M = \sum \tilde h_{\tilde\imath}
= \sum_i \gamma_i h_i$, and using the conditions on $\lambda_\omega$, in
particular \cref{eq:blockedham-obc}, we find that for PBC, $\gamma_i=1$ for all
$i$, while for OBC, $0<\gamma_i\le1$ (where $\gamma_i=1$ for all terms which don't
act on the $p$ sites at either boundary). $\gamma_i>0$ (which follows from
$\lambda_\omega>0$) implies that $\tilde H_M$ and $H_N$ have the same ground space.
$\gamma_i\le 1$ implies that $\tilde H_M\le H_N$; since they share the ground
space, this implies that 
\begin{equation}
\Delta(H_N)\ge\Delta(\tilde H_M)\ .
\end{equation}
Note that for PBC, we obtain the stronger statement that
$H_N=\tilde H_M$, and thus in particular their gaps are the same.
Moreover, since in \cref{eq:blockedham-obc}, all $\lambda_\omega>0$, $\tilde
h_{\tilde\imath}$ has the same ground space as $\sum_{i=1}^{2p+q-(n-1)} h_i\equiv
H_{2p+q}$, which is precisely the space spanned by the MPS.
\hspace*{\fill}$\square$

If we restrict to PBC, \cref{lem:block-obc} can be strengthened by summing
over all translations within a unit cell of $p+q$ unblocked sites.
\begin{lemma}
\label{lem:block-pbc}
Consider the setting of \cref{lem:block-obc} with PBC, with a blocked
Hamiltonian $\tilde H_M$, where the condition on the $\lambda_\omega>0$ is
changed to
\begin{equation}
\label{eq:lemma-pbc-lambdai}
\sum_{\omega=1}^{2p+q-(n-1)} \lambda_\omega = p+q\ .
\end{equation}

Then, $\tilde H_M$ and $H_N$ have the same ground space, and
\begin{equation}
\Delta(H_N)\ge\Delta(\tilde H_M)\ .
\end{equation}
Moreover, $\tilde H_M$ is frustration
free, and the ground space of its terms $\tilde h_{\tilde\imath}$ is
parametrized by the MPS from which $H$ has been obtained.
\end{lemma}
Just as for \cref{lem:block-obc}, the condition \cref{eq:lemma-pbc-lambdai} can
be relaxed to $\sum \lambda_\omega\le p+q$, but equality will always give the best
bounds on the gap. Clearly, a canonical choice for $\lambda_\omega$ is $\lambda_\omega = 
\tfrac{p+q}{2p+q-(n-1)}$,
in which case $\tilde h_{\tilde\imath} = 
\tfrac{p+q}{2p+q-(n-1)} H_{2p+q}$.

\emph{Proof.---}Consider $\hat H:=\tfrac{1}{p+q}\sum_{j=0}^{p+q-1}\tau_j(\tilde
H_M)$, with $\tau_j$ the operator which translates by $j$ sites on the original
(unblocked) lattice: One can easily verify that in $\hat H$, each term $h_i$
appears with the total weight
$\tfrac{1}{p+q}\sum_{\omega=1}^{2p+q-(n-1)}\lambda_\omega=1$, since
$\sum_j\tau_j(\tilde H_M)$ adds all translations \emph{within} a block of size
$p+q$, while $\tilde H_M$ is a sum over all translations by multiples of $p+q$ (recall that we assume PBC).
Moreover, since $\lambda_\omega>0$, $\tilde H_M$ and $H_N$ have the same ground
space; and since this ground space is translation invariant, $\hat H$ and $H_N$
have the same ground space. Thus, 
\begin{equation}
    \Delta(H_N) = \Delta(\hat H) \ge 
    \frac{1}{p+q} \sum_{j=0}^{p+q-1}\Delta(\tau_j(\tilde H_M))
    = \Delta(\tilde H_M)\ .
\end{equation}
\hspace*{\fill}$\square$

\subsection{The spectral gap of the blocked Hamiltonian}

Having established the relation between the spectral gap of the original and a
blocked Hamiltonian $\tilde H_M$, let us now turn towards the second step:
Bounding the gap of $\tilde H_M$ independent of the system size (and thus in the
thermodynamic limit). Why did we require the initial blocking? First, the
blocked Hamiltonian is nearest neighbor, that is, each Hamiltonian term $\tilde
h_{\tilde\imath}$ only overlaps with the adjacent term to its left and right --
though the subsequent results can be generalized to not rely on it. Second, and
that is the important point, by increasing the block sizes $p$ and $q$ we obtain improved bounds on the gap, which in particular also allows to prove that every parent Hamiltonian is gapped, and the obtained bound approaches the true value of the gap.

This section consists of two parts: \cref{theorem:gap-from-gamma} relates the
spectral gap of $\tilde H_M$ in the thermodynamic limit to a local inequality
which depends on a parameter $\gamma_{p,q}$. \cref{theorem:gamma-exact} provides
a way to explicitly compute $\gamma_{p,q}$. Combined with either
\cref{lem:block-obc} or \cref{lem:block-pbc}, this allows to obtain explicit
lower bounds on the gap of the original Hamiltonian $H_N$.

\begin{theorem}
\label{theorem:gap-from-gamma}
Consider a blocked Hamiltonian $\tilde H_M=\sum \tilde h_{\tilde\imath}$
(for OBC, $M\ge1$, for PBC, $M\ge 3$),  with
alternating blocks of $p$ and $q$ sites, as obtained from
Lemma~\ref{lem:block-obc} or~\ref{lem:block-pbc}, such that each $\tilde
h_{\tilde\imath}$ only overlaps with its directly neighboring terms.
Let $\Pi_{2p+q}$ denote the orthogonal projection onto $\mathcal S_{2p+q}=\ker \tilde h_{\tilde\imath} = \ker H_{2p+q}$,  and $\Pi_{2p+q}^{\perp} = (\mathbb{1}-\Pi_{2p+q})$. 
Let $Q:=(\Pi_{2p+q}^{\perp} \otimes \mathbb{1}_{q+p})$ and $R:= (\mathbb{1}_{p+q}
\otimes \Pi_{2p+q}^{\perp})$. Then, for any $\gamma_{p,q}\ge0$ such that
\begin{equation}
\label{eq:PQ-QP-condition}
QR+RQ+\gamma_{p,q}(Q+R) \ge 0\ ,
\end{equation}
\begin{equation}
\label{gapeq}
    \varepsilon := 
    (1-2\gamma_{p,q})\,
    \Delta(\tilde h_{\tilde\imath})
\end{equation}
is a uniform lower bound on the spectral gap of $\tilde H_M$, 
$\Delta(\tilde H_M)\ge \varepsilon$,
both for OBC and for PBC. 
\end{theorem}
Some remarks: (i) By virtue of Lemmas~\ref{lem:block-obc} and
\ref{lem:block-pbc}, this also provides a lower bound on the spectral gap of the
unblocked Hamiltonian; (ii) In \cref{gapeq}, a lower bound on $\Delta(\tilde
h_{\tilde\imath})$ and an upper bound on $\gamma_{p,q}$ (i.e., a non-optimal $\gamma_{p,q}$ in \cref{eq:PQ-QP-condition}) are sufficient; (iii)
$\Delta(\tilde h_{\tilde\imath})>0$ by definition of the gap. Variants of this
theorem, and its proof, date back to the work of
FNW~\cite{fannes_finitely_1992}.

\emph{Proof.---}We will use the shorthand $P:=\Pi_{2p+q}^\perp$. 
Define  $\sum P_{\tilde\imath}=:\tilde K_M$, where
$P_{\tilde\imath}$ acts starting at the blocked site $\tilde\imath$. 
Then,
\begin{equation}
\begin{aligned}
(\tilde K_M)^2-(1-2\gamma_{p,q})\tilde K_M &= 
\sum_{\tilde\imath}  P_{\tilde\imath}^2
    +\sum_{|\tilde\imath-\tilde\jmath|=1}  
    P_{\tilde\imath} P_{\tilde\jmath} 
 +\sum_{|\tilde\imath-\tilde\jmath|>1} P_{\tilde\imath} P_{\tilde\jmath}
- \sum_{\tilde\imath} (1-2\gamma_{p,q}) P_{\tilde\imath}
\\
&\stackrel{(\star)}{\ge} \sum_{\tilde\imath} 2\gamma_{p,q}  P_{\tilde\imath}
    +\sum_{|\tilde\imath-\tilde\jmath|=1}  
    P_{\tilde\imath} P_{\tilde\jmath} 
\\
&\stackrel{(\dagger)}{\ge}
    \sum_{\tilde\imath=\tilde\jmath+1}  
    P_{\tilde\imath} P_{\tilde\jmath} 
    +P_{\tilde\jmath} P_{\tilde\imath} 
    + \gamma_{p,q}(P_{\tilde\imath}+P_{\tilde\jmath})
\stackrel{\eqref{eq:PQ-QP-condition}}{\ge} 0\ ,
\end{aligned}
\end{equation}
where in $(\star)$, we have used $P_{\tilde\imath}^2=P_{\tilde\imath}$ and 
that for $|\tilde\imath-\tilde\jmath|>1$, $P_{\tilde\imath}P_{\tilde\jmath}\ge0$  as the projectors don't overlap, and $(\dagger)$ is an equality for PBC, while for OBC, we drop a term $\gamma_{p,q}P_{\tilde\imath}\ge0$ for $\tilde\imath$ at the left and right boundary.  
This is equivalent to the statement that $\tilde K_M$ has no eigenvalues between
$0$ and $(1-2\gamma_{p,q})$, i.e.\ $\Delta(\tilde K_M)\ge 1-2\gamma_{p,q}$.
Since $\tilde h_{\tilde\imath}\ge \Delta(\tilde h_{\tilde \imath}) \Pi^\perp_{2p+q}$,
$\tilde H_M = \sum \tilde h_{\tilde\imath} \ge \Delta(\tilde h_{\tilde\imath})\tilde
K_M$, and thus $\Delta(\tilde H_M)\ge \Delta(\tilde h_{\tilde\imath}) \Delta(\tilde K_M)
\ge \varepsilon$.
\hspace*{\fill}$\square$

Note that the condition \cref{eq:PQ-QP-condition} is monotonous in $\gamma_{p,q}$, i.e.,\ if it
holds for some $\gamma_{p,q}$, it also holds for any larger value, and that it always holds for $\gamma_{p,q}=1$ (where the l.h.s.\ becomes $(Q+R)^2$). Since the gap bound becomes non-trivial
only for $\gamma_{p,q}<1/2$ and improves with decreasing $\gamma_{p,q}$, we are
generally looking for the optimal, i.e., smallest $\gamma_{p,q}$. 

Naively, the effort for computing the optimal $\gamma_{p,q}$ will scale with the
underlying Hilbert space dimension, since we have to check positivity of
\cref{eq:PQ-QP-condition}, and thus exponentially with $p$ and $q$.
The key result which we present in the following theorem is that the optimal
$\gamma_{p,q}$ can be computed with an effort which scales linearly with $p$ and
$q$, and polynomially in the bond dimension $D$ of the MPS.

\begin{theorem}[Optimal $\bm{\gamma_{p,q}}$]
\label{theorem:gamma-exact}
Consider an MPS with tensor $A$, and a blocked Hamiltonian $\tilde H_M$ whose
ground space on any number of blocks is spanned by this MPS; this is in
particular the case if $A$ is normal or block-normal (cf.~\cref{sec:parenthams})
and $\tilde H_M$ has been obtained from blocking a parent Hamiltonian of the MPS.
Define  the $D^4 \times D^4$ matrix 
\begin{equation}
\label{eq:thm-gammaexact-defM}
\Xi:=\ 
    \begin{tikzpicture}[baseline]={([yshift=-.5ex]current bounding box.center)}]
            \draw (0,0) rectangle (3.125,0.5);
            \draw (1.5625,0.25) node[]{\small $T_{2p+q}^{-1}$};
            \draw (0.25,0.5) -- (0.25,0.75) -- (0.5,0.75);
            \draw (0.5,0.625) rectangle (1.375,1.125);
            \draw (0.9375,0.875) node[]{\small $\mathcal{T}^{p+q}$};
            \draw (0.5,1) -- (0.25,1) -- (0.25,1.5);
            \draw (1.375,1) -- (1.625,1) -- (1.625,1.5);
            \draw (1.375,0.75) -- (2.125,0.75);
            \draw (2.125,0.625) rectangle (2.625,1.125);
            \draw (2.375,0.875) node[]{\small $\mathcal{T}^{p}$};
            \draw (2.875,0.5) -- (2.875,0.75) -- (2.625,0.75);
            \draw (2.125,1) -- (1.875,1) -- (1.875,1.5);
            \draw (2.625,1) -- (3.375,1);
            \draw (0.25,0) -- (0.25,-1.5);
            \draw (2.875,0) -- (2.875,-0.25) -- (2.625,-0.25);
            \draw (2.625,-0.125) rectangle (2.125,-0.625);
            \draw (2.375,-0.375) node[]{\small $\mathcal{T}^p$};
            \draw (2.125,-0.25) -- (1.0625,-0.25) -- (1.0625,-1.5);
            \draw (2.125,-0.5) -- (1.875,-0.5) -- (1.875,-0.75);
            \draw (2.625,-0.5) -- (3.375,-0.5);
            \draw (1.625,-0.75) rectangle(4.75,-1.25);
            \draw (3.1875,-1) node[]{\small $T_{2p+q}^{-1}$};
            \draw (1.875,-1.25) -- (1.875,-1.5);
            \draw (4.5,-1.25) -- (4.5,-1.5);
            \draw (3.375,-0.625) rectangle (4.25,1.125);
            \draw (3.8125,0.25) node[]{\small $\mathcal{T}^{p+q}$};
            \draw (4.25,-0.5) -- (4.5,-0.5) -- (4.5,-0.75);
            \draw (4.25,1) -- (4.5,1) -- (4.5,1.5);
        \end{tikzpicture} \ ,
\end{equation}
where $T_{2p+q}^{-1}$ is the pseudoinverse (i.e., the
inverse on the support) of $T_{2p+q}$. Then, $\Xi$ has real non-negative spectrum, its leading
eigenvalue is $\lambda_1(\Xi)=1$, and its degeneracy equals the ground state
degeneracy of $\tilde h_{\tilde\imath}+\tilde h_{\tilde\imath+1}$. 

The optimal $\gamma_{p,q}$ in \cref{theorem:gap-from-gamma}, \cref{eq:PQ-QP-condition},  is
then $\gamma_{p,q}=\sqrt{\lambda_2(\Xi)}$, where $\lambda_2(\Xi)$ is the largest
eigenvalue of $\Xi$ distinct from $\lambda_1(\Xi)=1$.
\end{theorem}
A key point of this theorem is that it reduces \cref{eq:PQ-QP-condition} -- which involves computing properties of an operator on a space whose dimension depends (exponentially) on $p$ and $q$ -- to the analysis of an operator $\Xi$ which resides on a space of fixed dimension, independent of $p$ and $q$.

\noindent\emph{Proof.---}We want to find the optimal $\gamma_{p,q}$ such that \cref{eq:PQ-QP-condition},
\begin{equation}
QR+RQ+\gamma_{p,q}(Q+R) \ge 0\ ,
\tag{\ref{eq:PQ-QP-condition}}
\end{equation}
holds, where $Q$ and $R$ are projectors onto the orthogonal complement of the span of the MPS $A$ on $p+q+p$
sites, overlapping on the central $p$-site block. To this end, we use two facts:
Jordan's Lemma, and  the cyclicity of the spectrum. Jordan's Lemma states that
for two Hermitian projectors $Q$ and $R$, there exists a basis in which they are
jointly block-diagonal in $1\times 1$ and $2\times 2$ blocks: $Q=\bigoplus_i
Q^{(i)}$, $R=\bigoplus_i R^{(i)}$. Then, $1\times 1$ blocks are necessarily all
either $0$ or $1$, while the $2\times 2$ blocks can be further rotated to orthogonal
projectors onto $(1,0)^T$ and $(\cos\theta_i,\sin\theta_i)^T$ (the
$\theta_i$ are the principal angles between the subspaces onto which $Q$ and $R$
project). 

Let $\hat \Xi:=Q^\perp R^\perp$. Then, $\mathrm{spec}(\hat \Xi)=\mathrm{spec}(\sqrt{Q^\perp}R^\perp\sqrt{Q^\perp})\ge 0$, $\|\hat \Xi\|\le
\|Q^\perp\|,\|R^\perp\|=1$, and $\ker(\tilde h_{\tilde \imath}+\tilde h_{\tilde
\imath+1})=\mathrm{Im}\,Q^\perp\cap \mathrm{Im}\,R^\perp$ is precisely the eigenspace with
leading eigenvalue $+1$ of $\hat \Xi$. Working in the basis which
block-diagonalizes $Q$ and $R$,  it is now straightforward to check that the
optimal $\gamma_{p,q}$ in \cref{eq:PQ-QP-condition} equals $\sqrt{\lambda_2(\hat
\Xi)}$ -- one can just check this relation for each $2 \times 2$ block separately, and observe
that the block with the largest $\gamma_{p,q}$ and the largest $\lambda_2(Q^{(i)} R^{(i)})$
coincide. For completeness, we give a detailed derivation in \cref{appendix:exact-gamma-details}.

$\Pi_{2p+q}$ 
is the projector onto $\mathrm{Im}\,\mathcal P$, with
$\mathcal P:X\mapsto\ket{\Psi_{2p+q}(A,X)}$, cf.\ \cref{eq:Psi-n-M-X}, the map from
virtual to physical indices on $2p+q$ sites, and thus $\Pi_{2p+q} =
{\mathcal P}({\mathcal P}^\dagger {\mathcal P})^{-1}{\mathcal P}^\dagger$, 
where $(\,\cdot\,)^{-1}$ denotes the pseudoinverse. 
In terms of tensor diagrams, 
\begin{equation}
   \Pi_{2p+q}  = \begin{tikzpicture}[baseline={([yshift=-.5ex]current bounding box.center)}]
        \draw (-0.75,0.25) rectangle (0.75,0.75);
        \draw (-0.5,0.25) -- (-0.5,0);
        \draw (0.5,0.25) -- (0.5,0);
        \draw (-0.5,0.75) -- (-0.5,1);
        \draw (0.5,0.75) -- (0.5,1);
        \draw (-0.5,1) -- (0.5,1);
        \draw (-0.5,0) -- (0.5,0);
        \draw (0,0.5) node[]{\small $T_{2p+q}^{-1}$};
        \draw (-0.375,1) -- (-0.375,1.25);
        \draw (0.375,1) -- (0.375,1.25);
        \draw (-0.375,0) -- (-0.375,-0.25);
        \draw (0.375,0) -- (0.375,-0.25);
        \draw[fill = white] (-0.375,0) circle (2pt);
        \draw[fill = white] (-0.375,1) circle (2pt);
        \draw[fill = white] (0.375,0) circle (2pt);
        \draw[fill = white] (0.375,1) circle (2pt);
        \fill (-0.375,0) circle (1.25pt);
        \fill (0.375,0) circle (0.75pt);
        \fill (-0.375,1) circle (1.25pt);
        \fill (0.375,1) circle (0.75pt);
        \end{tikzpicture} =
        \begin{tikzpicture}[baseline={([yshift=-.5ex]current bounding box.center)}]
        \draw (-0.75,0.25) rectangle (0.75,0.75);
        \draw (-0.5,0.25) -- (-0.5,0);
        \draw (0.5,0.25) -- (0.5,0);
        \draw (-0.5,0.75) -- (-0.5,1);
        \draw (0.5,0.75) -- (0.5,1);
        \draw (-0.5,1) -- (0.5,1);
        \draw (-0.5,0) -- (0.5,0);
        \draw (0,0.5) node[]{\small $T_{2p+q}^{-1}$};
        \draw (-0.375,1) -- (-0.375,1.25);
        \draw (0.375,1) -- (0.375,1.25);
        \draw (-0.375,0) -- (-0.375,-0.25);
        \draw (0.375,0) -- (0.375,-0.25);
        \draw[fill = white] (-0.375,0) circle (2pt);
        \draw[fill = white] (-0.375,1) circle (2pt);
        \draw[fill = white] (0.375,0) circle (2pt);
        \draw[fill = white] (0.375,1) circle (2pt);
        \fill (-0.375,1) circle (0.75pt);
        \fill (0.375,1) circle (1.25pt);
        \fill (-0.375,0) circle (0.75pt);
        \fill (0.375,0) circle (1.25pt);
        \end{tikzpicture} \ ,
\end{equation}
where the dot symbols denote blocked tensors on $p+q$ and $p$ sites, respectively:
\begin{equation}\label{eq:blocked_mps}
    \begin{tikzpicture}[baseline={([yshift=-.5ex]current bounding box.center)}]
        \draw (0,0) -- (0,0.4);
        \draw (-0.4,0) -- (0.4,0);
        \draw[fill=white] (0,0) circle (2pt);
        \fill (0,0) circle (1.25pt);
    \end{tikzpicture} \coloneq 
    \begin{tikzpicture}[baseline={([yshift=.9ex]current bounding box.center)}]
        \draw (-0.4,0) -- (2.4,0);
        \foreach \x in {0,0.8,2}
            {\draw (\x,0) -- (\x,0.4);
            \fill (\x,0) circle (2pt);};
        \draw (1.4,0.2) node[] {$\dots$};
        \draw[decorate,decoration={brace,mirror}] (-0.1,-0.1) -- (2.1,-0.1);
        \draw (1,-0.3) node[] {\scriptsize $p+q$~\text{times}};
    \end{tikzpicture}\ , \quad \begin{tikzpicture}[baseline={([yshift=-.5ex]current bounding box.center)}]
        \draw (0,0) -- (0,0.4);
        \draw (-0.4,0) -- (0.4,0);
        \draw[fill=white] (0,0) circle (2pt);
        \fill (0,0) circle (0.75pt);
    \end{tikzpicture} \coloneq 
    \begin{tikzpicture}[baseline={([yshift=.9ex]current bounding box.center)}]
        \draw (-0.4,0) -- (2.4,0);
        \foreach \x in {0,0.8,2}
            {\draw (\x,0) -- (\x,0.4);
            \fill (\x,0) circle (2pt);};
        \draw (1.4,0.2) node[] {$\dots$};
        \draw[decorate,decoration={brace,mirror}] (-0.1,-0.1) -- (2.1,-0.1);
        \draw (1,-0.3) node[] {\scriptsize $p$~\text{times}};
    \end{tikzpicture}\ .
\end{equation}
In order to compute $\gamma_{p,q}^2=\lambda_2(\hat \Xi)$, $\hat\Xi=
(\Pi_{2p+q}\otimes \openone_{p+q}) (\openone_{p+q}\otimes
\Pi_{2p+q})$, we can now use that the non-zero part of the spectrum is
cyclic, and find
\begin{equation}
\gamma_{p,q}^2=
 \lambda_2 \left( \begin{tikzpicture}[baseline={([yshift=-.5ex]current bounding box.center)}]
        \draw (-0.75,0.25) -- (0.75,0.25);
        \draw (-0.75,0.25) -- (-0.75,0.75);
        \draw (-0.75,0.75) -- (0.75,0.75);
        \draw (0.75,0.75) -- (0.75,0.25);
        \draw (-0.5,0.25) -- (-0.5,0);
        \draw (0.5,0.25) -- (0.5,0);
        \draw (-0.5,0.75) -- (-0.5,1);
        \draw (0.5,0.75) -- (0.5,1);
        \draw (-0.5,1) -- (0.5,1);
        \draw (-0.5,0) -- (0.5,0);
        \draw (0,0.5) node[]{\small $T_{2p+q}^{-1}$};
        \draw (-0.375,1) -- (-0.375,1.25);
        \draw (0.375,1) -- (0.375,2.5);
        \draw (-0.375,0) -- (-0.375,-0.25);
        \draw (0.375,0) -- (0.375,-0.25);
        \draw (-1.25,1.25) -- (-0.25,1.25);
        \draw (-1.25,1.25) -- (-1.25,1.5);
        \draw (-0.25,1.25) -- (-0.25,1.5);
        \draw (-1.5,1.5) -- (0,1.5);
        \draw (-1.125,1.25) -- (-1.125,-0.25);
        \draw (-1.5,1.5) -- (-1.5,2);
        \draw (0,1.5) -- (0,2);
        \draw (-1.5,2) -- (0,2);
        \draw (-0.75,1.75) node[]{\small $T_{2p+q}^{-1}$};
        \draw (-1.25,2) -- (-1.25,2.25);
        \draw (-0.25,2) -- (-0.25,2.25);
        \draw (-1.25,2.25) -- (-0.25,2.25);
        \draw (-1.125,2.25) -- (-1.125,2.5);
        \draw (-0.375,2.25) -- (-0.375,2.5);
        \draw[blue, dashed, very thick] (0.625,0.125) -- (-0.875,0.125);
        \draw[blue, dashed, very thick] (-0.875,0.125) -- (-0.875,1.375);
        \draw[blue, dashed, very thick] (-0.875,1.375) -- (-1.375,1.375);
        \draw[blue, dashed, very thick] (0.625,0.125) -- (0.625,-0.375);
        \draw[blue, dashed, very thick]
        (0.625,-0.375) -- (-1.375,-0.375);
        \draw[blue, dashed, very thick]
        (-1.375,-0.375) -- (-1.375,1.375);
        \draw[fill=white] (-0.375,0) circle (2pt);
        \draw[fill=white] (-0.375,1) circle (2pt);
        \draw[fill=white] (-0.375,1.25) circle (2pt);
        \draw[fill=white] (-0.375,2.25) circle (2pt);
        \draw[fill=white] (0.375,0) circle (2pt);
        \draw[fill=white] (0.375,1) circle (2pt);
        \draw[fill=white] (-1.125,1.25) circle (2pt);
        \draw[fill=white] (-1.125,2.25) circle (2pt);
        \fill (-0.375,1) circle (0.75pt);
        \fill (0.375,1) circle (1.25pt);
        \fill (-0.375,0) circle (0.75pt);
        \fill (0.375,0) circle (1.25pt);
        \fill (-0.375,1.25) circle (0.75pt);
        \fill (-1.125,1.25) circle (1.25pt);
        \fill (-1.125,2.25) circle (1.25pt);
        \fill (-0.375,2.25) circle (0.75pt);
        \end{tikzpicture}\right) = \lambda_2\left( \begin{tikzpicture}[baseline={([yshift=-.5ex]current bounding box.center)}]
        \draw (-0.75,0.25) -- (0.75,0.25);
        \draw (-0.75,0.25) -- (-0.75,0.75);
        \draw (-0.75,0.75) -- (0.75,0.75);
        \draw (0.75,0.75) -- (0.75,0.25);
        \draw (-0.5,0.25) -- (-0.5,0);
        \draw (0.5,0.25) -- (0.5,0);
        \draw (-0.5,0.75) -- (-0.5,1);
        \draw (0.5,0.75) -- (0.5,1);
        \draw (-0.5,1) -- (0.5,1);
        \draw (-0.5,2.5) -- (0.5,2.5);
        \draw (0,0.5) node[]{\small $T_{2p+q}^{-1}$};
        \draw (-0.375,1) -- (-0.375,1.25);
        \draw (0.375,1) -- (0.375,2.5);
        \draw (-0.5,2.5) -- (-0.5,2.75);
        \draw (0.5,2.5) -- (0.5,2.75);
        \draw (-1,1.25) -- (-0.25,1.25);
        \draw (-1,1.25) -- (-1,0);
        \draw (-1.25,0) -- (-1.25,1.5);
        \draw (-0.25,1.25) -- (-0.25,1.5);
        \draw (-1.5,1.5) -- (0,1.5);
        \draw (-1.25,2.5) -- (-1,2.5);
        \draw (-1.25,2.5) -- (-1.25,2.75);
        \draw (-1,2.5) -- (-1,2.75);
        \draw (-1.5,1.5) -- (-1.5,2);
        \draw (0,1.5) -- (0,2);
        \draw (-1.5,2) -- (0,2);
        \draw (-0.75,1.75) node[]{\small $T_{2p+q}^{-1}$};
        \draw (-1.25,2) -- (-1.25,2.25);
        \draw (-0.25,2) -- (-0.25,2.25);
        \draw (-1.25,2.25) -- (-0.25,2.25);
        \draw (-1.125,2.25) -- (-1.125,2.5);
        \draw (-0.375,2.25) -- (-0.375,2.5);
        \draw[dashed, blue, very thick] (-1.375,2.375) -- (0.625,2.375);
        \draw[dashed, blue, very thick] (-1.375,2.375) -- (-1.375,2.875);
        \draw[dashed, blue, very thick] (-1.375,2.875) -- (0.625,2.875);
        \draw[dashed, blue, very thick] (0.625,2.875) -- (0.625,2.375);
        \draw[fill=white] (-0.375,2.5) circle (2pt);
        \draw[fill=white] (-0.375,1.25) circle (2pt);
        \draw[fill=white] (-0.375,1) circle (2pt);
        \draw[fill=white] (-0.375,2.25) circle (2pt);
        \draw[fill=white] (0.375,2.5) circle (2pt);
        \draw[fill=white] (0.375,1) circle (2pt);
        \draw[fill=white] (-1.125,2.25) circle (2pt);
        \draw[fill=white] (-1.125,2.5) circle (2pt);
        \fill (-0.375,2.5) circle (0.75pt);
        \fill (0.375,2.5) circle (1.25pt);
        \fill (-0.375,1.25) circle (0.75pt);
        \fill (-0.375,1) circle (0.75pt);
        \fill (0.375,1) circle (1.25pt);
        \fill (-1.125,2.25) circle (1.25pt);
        \fill (-1.125,2.5) circle (1.25pt);
        \fill (-0.375,2.25) circle (0.75pt);
        \end{tikzpicture}\right)\ ,
\end{equation}
where we have used the cyclicity of the spectrum to move the operator enclosed
by blue dashed lines from the bottom to the top. 
Using that
\begin{equation}
    \begin{tikzpicture}[baseline={([yshift=-.5ex]current bounding box.center)}]
        \draw (0,0) -- (0,0.8);
        \draw (-0.4,0) -- (0.4,0);
        \draw (-0.4,0.8) -- (0.4,0.8);
        \draw[fill=white] (0,0) circle (2pt);
        \draw[fill=white] (0,0.8) circle (2pt);
        \fill (0,0) circle (1.25pt);
        \fill (0,0.8) circle (1.25pt);
    \end{tikzpicture} = \begin{tikzpicture}[baseline={([yshift=-.5ex]current bounding box.center)}]
        \draw (-0.4375,-0.2) rectangle (0.4375,1);
        \draw (-0.4375,0) -- (-0.6875,0);
        \draw (0.4375,0) -- (0.6875,0);
        \draw (-0.4375,0.8) -- (-0.6875,0.8);
        \draw (0.4375,0.8) -- (0.6875,0.8);
        \draw (0,0.4) node[]{\small $\mathcal{T}^{p+q}$};
    \end{tikzpicture}\ , \quad 
    \begin{tikzpicture}[baseline={([yshift=-.5ex]current bounding box.center)}]
        \draw (0,0) -- (0,0.8);
        \draw (-0.4,0) -- (0.4,0);
        \draw (-0.4,0.8) -- (0.4,0.8);
        \draw[fill=white] (0,0) circle (2pt);
        \draw[fill=white] (0,0.8) circle (2pt);
        \fill (0,0) circle (0.75pt);
        \fill (0,0.8) circle (0.75pt);
    \end{tikzpicture} = \begin{tikzpicture}[baseline={([yshift=-.5ex]current bounding box.center)}]
        \draw (-0.25,-0.2) -- (0.25,-0.2);
        \draw (-0.25,-0.2) -- (-0.25,1);
        \draw (-0.25,1) -- (0.25,1);
        \draw (0.25,1) -- (0.25,-0.2);
        \draw (-0.5,0) -- (-0.25,0);
        \draw (0.25,0) -- (0.5,0);
        \draw (-0.5,0.8) -- (-0.25,0.8);
        \draw (0.25,0.8) -- (0.5,0.8);
        \draw (0,0.4) node[]{\small $\mathcal{T}^p$};
        \end{tikzpicture}\ ,
\end{equation}
this completes the proof.
\hspace{\fill}$\square$

\vspace*{1.5em}

Let us now discuss consequences of our results, in particular of our main
\cref{theorem:gamma-exact}, both for the existence and for quantitative
bounds on the gap. 

\subsection{Existence of spectral gap for normal and block-normal MPS}

An important consequence of \cref{theorem:gamma-exact}, in combination with the
preceding Lemmas~\ref{lem:block-obc} and~\ref{lem:block-pbc}, and
\cref{theorem:gap-from-gamma}, is a simple proof that parent Hamiltonians for
any normal or block-normal MPS are gapped.

\begin{corollary}
\label{cor:gamma-to-zero}
For normal and block-normal MPS, $\gamma_{p,q}\to 0$ 
as $p\to\infty$ in \cref{theorem:gamma-exact}.
\end{corollary}
\emph{Proof.---}First, note that all objects in \cref{eq:thm-gammaexact-defM} are of finite and fixed dimension, and thus, convergence in any norm is equivalent. It holds that
$\mathcal T^k\to
\mathcal T_\infty$ as $k\to\infty$, where  $\mathcal T_\infty$ is of the form
\cref{transferlimit} for normal MPS, and a direct sum of this form over the
blocks $\alpha$ for block-normal MPS. 
$T_\infty$ (i.e., $\mathcal T_\infty$ acting from bottom to top) is full rank on the entire space for normal MPS, or on the subspaces corresponding to the individual normal blocks for block-normal MPS, and therefore, also $(T_k)^{-1}\to (T_\infty)^{-1}$.
As all objects involved in \cref{eq:thm-gammaexact-defM} are of fixed and finite dimension, $\Xi$ converges to \eqref{eq:thm-gammaexact-defM} with $\mathcal T^k$ replaced by $\mathcal T_\infty$, and $T_k$ by $T_\infty$, which can be easily seen to be a projector, and thus has second largest eigenvalue $0$.
\hspace{\fill}$\square$

Note that this proof crucially relies on the fact that unlike $\hat\Xi$,  $\Xi$
lives in a space of fixed dimension, independent of $p$ and $q$.

\cref{cor:gamma-to-zero} implies the following result.

\begin{theorem}[All parent Hamiltonians are gapped]\label{thm:all_parent_gapped}
Any parent Hamiltonian $H_N$ of a normal or block-normal MPS is gapped, both for
open and periodic boundary conditions; that is, there is a lower bound to
$\Delta(H_N)$ which is uniform in $N$.
\end{theorem}
\emph{Proof.---}For system sizes $N$ which are multiples of a sufficiently
large $p_0$, this follows directly from combining \cref{lem:block-obc} with
\cref{theorem:gap-from-gamma} and \cref{cor:gamma-to-zero}. Specifically, let
$q=0$ and choose a $p_0$ which is sufficiently large such that
$\gamma_{p_0,0}<1/2$ (which holds for all $p_0$ above some threshold, as $\gamma_{p,0}\to0$). Then, by virtue of
\cref{lem:block-obc} it holds for the blocked Hamiltonian $\tilde H_M$ with
blocking $p=p_0$, $q=0$ that $\Delta(H_N)\ge \Delta(\tilde H_M)$, and from
\cref{theorem:gap-from-gamma} that $\Delta(\tilde H_M)\ge \varepsilon =
(1-2\gamma_{p_0,0})\Delta(\tilde h_{\tilde\imath})>0$. Thus, $\Delta(H_N)\ge
\varepsilon$ provides a uniform lower bound on the gap of $H_N$ for $N=p_0M$ ($M\ge3$ for PBC).

To obtain a uniform lower bound for arbitrary system sizes $N\ge 3p_0$, write $N=Mp_0+q$, $0\le q<p_0$. It is straightforward to check that \cref{theorem:gap-from-gamma} can be generalized to the case where one Hamiltonian term acts on $p_0+q+p_0$ sites, where the gap is governed by the maximal $\gamma_{p_0,q}$ and the minimal $\Delta(\tilde h_{\tilde\imath})$. Since $q$ takes finitely many values, this will give a non-trivial bound for sufficiently large $p_0$.

If one wants a bound which holds for \emph{all} $N$, even below $3p_0$, one can take the
minimum of this bound and the finite-size gaps for $N<3p_0$.
\hspace*{\fill}$\square$

\subsection{Explicit lower bounds on the spectral gap}\label{subsection:explicit_bounds}

\subsubsection*{Computing lower bounds on the gap}

In order to compute lower bounds on the gap, one can use either
\cref{lem:block-obc} or \cref{lem:block-pbc} to construct a blocked
Hamiltonian with block sizes $p$ and $q$, for which the exact
$\gamma_{p,q}$ can be computed using \cref{theorem:gamma-exact}. Then, one
applies \cref{theorem:gap-from-gamma} to obtain
\begin{equation}\label{eq:general_gap_bound}
\Delta(H_N) \ge \Delta(\tilde H_M) \ge 
    \varepsilon=(1-2\gamma_{p,q})\, \Delta(\tilde h_{\tilde\imath})\ .
\end{equation}
Note that depending on the Lemma used, the $\tilde h_{\tilde\imath}$ are
defined quite differently: When using \cref{lem:block-obc}, which applies
both to OBC and PBC, $\lambda_\omega+\lambda_{\omega'} = 1$ for specific
pairs $(\omega,\omega')$ and $\lambda_\omega=1$ otherwise, cf.\
\cref{eq:lemma-obc-lambdai}, while for \cref{lem:block-pbc}, which is
restricted to PBC, $\sum\lambda_\omega=p+q$, cf.\
\cref{eq:lemma-pbc-lambdai}. The choice of $\lambda_\omega$
can be used to optimize  $\Delta(\tilde{h}_{\tilde{\imath}})$ (for given $p$, $q$),
subject to the linear constraints above and the semidefinite constraints $\lambda_\omega\ge0$. 
A canonical choice is to choose
all $\lambda_\omega$ equal, within the corresponding constraints.

A key choice one has to make when computing a bound on the gap is the
choice of $p$ and $q$. We discuss this tradeoff more specifically in
\cref{sec:examples}.  For now, let us point out that the computationally
most costly step is computing $\Delta(\tilde h_{\tilde\imath})$, which
scales exponentially with the support $r=2p+q$ of $\tilde
h_{\tilde\imath}$. Given a maximal feasible $r$, there is just few
choices of $p$ and $q$ one has to check.  While we expect the largest
$r$ to give the best bound, additionally checking blockings $(p,q)$
with smaller $r=2p+q$ comes at a subleading cost, and can thus be
easily incorporated.

\subsubsection*{Convergence to infinite system value}

The uniform lower bound on the gap $\Delta(H_N^{\mathrm{PBC}})$ of the PBC
Hamiltonian obtained using \cref{lem:block-pbc} converges to the infinite system
value of the OBC gap, $\limsup_n \Delta(H_n^\mathrm{OBC})$, as $p,q\to\infty$,
$q/p\to\infty$. To this end, choose $\lambda_\omega=
\rho_{p,q}:=\tfrac{p+q}{2p+q-(n-1)}\to 1$, such that $\tilde h_{\tilde\imath}=
\rho_{p,q}H_{2p+q}^\mathrm{OBC}$. Then, the bound becomes $\Delta(H_N^\mathrm{PBC})\ge
(1-2\gamma_{p,q}) \rho_{p,q}\Delta(H_{2p+q}^\mathrm{OBC})$, which (since
$\gamma_{p,q}\to 0$, \cref{cor:gamma-to-zero}) converges to
$\limsup_n\Delta(H_{n}^\mathrm{OBC})$.

\subsubsection*{Computational cost}

The computational bottleneck of the method is the computation of the ``local
gap'' $\Delta(\tilde h_{\tilde\imath})$. $\tilde h_{\tilde\imath}$ is a square
matrix of dimension $d^r$, $r={2p+q}$, whose gap can be computed in time
$\mathcal O(d^{3r})=\mathcal O(d^{6p+3q})$ using full diagonalization or
$\mathcal O(d^{r+n})=\mathcal O(d^{2p+q+n})$ using a Krylov method (using that
$\tilde h_{\tilde\imath}$ is a sum of $n$-body terms), respectively.

In comparison, the resources for computing $\gamma_{p,q}$ using
\cref{theorem:gamma-exact} are as follows. The central task is to compute the
second leading eigenvalue of the matrix $\Xi$.  $\Xi$ can be computed at a
computational cost $\mathcal{O}(D^{10})$, and fully diagonalized in
$\mathcal{O}(D^{12})$. Using a Krylov solver and using the structure of $\Xi$,
the computational cost scales as $\mathcal{O}(D^{6})$. The $p$- and
$q$-dependence only enters in computing $\mathcal{T}_p$, $\mathcal{T}_{p+q}$,
and $T_{2p+q}^{-1}$, which are $D^2\times D^2$ matrices, and thus enters as
$\mathcal O((2p+q)D^6)=\mathcal O(rD^6)$. 
Note that the
bottleneck of the method is thus always the computation of $\Delta(\tilde
h_{\tilde\imath})$ in \cref{theorem:gap-from-gamma}: The smallest blocking is
$p=n-1$, $q=0$, and the smallest $n$ is the injectivity length, which in turn is
at least $n \gtrsim \log_d(D^2)$.  Thus, computing $\Delta(\tilde
h_{\tilde\imath})$ scales at least as $\mathcal O(d^{2p+q+n})=\mathcal O(D^6)$. Additionally, one
should keep in mind that the optimal bound on the gap (or even a bound at all)
is typically reached for larger values of $p$.

Note that without our \cref{theorem:gamma-exact}, the bottleneck would be
the computation of $\gamma_{p,q}$ from \cref{eq:PQ-QP-condition}, which
requires computing the lowest eigenvalue of matrix on the l.h.s.\ of
\cref{eq:PQ-QP-condition}  and thus scales as $\mathcal O(d^{3p+2q+n})$ when
using a Krylov method, and $\mathcal O(d^{9p+6q})$ with full
diagonalization.

\subsection{Generalizations}

Let us briefly discuss some possible generalizations to our techniques.

First, one can optimize the $\lambda_\omega$ (given $p$ and $q$)  to yield the
best possible bound on $\Delta(\tilde h_{\tilde\imath})$, and thus on the gap.
Note that this is a semidefinite program (SDP)
and can thus be carried out efficiently. This optimization can be significantly
generalized by not simply rearranging the existing terms $h_i$ with different
weights $\lambda_\omega$, but rather expressing the sum of the $h_i$ within 
each $p$-sized block as $H_p=\sum h_i=A+B$, 
subject to the constraint that 
$A,B\ge0$ and $\ker A=\ker B=\ker H_p$, 
where $A$ and $B$ are absorbed in the left 
and right $\tilde h_{\tilde \imath}$, respectively,
cf.\ Ref.~\cite{Rai2026hierarchyofspectral};
this optimization is yet again an SDP.

Another direction in which one can generalize the method is to not block
to a nearest-neighbor Hamiltonian, but to a Hamiltonian which is a sum of
terms which overlap in more than one way. Then,
\cref{theorem:gap-from-gamma} needs to be generalized to bound each
distinct pair
of overlapping Hamiltonian terms individually. Depending on how these bounds
are constructed, either each of them depends only on one pair of terms (in
which case the variant of \cref{eq:PQ-QP-condition} with the worst
$\gamma$ determines the gap), or one can rearrange terms between them,
giving a set of additional optimization parameters. Another modification in
\cref{theorem:gap-from-gamma} is to not lower-bound the $\tilde h_{\tilde
\imath}$ by projectors. In that case, $\tilde h_{\tilde\imath}$ itself
appears in \cref{eq:PQ-QP-condition}: This might yield a better 
$\gamma_{p,q}$ for the gap bound \cref{gapeq}, yet at the cost that Jordan's Lemma cannot be applied any
more, and thus the optimal $\gamma_{p,q}$ for the modified
\cref{eq:PQ-QP-condition} must be determined directly from said inequality
 (and thus at a higher
computational cost).

Finally, our method can be generalized to apply to arbitrary frustration-free
Hamiltonians (translation invariant or not). To this end, note that we
have only ever used that the ground space for any contiguous region can be
parametrized by an MPS, but we have not assumed that this MPS is
translation invariant, nor even that the MPS tensors are independent of
the block size. Since for any frustration free Hamiltonian, the ground
space can be parametrized by an MPS, which in addition can be constructed
efficiently from the Hamiltonian~\cite{schuch_simple_2025}, this allows to
apply our technique to bound the gap of arbitrary frustration free
Hamiltonians, with the computational cost scaling polynomially with the
ground state degeneracy. A special case of this scenario is given
by not translation invariant MPS and their parent Hamiltonians, to
which our technique straightforwardly generalizes.

\section{Examples}
\label{sec:examples}

\subsection{Introduction}
\label{sec:examples-introduction}

In this section, we provide examples which illustrate the application of our method and showcase its advantages.
We examine Hamiltonians with a normal MPS ground state, such as the AKLT model in~\cref{secaklt} and the VBS for an $\operatorname{SU}(3)$ chain with a 10-dimensional spin representation in~\cref{secrep10}, as well as Hamiltonians with a degenerate ground state spanned by block-normal MPS, namely a three-dimensional deformed clock model (see~\cref{z3sec}) and another $\operatorname{SU}(3)$ VBS with an 8-dimensional representation (\cref{secvbs8}). Through the sufficient conditions provided in~\cref{subsec:parenthams}, we verify that all investigated Hamiltonians are indeed parent Hamiltonians of either normal or block-normal MPS, and thus they are gapped by~\cref{thm:all_parent_gapped}. 
We then apply our method to derive quantitative uniform lower bounds on the spectral gap of each of these models. To this end, we first construct a blocked Hamiltonian:  For OBC,
we do so by using \cref{lem:block-obc} with the canonical choice of $\lambda_\omega=\lambda_{\omega'}=\tfrac{1}{2}$ for $\omega=1,\dots,p-(n-1)$, $\omega'= \omega+(p+q)$, and $\lambda_\omega = 1$ otherwise; for PBC, we apply~\cref{lem:block-pbc} with $\lambda_\omega = \tfrac{p+q}{2p+q-(n-1)}$ for all $\omega$. Subsequently, we apply 
\cref{theorem:gap-from-gamma} to the blocked Hamiltonian, using \cref{theorem:gamma-exact} to compute $\gamma_{p,q}$.

In order to assess the strengths and limitations of our method, we compare it
with other known strategies to bound the spectral gap of frustration-free
systems. To start, we demonstrate the improvements our method provides over
directly applying the FNW bound~\cite{fannes_finitely_1992} or the more general
martingale method of Nachtergaele~\cite{nachtergaele_spectral_1996} for parent
Hamiltonians of (block-)normal MPS. These bounds have the same form
as~\cref{eq:general_gap_bound}, but with the choice $q=0$, $\lambda_\omega =
\tfrac{1}{2}$ for all $\omega$, and with a specific upper bound on
$\gamma_{p,q}$ which depends on spectral properties of the transfer matrix, and
which we provide in~\cref{appendix:fnw-n} for completeness. Moreover, we
compare our method to strategies following a different philosophy -- the
so-called finite-size criteria, which establish a relation between the gap of a
finite-size version of the given Hamiltonian to a uniform bound on the gap for
arbitrary system sizes. Specifically, these are  Knabe's method with
periodic~\cite{knabe_energy_1988} and open
boundaries~\cite{wouters_interrelations_2021}, and the Gosset-Mozgunov
bound~\cite{gosset_local_2016}. We provide a summary of these methods in
\cref{appendix:knabe_and_gosset}.

In the following, in order to distinguish between boundary conditions, we will
denote OBC Hamiltonians by $H$ and PBC Hamiltonians by $H^\circ$. 

\subsection{The AKLT model}
\label{secaklt}
The AKLT model, introduced by Affleck, Kennedy, Lieb, and Tasaki~\cite{affleck_valence_1988}, is a paradigmatic example of a quantum spin chain which has a normal MPS as its unique ground state. The Hamiltonian defining the model is 2-local, and it is given by
\begin{equation}
    h_\mathrm{AKLT} = \frac{1}{6}(\mathbf{S}_i \cdot\mathbf{S}_{i+1})^2 + \frac{1}{2} \mathbf{S}_i \cdot\mathbf{S}_{i+1} + \frac{1}{3}\ ,
\end{equation}
where $\mathbf{S}=(S_x,S_y,S_z)$ are the spin-1 operators. The Hamiltonian $H_\mathrm{AKLT} = \sum_i (h_\mathrm{AKLT})_i$ is by construction the canonical weak parent Hamiltonian (cf.~\cref{localparent}) of the MPS generated by the tensor $A$ with
\begin{equation}
\label{eq:aklt-tensor}
    A^1 = \sqrt{\frac{1}{3}}\begin{pmatrix}
        0 & \sqrt{2} \\ 0 & 0
    \end{pmatrix}, \quad A^0 = \sqrt{\frac{1}{3}}\begin{pmatrix}
        -1 & 0 \\ 0 & 1
    \end{pmatrix}, \quad A^{-1} = \sqrt{\frac{1}{3}}\begin{pmatrix}
        0 & 0 \\ -\sqrt{2} & 0
    \end{pmatrix}.
\end{equation}
This MPS tensor is normal with injectivity length $\ell=2$: It is straightforward to
see that $\operatorname{span}\{A^{i} A^{j}|\ i,j=-1,0,1\} =\mathcal{M}_2$. One
can also directly check that the MPS $\ket{\Psi_3(A,X)}$ spans precisely the ground
space
of $h_{\mathrm{AKLT},3}:=h_\mathrm{AKLT} \otimes \mathbb{1} + \mathbb{1} \otimes
h_\mathrm{AKLT}$,\footnote{Recall that since $h_\mathrm{AKLT}$ is a weak
parent Hamiltonian, $\mathcal S_3=\{\ket{\Psi_3(A,X)}|X\}$ is contained in 
the ground space of $h_{\mathrm{AKLT},3}$, 
 and thus it is sufficient to check that 
$\dim(\ker(h_{\mathrm{AKLT},3})) =\dim S_3=4$.}
which acts on $\ell+1=3$ sites,
and thus $H_\mathrm{AKLT}$ is a
parent Hamiltonian of the MPS \eqref{eq:aklt-tensor}; in particular, the MPS
$\ket{\Psi_N(A,X)}$ is its unique ground  state (with $X=\openone$ for PBC, while the
boundary of the OBC ground state depends on $X$, cf.~\cref{subsec:parenthams}).
\cref{thm:all_parent_gapped} then implies that $H_\mathrm{AKLT}$ is gapped above
the ground state.

Let us now apply \cref{eq:lemma-obc-lambdai} and \cref{eq:lemma-pbc-lambdai}, respectively (with the choice of $\lambda_\omega$ laid out in \cref{sec:examples-introduction}), followed by \cref{theorem:gap-from-gamma} and \cref{theorem:gamma-exact}, to obtain quantitative bounds on the gap.

\begin{figure}
\includegraphics[width=\columnwidth]{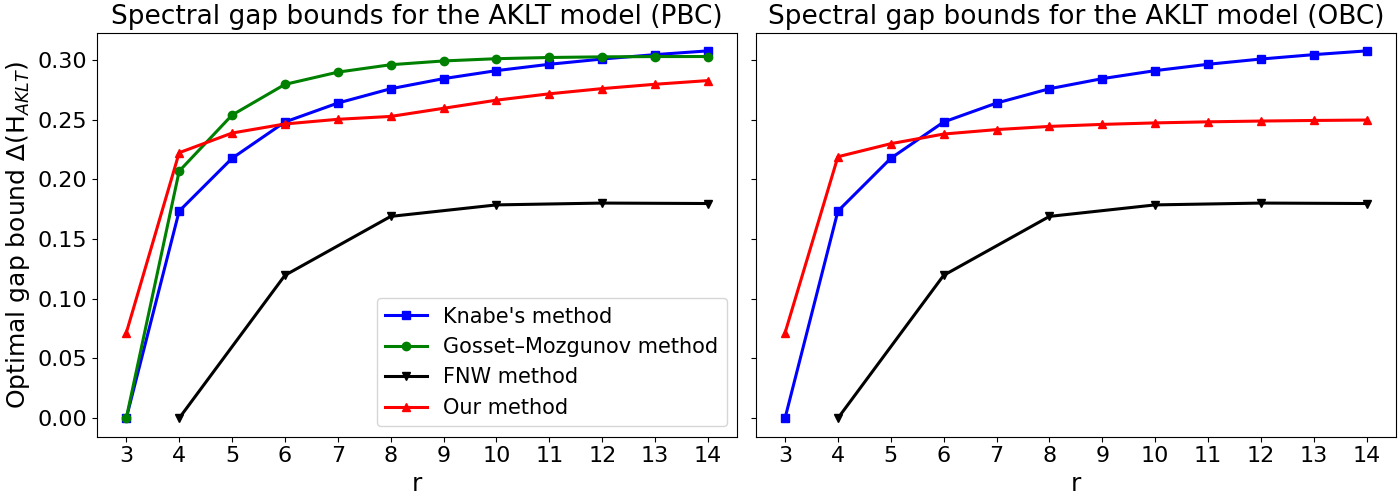}
\caption{
Lower bounds on the spectral gap for the AKLT chain, obtained with
different methods, plotted vs.\ the computational cost (exponential in $r$,
namely, computing the spectral gap of a matrix of size $3^r$),
for PBC (left) and
OBC (right), with shared $y$ axis. We compare our bounds to those obtained with
the FNW method, upon which our method is built, as well as two finite-size
criteria, the Knabe bound and the Gosset-Mozgunov bound (where the latter only
applies to PBC). See text for details. } \label{fig:aklt_gaps} 
\end{figure}

Recall that the computational bottleneck of the method is the diagonalization of
a matrix of size exponential in $r=2p+q$ (which thus scales exponentially in $r$). The smallest such $r$ for which 
\cref{theorem:gamma-exact} gives $\gamma_{p,q} < 1/2$ (and thus a non-trivial 
bound on the gap) is $r=3$, with $p=q=1$; 
finding a nonzero uniform gap thus only requires the diagonalization of a 3-body
Hamiltonian. The obtained bounds improve with larger block sizes:
In~\cref{fig:aklt_gaps}, we show the best bounds obtained with our method
(optimizing over $p$ and $q$ s.th.\ $r=2p+q$) for $3\le r\le 14$. 
For PBC, we obtain the best bound  for $p=3,\ q=8$, namely
\begin{equation}
    \Delta(H^\circ_{\mathrm{AKLT}})  \ge 0.2827\ .
\end{equation}
For OBC, the best bound is obtained for $p=2$ and $q=10$, which is
\begin{equation}
    \Delta(H_{\mathrm{AKLT}}) \geq 0.2497\ .
\end{equation}
For comparison, the best bound provided by the FNW method for $r\le14$ is
$\Delta(H_{\mathrm{AKLT}}) \geq 0.1801$ for the choice $p=6$.\footnote{To
calculate these bounds -- using the method as described in \cref{sec:fnw} -- 
we used that the transfer matrix of the AKLT MPS is normal, so we set $a(p) =
4\cdot\lambda_2(\mathcal{T})^p$ in~\cref{fnw}.}

\cref{fig:aklt_gaps} also provides a comparison of the OBC and PBC bounds
obtained with our method vs.\ other methods. In addition to the FNW method, we
also compare to the Knabe and the Gosset-Mozgunov
bound~\cite{knabe_energy_1988,gosset_local_2016} (cf.~\cref{appendix:knabe_and_gosset}), which are both based on
relating the system-size independent gap bound to the finite-size gap of the
same Hamiltonian. In all of these methods, the computational bottleneck is to
determine the gap of a Hamiltonian on $r$ sites, which we use as cost figure
(i.e., x-axis) to compare the methods in
\cref{fig:aklt_gaps}. Our method significantly improves upon the FNW bound for
all block sizes $r$. On the other hand, we find that while for small $r$, our
method also gives better bounds than the Knabe and Gosset-Mozgunov finite-size
criteria, both of them outperform our method on the AKLT model as $r$ is increased.

\subsection{$\mathbb{Z}_3$-XY model}
\label{z3sec}
The $\mathbb{Z}_3$-XY model, introduced in Ref.~\cite{wouters_interrelations_2021}, is a deformation of the three-state Potts clock model without external field which generalizes the frustration-free line of the $\mathbb{Z}_2$-XY model with a transverse field. The deformation is parametrized by a positive real number, denoted by $t$.

The two-local Hamiltonian of the $\mathbb{Z}_3$-XY model $H_\mathrm {XY}(t)=\sum_i (h_\mathrm{XY}(t))_i$ is given by
\begin{equation}
\label{eq:def-hxy}
    h_{\mathrm{XY}}(t) = \epsilon(t) - \left[ \left( 1 + b(t) \tau_i + b(t) \tau^{\dagger}_i\right) \sigma^{\dagger}_i \sigma_{i+1} \left(1 + b(t) \tau_{i+1} + b(t) \tau^{\dagger}_{i+1} \right) + \frac{f(t)}{2} \left(\tau_i + \tau_{i+1} \right) + \text{h.c.} \right],
\end{equation}
with 
\begin{equation}
    \sigma = \begin{pmatrix}
        0 & 1 & 0 \\
        0 & 0 & 1 \\
        1 & 0 & 0
    \end{pmatrix}, \qquad \tau = \begin{pmatrix}
        1 & 0 & 0 \\
        0 & e^{i\frac{2\pi}{3}} & 0 \\
        0 & 0 & e^{i \frac{4\pi}{3}}
    \end{pmatrix},
\end{equation}
and parameters
\begin{equation}
    f(t) = \frac{6(1-t^6)}{(t^3+2)^2}, \qquad b(t) = \frac{t^3-1}{t^3+2}, \qquad \epsilon(t) = \frac{6(t^6+2)}{(t^3+2)^2}\ ,
\end{equation}
such that for $t=1$ we recover the Potts Hamiltonian:
\begin{equation}
    h_{\mathrm{XY}}(1) =  h_{\mathrm{Potts}} = 2-\sigma_i^\dagger \sigma_{i+1} - \sigma_i \sigma_{i+1}^\dagger\ .
\end{equation}
\begin{figure}[t]
    \centering
    \includegraphics[width=1\linewidth]{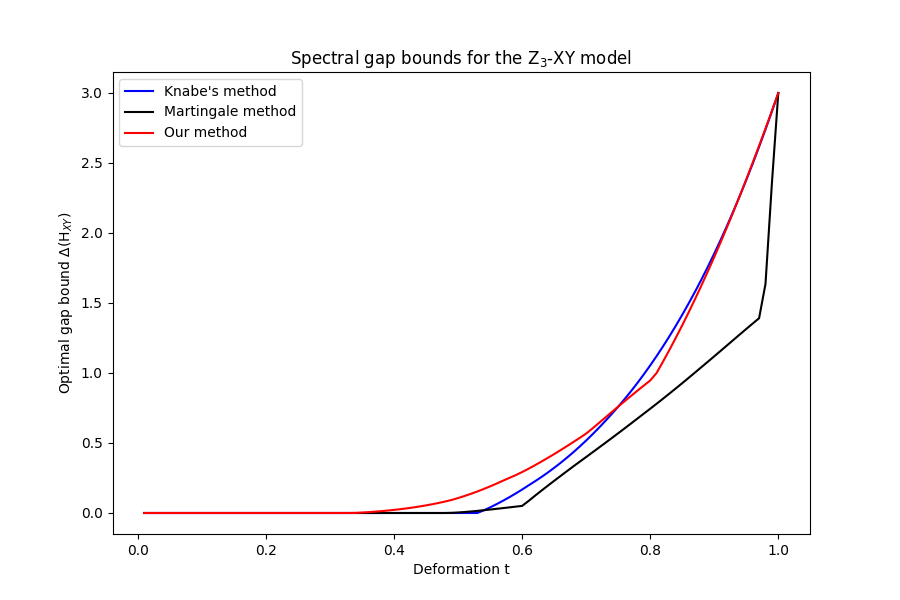}
\caption{\label{fig:02}
Comparison of gap bounds obtained by our method vs.\ a 
variant of  Nachtergaele's martingale method (\cref{appendix:fnw-n},
\cref{thm:martingale,theorem:nachtergaele2}) and a modification of Knabe's method for OBC~\cite{wouters_interrelations_2021}.
The intervals in which the different methods yield a non-zero bound on
the gap are 
$0.4775 < t < 2.0944$ (martingale method), 
$0.5310 < t < 1.8832$ (Knabe method), and 
$0.3374 < t < 2.9638$ (our method, with $(p,q)=(5,0)$).
The kink shown in the curve obtained through our method around $t\approx
0.81$ can be explained from the fact that for $0.81\lesssim t \le 1$,
the choice $(p,q)=(1,0)$ yields the best bound, while for $t \lesssim
0.81$, $r=2p+q=10$ provides the best bound, where  $p$ and $q$ vary. 
The kinks observed in the result from the martingale method have a similar origin.\protect\footnotemark
}
\end{figure}
\footnotetext{For the martingale method there is a kink at $t \approx 0.97$, as well as another kink at $t \approx 0.61$, since \cite{nachtergaele_spectral_1996} provides two similar ways to bound the same quantity, and far from $t=1$ the bound given in~\cref{theorem:nachtergaele2} with $r=10$ is larger, and provides an explicit lower bound on the spectral gap in a larger parameter interval than the ``standard martingale method'', i.e., the combination of~\cref{theorem:gap-from-gamma} with $q=0$, $\lambda_\omega = \tfrac{1}{2}$ for all $\omega$ and~\cref{thm:martingale} for obtaining $\gamma_{p,q}$.}

It is straightforward to check that the OBC Potts Hamiltonian is frustration
free with ground state energy zero, and has a three-fold degenerate ground state
for any system size. The deformation that results in \cref{eq:def-hxy}
preserves that degeneracy for any $t\in (0,\infty)$. This can be shown by
observing that $h_\mathrm{XY}(t)$ is obtained by conjugating $h_\mathrm{Potts}$,
which is positive semidefinite with ground state energy zero, with a product of
single-site invertible operators $\Lambda^{\otimes 2}$, and therefore the zero-energy ground
space of $H_{\mathrm{XY}}(t)$ and $H_\mathrm{Potts}$ are related by conjugation
with  $(\Lambda^{-1})^{\otimes N}$; see  Ref.~\cite{wouters_interrelations_2021}
for details. 
The ground states of $H_{\mathrm{XY}}(t)$ are thus
product states, whose overlaps decay exponentially with the system size. Therefore,
they have a block-injective MPS representation which is generated by the (unnormalized) tensor
\begin{equation}
    A^j = \frac{1}{\sqrt{3}} \begin{pmatrix}
        t^j & 0 & 0 \\
        0 & t^j e^{i\frac{2j\pi}{3}} & 0 \\
        0 & 0 & t^j e^{i\frac{4j\pi}{3}}
    \end{pmatrix}, \quad j\in\{0,1,2\}.
\end{equation}
This particularly simple ground space structure implies that $H_{\mathrm{XY}}$
is a parent Hamiltonian of the MPS generated by $A$ for any deformation
$0<t<\infty$. It thus immediately follows 
by~\cref{thm:all_parent_gapped} 
that $H_\mathrm{XY}(t)$ is gapped 
for all $0<t<\infty$.

Let us now compare the performance of our method to the Knabe bound
and Nachtergaele's martingale method. Since the Hamiltonian is defined
with OBC, we apply~\cref{theorem:gap-from-gamma} together with the
canonical blocking from~\cref{lem:block-obc} and the optimal
$\gamma_{p,q}$ from~\cref{theorem:gamma-exact}. 
To ensure comparability, we fix the computational cost to
diagonalizing a Hamiltonian on $r\le 10$ sites. 
\cref{fig:02} shows the optimal bounds on the gap (optimized over $r\le 10$)
in the interval $0<t\leq 1$ for the different methods.
We see that our
approach, based on applying~\cref{theorem:gap-from-gamma} together
with~\cref{theorem:gamma-exact} outperforms the martingale method, and
provides an improved bound as compared to Knabe's method in the high-deformation (small $t$) regime. 
Moreover, our method is able to provide a quantitative lower bound also in regimes where Knabe's bound fails to do so ($0.3374 < t < 0.5310$).

\subsection{$\operatorname{SU}(3)$ valence bond solids}
\label{ex3}
Historically, valence bond solid (VBS) models are regarded to be precursors to the MPS ansatz, with all VBS states also possessing an exact matrix product state representation~\cite{perez-garcia_matrix_2007}. Several $\operatorname{SU}(3)$-symmetric VBS models were introduced in~\cite{greiter_valence_2007}, with some of them conjectured to being gapped. In this subsection, we prove that two of these Hamiltonians are indeed gapped, and provide explicit lower bounds on their gaps using our method, which we then compare to other finite-size strategies.

\subsubsection{General strategy for constructing 
$\operatorname{SU}(3)$-symmetric models}

We first sketch the basic strategy to construct $\operatorname{SU}(3)$-symmetric VBS states and frustration-free two-local Hamiltonians with a zero energy state described by such a VBS state in abstract terms.  
For a more detailed discussion, we refer to~\cite{greiter_valence_2007}. 
The reader might find it helpful to study this abstract exposition alongside with or after going through the concrete realizations of such models discussed in 
\cref{secvbs8} and
\cref{secrep10}.
We take a tensor product of translations of a singlet in $\bigotimes_i \mathbf{j}_i$ on a grid, where $\mathbf{j}_i$ are irreducible representations of $\operatorname{SU}(3)$, and project onto a specific irreducible representation $\mathbf{r}$ that appears in the direct sum decomposition of $\bigotimes_i \mathbf{j}_i$ on each lattice site (see e.g.,~\cref{repr8,rep10}). We set the two-local Hamiltonian to be
\begin{equation}\label{eq:general_vbs}
    h = \prod_{\mathbf{m} \in S} ((\mathbf{J} \otimes \mathbb{1})(\mathbb{1} \otimes \mathbf{J}) - \sigma(\mathbf{m})) \eqcolon  \prod_{\mathbf{m} \in S} (\mathbf{J}_i\mathbf{J}_{i+1} - \sigma(\mathbf{m})),
\end{equation}
where $\mathbf{J}$ denotes the vector with entries being the eight $r \times r$ matrices that are the generators of $\operatorname{SU}(3)$ in the irreducible representation $\mathbf{r}$, $S$ is the set of irreducible representations $\mathbf{m}$ that appear both in the direct sum decomposition of $\mathbf{r} \otimes \mathbf{r}$ and in the two-site restriction of the unprojected chain (i.e., the tensor product of the translations of the singlet), and $\sigma(\mathbf{m})$ is the eigenvalue of the quadratic Casimir operator in the representation $\mathbf{m}$. This Hamiltonian annihilates the projected state on two sites. If not already positive semidefinite, it can be modified to yield positive eigenvalues for all irreducible representations in the decomposition of $\mathbf{r} \otimes \mathbf{r}$ that are not in $S$, e.g.\ by squaring, or by multiplying with a suitable polynomial in $\mathbf J_i\cdot \mathbf J_{i+1}$, making the constructed VBS state a ground state of $h$ (the latter can also be used to turn $h$ into a projector).

Throughout this section, we will employ the following graphical notation: \begin{tikzpicture}[baseline={([yshift=-.5ex]current bounding box.center)}]
    \draw[fill=white] (0,0) circle (2pt); \end{tikzpicture} will denote an $\operatorname{SU}(3)$ spin that transforms under the fundamental representation $\mathbf{3}$, whereas $\times$ will denote a spin that transforms under its conjugate representation $\bar{\mathbf{3}}$.

\subsubsection{The representation $\mathbf{8}$ VBS}
\label{secvbs8}

We start by discussing a VBS with $\operatorname{SU}(3)$-representation $\mathbf{8}$ per site, that is, the adjoint representation.
In order to obtain a VBS that is symmetric with respect to the adjoint representation 
of $\operatorname{SU}(3)$, we consider the singlet in the space $\mathbf{3} \otimes \bar{\mathbf{3}}=\mathbf{1} \oplus \mathbf{8}$, denoted graphically by \begin{tikzpicture}[baseline={([yshift=-.5ex]current bounding box.center)}]
    \draw (0,0) -- (0.5,0);
    \draw[fill=white] (0,0) circle (2pt);
    \draw (0.5,0) node {$\times$};\end{tikzpicture}. Then, we project the neighboring ends of each singlet pair (corresponding to one lattice site) onto the representation $\mathbf{8}$. This construction can be visually represented as
\begin{equation}
\label{repr8}
    \begin{tikzpicture}[baseline={([yshift=-.5ex]current bounding box.center)}]
    \draw (0,0) -- (0.5,0);
    \draw[fill=white] (0,0) circle (2pt);
    \draw (0.5,0) node {$\times$};
    \draw (1,0) -- (1.5,0);
    \draw[fill=white] (1,0) circle (2pt);
    \draw (1.5,0) node {$\times$};
    \draw (2,0) -- (2.25,0);
    \draw[fill=white] (2,0) circle (2pt);
    \draw (2.6,-0.125) node {$\dots$};
    \draw (0,-0.25) node {$\times$};
    \draw (0,-0.25) -- (-0.25,-0.25);
    \draw (-0.6,-0.125) node {$\dots$};
    \draw (0.5,-0.25) -- (1,-0.25);
    \draw[fill=white] (0.5,-0.25) circle (2pt);
    \draw (1,-0.25) node {$\times$};
    \draw (1.5,-0.25) -- (2,-0.25);
    \draw[fill=white] (1.5,-0.25) circle (2pt);
    \draw (2,-0.25) node {$\times$};
    \draw[blue, fill=blue,opacity=0.2] (-0.15,-0.5) rectangle (0.15,0.25);
    \draw[blue, thick] (0,-0.5) -- (0,-0.75);
    \draw[blue, fill=blue,opacity=0.2] (0.35,-0.5) rectangle (0.65,0.25);
    \draw[blue, thick] (0.5,-0.5) -- (0.5,-0.75);
    \draw[blue, fill=blue,opacity=0.2] (0.85,-0.5) rectangle (1.15,0.25);
    \draw[blue, thick] (1,-0.5) -- (1,-0.75);
    \draw[blue, fill=blue,opacity=0.2] (1.35,-0.5) rectangle (1.65,0.25);
    \draw[blue, thick] (1.5,-0.5) -- (1.5,-0.75);
    \draw[blue, fill=blue,opacity=0.2] (1.85,-0.5) rectangle (2.15,0.25);
    \draw[blue, thick] (2,-0.5) -- (2,-0.75); 
    \end{tikzpicture},
\end{equation}
where the vertical blue lines represent the projection onto $\mathbf{8}$. The simplest MPS tensor that generates the state in \cref{repr8} consists of the Gell-Mann matrices $\Lambda^{a}$, $a = 1,\dots,8$ (see \cref{gm} and \cite{morimoto_z_3_2014} for details). However, it is convenient to transform to a physical basis where the MPS tensor consists of real matrices. The resulting (unnormalized) MPS tensor $A$ is then given by the raising and lowering operators, together with the generators of the Cartan subalgebra of $\mathfrak{su}(3)$:
\begin{equation}
\label{rep8mps}
\begin{aligned}
    A^{\pi^{\pm}} = & \frac{1}{2}(\Lambda^1 \pm i\Lambda^2), \quad A^{\pi^0} = \frac{1}{\sqrt{2}}\Lambda^3, \quad A^{K^{\pm}} = \frac{1}{2}(\Lambda^4\pm i\Lambda^5), \\ 
    A^{K^0} = & \frac{1}{2}(\Lambda^6 + i\Lambda^7), \quad A^{\bar{K}^0} = \frac{1}{2}(\Lambda^6 - i\Lambda^7), \quad A^{\eta} = \frac{1}{\sqrt{2}}\Lambda^8,
\end{aligned}
\end{equation}
where the physical indices correspond to the octet of mesons \cite{morimoto_z_3_2014}. As the linear span of $A^jA^k$ is $\mathcal{M}_3$, the tensor $A$ is normal with injectivity length 2.

For the state sketched in~\cref{repr8}, one can easily check that the Hamiltonian constructed in~\cref{eq:general_vbs} is
\begin{equation}
\label{vbs8ham}
    h_{\operatorname{SU}(3),\mathbf{8}} =  (\mathbf{J}_i \mathbf{J}_{i+1})^2 + \frac{9}{2} \mathbf{J}_i \mathbf{J}_{i+1} + \frac{9}{2},
\end{equation}
where the coefficients of $\mathbf{J}$ are given by \cite{morimoto_z_3_2014}
\begin{equation}
    (J^{a})^{jk} = \frac{1}{2} \operatorname{tr}((A^{j})^{\dagger}(\Lambda^{a}A^k - A^k \Lambda^{a})),
\end{equation}
where $a \in \{1,\dots,8\}$ and the indices $j$ and $k$ correspond to those 
given in \cref{rep8mps}.\footnote{The projector-valued, i.e., spectrally flattened, Hamiltonian is
\begin{equation}
\label{vbs8proj}
    k_{\operatorname{SU}(3),\mathbf{8}} = \left((\mathbf{J}_i \mathbf{J}_{i+1})^2 + \frac{9}{2} \mathbf{J}_i \mathbf{J}_{i+1} + \frac{9}{2}\right)\left(\frac{2}{9} - \frac{11}{90}\mathbf{J}_i \mathbf{J}_{i+1}\right).
\end{equation}
}
This Hamiltonian satisfies $\mathcal{S}_2^A \subset \ker h_{\operatorname{SU}(3),\mathbf{8}}$, where $\mathcal{S}_2^A$ is the space of two-body MPS generated by the tensor $A$ (cf.~\cref{mpsspace}). Notice that $h_{\operatorname{SU}(3),\mathbf{8}}$ is reflection invariant, thus the reflection of the state sketched in~\cref{repr8} is also a ground state. This state is also an MPS, generated by the tensor $B$, where $B^j = (A^j)^T$~\cite{greiter_valence_2007}, thus $B$ is also normal with injectivity length 2. One can verify that $\mathcal{S}_3^A \cap \mathcal{S}_3^B = \{0\}$,\footnote{Note that $\dim\ker h_{\operatorname{SU}(3),\mathbf{8}} = 17$, not the expected 18, since $\mathcal{S}_2^A \cap \mathcal{S}_2^B \neq \{0\}$, but a one-dimensional vector space.} and that the space of the 3-body MPS generated by the tensor $C=A\oplus B$ is exactly the kernel of $h_{\operatorname{SU}(3),\mathbf{8}}\otimes \mathbb{1} + \mathbb{1} \otimes h_{\operatorname{SU}(3),\mathbf{8}}$. As a consequence, $C$ is block-normal, and $H_{\operatorname{SU}(3),\mathbf{8}} = \sum_i (h_{\operatorname{SU}(3),\mathbf{8}})_i$ is a parent Hamiltonian of the MPS generated by $C$. 
In particular, this implies that $H_{\operatorname{SU}(3),\mathbf{8}}$ has a $2$-fold degenerate ground state, and that it is gapped by virtue of \cref{thm:all_parent_gapped}.

Now, we look for explicit lower bounds on the spectral gap of
$H_{\operatorname{SU}(3),\mathbf{8}}$, fixing the computational cost by setting
$r \le 6$. Nachtergaele's martingale method (\cref{thm:martingale} in
\cref{appendix:fnw-n}) does not provide an explicit bound on the gap: The
relevant equation~\eqref{martingale} requires $p\ge5$ to give a bound
$\gamma_{p,q}<\tfrac{1}{2}$; to obtain an explicit bound, one would thus need to
diagonalize a 10-body Hamiltonian (with local dimension 8). In
contrast,~\cref{theorem:gamma-exact} enables the computation of the optimal
$\gamma_{p,q}$, and we obtain that $\gamma_{p,q} < \tfrac{1}{2}$ already for
$(p,q) = (2,0)$. For $r\le 6$, we then obtain the optimal bounds
\begin{equation}
    \Delta(H_{\operatorname{SU}(3),\mathbf{8}}) \geq 0.3809
\end{equation}
with $(p,q) = (2,0)$ for OBC, and
\begin{equation}
    \Delta(H_{\operatorname{SU}(3),\mathbf{8}}^{\circ}) \ge 0.3917
\end{equation}
with $(p,q) = (3,0)$ for PBC.

Comparing to the Knabe and Gosset-Mozgunov bound, we find the following: For the
Knabe bound, the finite-size gap for $r\le 6$ sites is too small for the method
to detect a gap. The Gosset-Mozgunov bound~\cite{gosset_local_2016}, applicable for
PBC, provides the following lower bound on the gap with $r=6$:
\begin{equation}
    \Delta(H_{\operatorname{SU}(3),\mathbf{8}}^{\circ}) \geq
    0.1945\ .
\end{equation}
The bound obtained by our method thus outperforms the Gosset-Mozgunov bound by a factor of $2$.

\subsubsection{The representation $\mathbf{10}$ VBS}
\label{secrep10}

Let us now turn towards an $\operatorname{SU}(3)$ VBS with representation $\mathbf{10}$ per site.
We denote the $\operatorname{SU}(3)$ singlet in the space $\mathbf{3} \otimes \mathbf{3} \otimes \mathbf{3} = \mathbf{1} \oplus 2\cdot \mathbf{8} \oplus \mathbf{10}$ by $\begin{tikzpicture}[baseline={([yshift=-.5ex]current bounding box.center)}]
    \draw (0,0) -- (1,0);
    \draw[fill=white] (0,0) circle (2pt);
    \draw[fill=white] (0.5,0) circle (2pt);
    \draw[fill=white] (1,0) circle (2pt);
\end{tikzpicture}$. We can construct a VBS state that transforms under the representation $\mathbf{10}$ by inserting translations of such a singlet on each lattice site and then projecting onto the representation $\mathbf{10}$. This state can be visualized as
\begin{equation}
\label{rep10}
    \begin{tikzpicture}[baseline={([yshift=-.5ex]current bounding box.center)}]
    \draw (0,0) -- (1,0);
    \draw[fill=white] (0,0) circle (2pt);
    \draw[fill=white] (0.5,0) circle (2pt);
    \draw[fill=white] (1,0) circle (2pt);
    \draw (1.5,0) -- (2.5,0);
    \draw[fill=white] (1.5,0) circle (2pt);
    \draw[fill=white] (2,0) circle (2pt);
    \draw[fill=white] (2.5,0) circle (2pt);
    \draw (3,0) -- (3.75,0);
    \draw[fill=white] (3,0) circle (2pt);
    \draw[fill=white] (3.5,0) circle (2pt);
    \draw (-0.6,-0.25) node {$\dots$};
    \draw (-0.25,-0.25) -- (0,-0.25);
    \draw[fill=white] (0,-0.25) circle (2pt);
    \draw (0.5,-0.25) -- (1.5,-0.25);
    \draw[fill=white] (0.5,-0.25) circle (2pt);
    \draw[fill=white] (1,-0.25) circle (2pt);
    \draw[fill=white] (1.5,-0.25) circle (2pt);
    \draw (2,-0.25) -- (3,-0.25);
    \draw[fill=white] (2,-0.25) circle (2pt);
    \draw[fill=white] (2.5,-0.25) circle (2pt);
    \draw[fill=white] (3,-0.25) circle (2pt);
    \draw (3.5,-0.25) -- (3.75,-0.25);
    \draw[fill=white] (3.5,-0.25) circle (2pt);
    \draw (4.1,-0.25) node {$\dots$};
    \draw (-0.25,-0.5) -- (0.5,-0.5);
    \draw[fill=white] (0,-0.5) circle (2pt);
    \draw[fill=white] (0.5,-0.5) circle (2pt);
    \draw (1,-0.5) -- (2,-0.5);
    \draw[fill=white] (1,-0.5) circle (2pt);
    \draw[fill=white] (1.5,-0.5) circle (2pt);
    \draw[fill=white] (2,-0.5) circle (2pt);
    \draw (2.5,-0.5) -- (3.5,-0.5);
    \draw[fill=white] (2.5,-0.5) circle (2pt);
    \draw[fill=white] (3,-0.5) circle (2pt);
    \draw[fill=white] (3.5,-0.5) circle (2pt);
    \draw[blue, fill=blue,opacity=0.2] (-0.15,-0.75) rectangle (0.15,0.25);
    \draw[blue, thick] (0,-0.75) -- (0,-1);
    \draw[blue, fill=blue,opacity=0.2] (0.35,-0.75) rectangle (0.65,0.25);
    \draw[blue, thick] (0.5,-0.75) -- (0.5,-1);
    \draw[blue, fill=blue,opacity=0.2] (0.85,-0.75) rectangle (1.15,0.25);
    \draw[blue, thick] (1,-0.75) -- (1,-1);
    \draw[blue, fill=blue,opacity=0.2] (1.35,-0.75) rectangle (1.65,0.25);
    \draw[blue, thick] (1.5,-0.75) -- (1.5,-1);
    \draw[blue, fill=blue,opacity=0.2] (1.85,-0.75) rectangle (2.15,0.25);
    \draw[blue, thick] (2,-0.75) -- (2,-1);
    \draw[blue, fill=blue,opacity=0.2] (2.35,-0.75) rectangle (2.65,0.25);
    \draw[blue, thick] (2.5,-0.75) -- (2.5,-1);
    \draw[blue, fill=blue,opacity=0.2] (2.85,-0.75) rectangle (3.15,0.25);
    \draw[blue, thick] (3,-0.75) -- (3,-1);
    \draw[blue, fill=blue,opacity=0.2] (3.35,-0.75) rectangle (3.65,0.25);
    \draw[blue, thick] (3.5,-0.75) -- (3.5,-1);
    \end{tikzpicture},
\end{equation}
where the vertical blue line denotes the projection onto $\mathbf{10}$. To write this state as an MPS, we first provide a translation invariant MPS description of the three singlet thirds found inside each blue box in \cref{rep10}:
\begin{equation}
    \begin{tikzpicture}[baseline={([yshift=-.5ex]current bounding box.center)}]
        \fill (0,0) circle (3pt);
        \draw (0,0) node[anchor=north] {$M$};
        \draw[very thick] (-0.5,0) -- (0.5,0);
        \draw (0,0) -- (0,0.5);
        \draw (-0.05,0) -- (-0.05,0.5);
        \draw (0.05,0) -- (0.05,0.5);
    \end{tikzpicture} = \sum_{a,b,c} M^{abc} \otimes \ket{abc} = \sum_{\substack{i,j,k,l,\\ a,b,c}}\delta_i^{a} \epsilon_{jk}^{b} \delta_l^{c} \ket{ij}\bra{kl} \otimes \ket{abc} = \begin{tikzpicture}[baseline={([yshift=-.5ex]current bounding box.center)}]
        \draw (0,0) -- (0.5,0);
        \draw (0,-0.125) -- (0.625,-0.125);
        \fill (0.5,0) circle (2pt);
        \draw (0.5,0) -- (0.5,0.5);
        \draw (0.625,-0.125) -- (0.625,0.5);
        \fill (0.625,-0.125) circle (2pt);
        \draw (0.75,-0.25) -- (0.75,0.5);
        \fill (0.75,-0.25) circle (2pt);
        \draw (0.75,-0.25) .. controls (0.875,-0.125) .. (1.25,-0.125);
        \draw (0.625,-0.125) .. controls (0.875,0) .. (1.25,0);
    \end{tikzpicture} \ ,
\end{equation}
where $\epsilon_{jk}^b$ is the fully antisymmetric tensor.

The projection onto the representation \textbf{10} is realized by a partial isometry denoted by $W$ (see \cref{W} for details) that matches the basis vectors of the representations $\mathbf{3} \otimes \mathbf{3} \otimes \mathbf{3}$ and $\mathbf{10}$, so the tensor generating the state given in \cref{rep10} is
\begin{equation}
    A^{i} = \sum_{a,b,c} W_{abc}^{i}M^{abc} = \begin{tikzpicture}[baseline={([yshift=-.5ex]current bounding box.center)}]
        \fill (0,0) circle (3pt);
        \draw (0,0) node[anchor=north] {$M$};
        \draw[very thick] (-0.5,0) -- (0.5,0);
        \draw (0,0) -- (0,0.5);
        \draw (-0.05,0) -- (-0.05,0.5);
        \draw (0.05,0) -- (0.05,0.5);
        \draw (-0.25,0.5) -- (0.25,0.5);
        \draw (-0.25,0.5) -- (0,0.9);
        \draw (0.25,0.5) -- (0,0.9);
        \draw (0,0.9) -- (0,1.05);
        \draw (0,0.63) node[] {\tiny $W$};
    \end{tikzpicture} \ .
\end{equation}

So far, this MPS representation has bond dimension 9, and the bonds can be understood to belong to the space $\mathbf{3} \otimes \bar{\mathbf{3}}$. The transfer operator $\mathbb{T}: \rho \mapsto \sum_i (A^{i})^\dagger \rho A^{i}$ of the MPS has 1 as its unique largest eigenvalue, but the corresponding eigenvector is not positive definite: it has a zero eigenvalue. The eigenvector $\ket{\Omega}$ belonging to that zero eigenvalue is the singlet in $\mathbf{3} \otimes \bar{\mathbf{3}}$, thus it transforms under the trivial representation $\mathbf{1}$. 
In fact, one can verify that for all $i$, $A^{i}\ket{\Omega} = 0$ -- that is, the bond of the MPS is not supported on the singlet space of $\mathbf{3}\otimes\bar{\mathbf{3}}$\footnote{This can be understood from the fact that the three $\bm{3}$-irreps have to fuse to the $\mathbf{10}$-representation, which is impossible when starting from a $1$-dimensional irrep on either the left or the right bond.} -- and thus, we can compress the MPS tensor to the $\mathbf{8}$-representation in $\mathbf{3}\otimes\bar{\mathbf{3}}$, resulting in an MPS with bond dimension $8$. As for the previous example, it can be verified that the resulting MPS tensor is normal with injectivity length $\ell =3$.

For the state constructed in~\cref{rep10}, one can construct the Hamiltonian based on \cref{eq:general_vbs}, given by the local terms~\cite{greiter_valence_2007}:
\begin{equation}
\label{vbs10ham}
    h_{\operatorname{SU}(3),\mathbf{10}} =  (\mathbf{J}_i \mathbf{J}_{i+1})^2 + 5\  \mathbf{J}_i \mathbf{J}_{i+1} + 6,
\end{equation}
where $\mathbf{J}$ now consists of the $10 \times 10$ matrices that are the generators of $\operatorname{SU}(3)$ in the 10-dimensional irreducible representation.\footnote{
The projector-valued, i.e., spectrally flattened, Hamiltonian is
\begin{equation}
    k_{\operatorname{SU}(3),\mathbf{10}} = \left((\mathbf{J}_i \mathbf{J}_{i+1})^2 + 5\  \mathbf{J}_i \mathbf{J}_{i+1} + 6 \right) \left(\frac{1}{6} - \frac{2}{45} \mathbf{J}_i\mathbf{J}_{i+1} \right),
\end{equation}
}
It is straightforward to verify that the Hamiltonian $H_{\operatorname{SU}(3),\mathbf{10}} = \sum_{i} (h_{\operatorname{SU}(3),\mathbf{10}})_i$ is a weak parent Hamiltonian (cf.~\cref{sec:weakparent}) of the previously obtained normal MPS. To prove that this is a parent Hamiltonian, it is straightforward to check that the dimension of the Hamiltonian's kernel is indeed $D^2 = 64$ for the system size $\ell+1 = 4$. Consequently,  $H_{\operatorname{SU}(3),\mathbf{10}}$ has a unique ground state and is gapped by~\cref{thm:all_parent_gapped}.

Let us now look for uniform lower bounds on the spectral gap, fixing the
computational cost to diagonalizing a Hamiltonian of up to $r\le 6$ sites. We
find that to obtain an explicit bound using the original FNW method  (\cref{thm:FNW92}
in \cref{sec:fnw}), we would need to diagonalize a $16$-site Hamiltonian, which is prohibitive.
In contrast, the optimal $\gamma_{p,q}$, calculated
using~\cref{theorem:gamma-exact}, gets below the $\tfrac{1}{2}$ threshold -- and
thus, one can obtain a quantitative bound on the gap -- already for $(p,q)
=(2,0)$, i.e., $r=4$. By applying~\cref{theorem:gap-from-gamma} we find that
choosing larger block sizes gives tighter lower bounds for the spectral gap.
Specifically,  we obtain
\begin{equation}
    \Delta(H_{\operatorname{SU}(3),\mathbf{10}}) \geq 0.2158
\end{equation}
with $(p,q) = (3,0)$ for OBC and
\begin{equation}
    \Delta(H_{\operatorname{SU}(3),\mathbf{10}}^\circ) \geq 0.2204
\end{equation}
with $(p,q) =(3,0)$ for PBC.
In contrast, both for the Knabe and the Gosset-Mozgunov method, the finite-size gaps for $r\le 6$ are too small to yield a non-zero uniform lower bound on the spectral gap.
Our method is thus the only strategy which can provide an explicit bound on the gap of this model.

\section{Conclusion}
In this article, we have provided a technique to compute rigorous lower bounds
on the  spectral gaps of parent Hamiltonians of normal and block-normal Matrix
Product States (MPS), building on the method of Fannes, Nachtergaele and Werner (FNW) and improving it in several
places. The key improvement is an efficient algorithm to compute the central
quantity $\gamma_{p,q}$ of the method -- measuring the angle between 
ground spaces on overlapping
blocks -- exactly, at a cost which only scales linearly with the block size.
This enables a significant improvement over both the original FNW method,
which only gave a bound on $\gamma_{p,q}$, and over direct computation of
the exact $\gamma_{p,q}$, which scales exponentially in the size of
the overlapping blocks.  The computational bottleneck then becomes computing the gap
of a suitable blocked Hamiltonian on a single block, which scales exponentially in the
size of the block, rather than of two overlapping ones. The ability to compute
$\gamma_{p,q}$ exactly \emph{and} efficiently allows to reach
significantly larger block sizes, and thus to obtain better bounds on the gap, or,
in some cases, to obtain a non-trivial
bound at all.  We demonstrated the ability of our method to outperform other
methods through a number of examples: We observed that it performs better
than finite-size methods such as the Knabe and the Gosset-Mozgunov method, in particular for
high-dimensional systems, including $\mathrm{SU}(3)$ valence bond models, or for
strongly deformed models with large correlation lengths. As a side result, we
found that our numerically motivated method to compute $\gamma_{p,q}$ also
gave a much simplified proof that $\gamma_{p,q}\to0$, which implies the
existence of a spectral gap for all parent Hamiltonians of normal and block-normal MPS.

Two natural extensions of the method, to be explored in future work, are its
application to general frustration free Hamiltonians, where a (non-translation invariant) MPS
describing the ground space is generated dynamically from the Hamiltonian, and
the extension of the reshuffling in \cref{lem:block-obc} and~\ref{lem:block-pbc}
to a telescopic sum of arbitrary operators.

\section*{Author contributions}
All authors contributed to the results and the writing. M.Á.R. performed the numerical calculations. AI tools have been used for the final proofreading.

\section*{Acknowledgements}
We thank Ilya Kull for useful discussions. This research was funded by the Austrian Science Fund (FWF) [Grant Numbers 
\href{https://doi.org/10.55776/COE1}{10.55776/COE1},
\href{https://doi.org/10.55776/F71}{10.55776/F71},
and \href{https://doi.org/10.55776/P36305}{10.55776/P36305}] the European Union through NextGenerationEU, and the European Research Council (ERC) under the European Union’s Horizon 2020 research and innovation programme (grant agreement No.~863476). For open access purposes, the authors have applied a CC BY public copyright license to any author accepted manuscript version arising from this submission.

\bibliographystyle{quantum.bst}
\bibliography{references}

\appendix
\crefalias{section}{appsec}
\crefalias{subsection}{appsec}
\section{Additional details for the proof of~\cref{theorem:gamma-exact}}
\label{appendix:exact-gamma-details}
\cref{theorem:gamma-exact} provides an efficient and optimal solution to the inequality~\cref{eq:PQ-QP-condition} for parent Hamiltonians of normal and block-normal MPS. In this section, we provide a more detailed derivation of the equality $\gamma_{p,q} = \sqrt{\lambda_2(\hat{\Xi}})$ given in its proof.

We first state Jordan's lemma \cite{bhatia_matrix_1997}: let $Q$ and $R$ be orthogonal projectors on a finite-dimensional Hilbert-space $\mathcal{H}$ and let $0 < \theta_1 \leq \dots \leq \theta_k < \pi/2$ be the principal angles between the subspaces $\operatorname{Im}(Q)$ and $\operatorname{Im}(R)$. Then there exists an orthonormal basis of $\mathcal{H}$ in which $Q$ and $R$ are simultaneously block-diagonal, and the blocks are either $2 \times 2$:
\begin{equation}
    \label{jordaneq2}
        Q^{(i)} = \begin{pmatrix}
            1 & 0 \\ 0 & 0
        \end{pmatrix}, \quad R^{(i)} = \begin{pmatrix}
            \cos^2\theta_i & \cos \theta_i \sin \theta_i \\ \cos \theta_i \sin \theta_i & \sin^2 \theta_i
        \end{pmatrix} \quad (i \in \{1,\dots,k\}),
\end{equation}
or $1 \times 1$ of the form
\begin{equation}
    \label{jordaneq1}
        \mathfrak{q}^{(i)},\mathfrak{r}^{(i)} \in \{(0),(1)\} \quad (i \in \{k+1,\dots,n-k\}).
    \end{equation}

The power of this result is the possibility to determine polynomials of $Q$ and $R$ simply by reducing the problem to $1 \times 1$ and $2 \times 2$ blocks. 

By writing out the projection $\Pi$ explicitly, \cref{eq:PQ-QP-condition} becomes
\begin{equation}
\label{prineq}
    \{\Pi_{2p+q}^{\perp}\otimes \mathbb{1}_{p+q}, \mathbb{1}_{p+q} \otimes \Pi_{2p+q}^{\perp}\} \geq - \gamma_{p,q} (\Pi_{2p+q}^{\perp}\otimes \mathbb{1}_{p+q} + \mathbb{1}_{p+q} \otimes \Pi_{2p+q}^{\perp}).
\end{equation}

We now compute the smallest nonnegative number $\gamma_{p,q}$ that solves~\cref{eq:PQ-QP-condition} for $Q = \Pi_{2p+q}^{\perp} \otimes \mathbb{1}_{p+q}$ and $R = \mathbb{1}_{p+q} \otimes \Pi_{2p+q}^{\perp}$ by applying Jordan's lemma. Due to the block-diagonal structure of $Q$ and $R$, we set $\gamma_{p,q} = \max_i \gamma_{p,q}^{(i)}$, where $\gamma_{p,q}^{(i)}$ is the smallest nonnegative number satisfying the inequality in the $1 \times 1$ or $2 \times 2$ block indexed by $i$. Note that choosing any $\gamma_{p,q}$ larger than $\max_i \gamma_{p,q}^{(i)}$ would automatically solve~\cref{eq:PQ-QP-condition}, due to $\gamma_{p,q}$ being a prefactor of a positive semidefinite term. Let us now determine the values of $\gamma_{p,q}^{(i)}$.

Since the $1\times 1$ blocks $\mathfrak{q}^{(i)}$ and $\mathfrak{r}^{(i)}$ are either 0 or 1, we have $\mathfrak{q}^{(i)}\mathfrak{r}^{(i)} + \mathfrak{r}^{(i)}\mathfrak{q}^{(i)} \geq 0$, so we have $\gamma_{p,q}^{(i)} = 0$ for all $i \in \{k+1,\dots,n-k\}$. On the level of $2 \times 2$ blocks $Q^{(i)}$ and $R^{(i)}$, \cref{eq:PQ-QP-condition} reduces to the system of inequalities
\begin{equation}
    Q^{(i)}\cdot R^{(i)} + R^{(i)} \cdot Q^{(i)} + \gamma_{p,q}^{(i)} (Q^{(i)}+R^{(i)}) \geq 0
\end{equation}
for $i \in \{1,\dots,k\}$. Any of these inequalities can be solved in a straightforward way by diagonalizing the $2 \times 2$ matrix on the left-hand side, leading us to
\begin{equation}
\label{gammasineq}
    \gamma_{p,q}^{(i)} = \cos \theta_i \quad \forall i \in \{1,\dots,k\}.
\end{equation}
Since the eigenvalues of each $Q^{(i)} \cdot R^{(i)}$ are $\{\cos^2\theta_i, 0\}$, the value of $\gamma_{p,q}^{(i)}$ is the square root of the nonzero eigenvalue of $Q^{(i)}\cdot R^{(i)}$ for any $i \in \{1,\dots,k\}$ and thus $\gamma_{p,q}$ is the square root of the largest eigenvalue of $Q \cdot R$ in the interval $(0,1)$.

Exchanging $Q$ with $\mathbb{1}-Q$ and $R$ with $\mathbb{1}-R$ in the block-diagonal form provides the same system of inequalities as in \cref{gammasineq}, since it only flips the diagonal entries of each $Q^{(i)}$ and $R^{(i)}$. Consequently, it is possible to work directly with ground space projectors, and the value of $\gamma_{p,q}$ remains unchanged if we exchange $\Pi_{2p+q}^{\perp}$ with $\Pi_{2p+q}$ in~\cref{prineq}. In addition, the largest eigenvalue of $(\Pi_{2p+q}\otimes \mathbb{1}_{p+q})(\mathbb{1}_{p+q} \otimes \Pi_{2p+q})$ is 1, since it leaves the $(3p+2q)$-site ground states invariant, so $\gamma_{p,q}$ is the square root of the largest eigenvalue of $(\Pi_{2p+q}\otimes \mathbb{1}_{p+q})(\mathbb{1}_{p+q} \otimes \Pi_{2p+q}) = \hat \Xi$ distinct from 1.

\section{Overview of the results of Fannes-Nachtergaele-Werner (FNW) and Nachtergaele (martingale method)}
\label{appendix:fnw-n}
In this appendix, we provide an overview of the original results of Fannes-Nachtergaele-Werner (FNW) and of Nachtergaele in the terminology of MPS, and provide a self-contained derivation of the FNW result.

\subsection{Normal MPS and the result of Fannes-Nachtergaele-Werner}\label{sec:fnw}

In the context of parent Hamiltonians of normal MPS, \cref{theorem:gap-from-gamma} was first introduced in~\cite{fannes_finitely_1992}, with the specific choice $\tilde{h}_{\tilde{\imath}}=\tfrac{H_{2p}}{2}$. In order to prove that parent Hamiltonians of normal MPS are gapped, they provided a $p$-dependent upper bound on the value of $\gamma_{p,q}$ for any choice of $q\ge 0$. In the following, we restate this upper bound and reformulate its proof by Jordan's lemma and the graphical notation natural for MPS.
\begin{theorem}[FNW~\cite{fannes_finitely_1992}]\label{thm:FNW92}
    For parent Hamiltonians of normal MPS with injectivity length $\ell$, let $\rho$ be the unique fixed point of $\mathcal T$ with $\Tr(\rho)=1$. We set $a(p):=\Tr(\rho^{-1})\|\mathcal{T}^p-\mathcal{T}_\infty\|$, then we have for all $q \ge 0$ and $p \ge \ell$ large enough such that $a(p) <1$:
    \begin{equation}
    \label{fnw}
    \gamma_{p,q} \leq a(p) \frac{1+a(p)}{1-a(p)} 
\ .
\end{equation}
If $\mathcal T$ is a normal matrix, it can be unitarily
diagonalized, and  $\rho=\frac1D\mathbb{1}$, hence one can choose $a(p)\le D^2
|\lambda_2(\mathcal{T})|^p$.
\end{theorem}

Since $a(p)\to 0$, we have that $\gamma_{p,q}\to0$ as $p\to\infty$, and
thus it will always give a non-trivial lower bound on the gap in
\cref{theorem:gap-from-gamma} for large enough values of $p$. This bound provides a different way of proving that parent Hamiltonians of normal MPS are always gapped.

\noindent\emph{Proof.---} We first prove that
\begin{equation}
\label{gammanorm}
\begin{aligned}
     \{\Pi_{2p+q}^{\perp} & \otimes \mathbb{1}_{p+q}, \mathbb{1}_{p+q} \otimes \Pi_{2p+q}^{\perp}\} \\ &\geq - \|(\Pi_{2p+q}\otimes \mathbb{1}_{p+q})(\mathbb{1}_{p+q} \otimes \Pi_{2p+q})-\Pi_{3p+2q}\|\cdot(\Pi_{2p+q}^{\perp}\otimes \mathbb{1}_{p+q} + \mathbb{1}_{p+q} \otimes \Pi_{2p+q}^{\perp}),
\end{aligned}
\end{equation}
more precisely, that $\gamma_{p,q} =\|(\Pi_{2p+q}\otimes \mathbb{1}_{p+q})(\mathbb{1}_{p+q} \otimes \Pi_{2p+q})-\Pi_{3p+2q}\|$. As seen in~\cref{appendix:exact-gamma-details}, we can simultaneously block-diagonalize the orthogonal projectors $(\Pi_{2p+q}\otimes \mathbb{1}_{p+q})$ and $(\mathbb{1}_{p+q} \otimes \Pi_{2p+q})$ due to Jordan's lemma. On the level of $1\times 1$ blocks, the eigenvalues of the anticommutator are nonnegative, so the inequality is trivially satisfied. Due to the intersection property of normal MPS~\cite{schuch_peps_2010}, $\Pi_{3p+2q}$ is zero on the $2 \times 2$ blocks, and direct calculation shows that~\cref{gammanorm} reduces to a system of inequalities
\begin{equation}
    \cos\theta_i (1\pm \cos\theta_i) \geq \cos\theta_1(\cos\theta_i \mp 1), 
\end{equation}
which is satisfied for $\theta_i \geq \theta_1$, and thus the claim is proven.

By setting $\mathcal{Q} \coloneq \operatorname{Im}(\Pi_{2p+q}\otimes \mathbb{1}_{p+q}) \cap \operatorname{Im}(\Pi_{3p+2q})^{\perp}$ and $\mathcal{R} \coloneq \operatorname{Im}(\mathbb{1}_{p+q} \otimes \Pi_{2p+q}) \cap \operatorname{Im}(\Pi_{3p+2q})^{\perp}$, we can write
\begin{equation}
\gamma_{p,q} = \sup_{\substack{\ket{\Psi} \in \mathcal{R} \setminus \{0\} \\ \ket{\Phi} \in  \mathcal{Q} \setminus \{0\}}} \frac{|\braket{\Phi}{\Psi}|}{ \|\ket{\Phi}\| \cdot \|\ket{\Psi}\|}.
\end{equation}
As $\ket{\Phi} \in \mathcal{Q} \setminus \{0\}$ and $\ket{\Psi} \in \mathcal{R} \setminus \{0\}$, in graphical notation (cf.~\cref{eq:blocked_mps}) we have
\begin{equation}
    \ket{\Phi} = \begin{tikzpicture}[baseline={([yshift=-.5ex]current bounding box.center)}]
        \draw (0,0) -- (1.6,0);
        \draw (0,-0.8) -- (0,0);
        \draw (1.6,-0.8) -- (1.6,0);
        \draw (0,-0.8) -- (1.6,-0.8);
        \draw (0.4,0) -- (0.4,0.4);
        \draw (1,0) -- (1,0.4);
        \draw (1.6,0) -- (1.6,0.4);
        \fill (1.6,-0.4) ellipse (2pt and 14pt);
        \draw (1.6,-0.4) node[anchor=west]{$\hat{\Phi}$};
        \draw[fill=white] (0.4,0) circle (2pt);
        \draw[fill=white] (1,0) circle (2pt);
        \fill (1,0) circle (0.75pt);
        \fill (0.4,0) circle (1.25pt);
    \end{tikzpicture}, \qquad \ket{\Psi} = \begin{tikzpicture}[baseline={([yshift=-.5ex]current bounding box.center)}]
        \draw (0,0) -- (1.6,0);
        \draw (0,-0.8) -- (0,0);
        \draw (1.6,-0.8) -- (1.6,0);
        \draw (0,-0.8) -- (1.6,-0.8);
        \draw (0.6,0) -- (0.6,0.4);
        \draw (1.2,0) -- (1.2,0.4);
        \draw (0,0) -- (0,0.4);
        \fill (0,-0.4) ellipse (2pt and 14pt);
        \draw (0,-0.4) node[anchor=east]{$\hat{\Psi}$};
        \draw[fill=white] (0.6,0) circle (2pt);
        \draw[fill=white] (1.2,0) circle (2pt);
        \fill (0.6,0) circle (0.75pt);
        \fill (1.2,0) circle (1.25pt);
    \end{tikzpicture} \, ,
\end{equation}
where $\hat{\Phi}$ and $\hat{\Psi}$ are $d^{p+q}$-tuples of $D\times D$ matrices with the additional property that $\Pi_{3p+2q}\ket{\Phi} = \Pi_{3p+2q}\ket{\Psi} =0$, i.e., $\ket{\Phi}$ and $\ket{\Psi}$ are orthogonal to the matrix product ground state $\ket{\Psi_{3p+2q}(A,Z)}$ for any boundary $Z \in \mathcal{M}_D$:
\begin{equation}
\label{orth}
    \begin{tikzpicture}[baseline={([yshift=-.5ex]current bounding box.center)}]
        \draw (0,0) -- (1.6,0);
        \draw (0,-0.8) -- (0,0);
        \draw (1.6,-0.8) -- (1.6,0);
        \draw (0,-0.8) -- (1.6,-0.8);
        \draw (0.4,0) -- (0.4,0.4);
        \draw (1,0) -- (1,0.4);
        \draw (1.6,0) -- (1.6,0.4);
        \fill (1.6,-0.4) ellipse (2pt and 14pt);
        \draw (1.6,-0.4) node[anchor=west]{$\hat{\Phi}$};
        \draw (0.2,0.4) rectangle (1.8,1.2);
        \fill (1,1.2) circle (2pt);
        \draw (1,1.2) node[anchor=south]{$\bar{Z}$};
        \draw[fill=white] (0.4,0) circle (2pt);
        \draw[fill=white] (1,0) circle (2pt);
        \draw[fill=white] (1.6,0.4) circle (2pt);
        \draw[fill=white] (0.4,0.4) circle (2pt);
        \draw[fill=white] (1,0.4) circle (2pt);
        \fill (0.4,0) circle (1.25pt);
        \fill (1,0) circle (0.75pt);
        \fill (1.6,0.4) circle (1.25pt);
        \fill (0.4,0.4) circle (1.25pt);
        \fill (1,0.4) circle (0.75pt);
    \end{tikzpicture}
    = \begin{tikzpicture}[baseline={([yshift=-.5ex]current bounding box.center)}]
        \draw (0,0) -- (1.6,0);
        \draw (0,-0.8) -- (0,0);
        \draw (1.6,-0.8) -- (1.6,0);
        \draw (0,-0.8) -- (1.6,-0.8);
        \draw (0.6,0) -- (0.6,0.4);
        \draw (1.2,0) -- (1.2,0.4);
        \draw (0,0) -- (0,0.4);
        \fill (0,-0.4) ellipse (2pt and 14pt);
        \draw (0,-0.4) node[anchor=east]{$\hat{\Psi}$};
        \fill (0.6,1.2) circle (2pt);
        \draw (-0.2,0.4) rectangle (1.4,1.2);
        \draw (0.6,1.2) node[anchor=south] {$\bar{Z}$};
        \draw[fill=white] (0.6,0) circle (2pt);
        \draw[fill=white] (1.2,0) circle (2pt);
        \draw[fill=white] (0,0.4) circle (2pt);
        \draw[fill=white] (0.6,0.4) circle (2pt);
        \draw[fill=white] (1.2,0.4) circle (2pt);
        \fill (0.6,0) circle (0.75pt);
        \fill (1.2,0) circle (1.25pt);
        \fill (0,0.4) circle (1.25pt);
        \fill (0.6,0.4) circle (0.75pt);
        \fill (1.2,0.4) circle (1.25pt);
    \end{tikzpicture} = 0.
\end{equation}

We assume that the MPS tensor is in the canonical form given in~\cref{canonical} with the right fixed point of the transfer matrix denoted by $\rho$. On the space of $D \times D$ matrices we define $\langle X,Y\rangle_\rho = \Tr (\rho X^{\dagger}Y)$ and denote its induced norm by $\|.\|_\rho$. We introduce the Schatten $p$-norms given by $\|X\|_p = \left(\Tr\sqrt{X^{\dagger}X}^{\ p}\right)^{1/p}$ for $p \in [1,\infty]$ and recall Hölder's inequality which states that for $x,y,z\in[1,\infty]$ satisfying
\begin{equation}
    \frac{1}{x} + \frac{1}{y} = \frac{1}{z}
\end{equation}
we have
\begin{equation}
    \|XY\|_z \leq \|X\|_x \|Y\|_y\ .
\end{equation}
The operator norm corresponds to the Schatten $\infty$-norm, and we furthermore have $\|X\|_\rho = \|\sqrt{\rho} X\|_2$. Throughout the proof, we will repeatedly use the fact that the Schatten 2-norm of a tensor is invariant under arbitrary permutations of its indices.

To find an upper bound on $|\braket{\Phi}{\Psi}|$, we first use the triangle-inequality:
\begin{equation}
\label{fnw_subresult0}
    |\braket{\Phi}{\Psi}| \leq \left\vert \begin{tikzpicture}[baseline={([yshift=-.5ex]current bounding box.center)}]
        \draw (0,0) -- (1.6,0);
        \draw (0,-0.8) -- (0,0);
        \draw (1.6,-0.8) -- (1.6,0);
        \draw (0,-0.8) -- (1.6,-0.8);
        \draw (0.6,0) -- (0.6,0.4);
        \draw (1.2,0) -- (1.2,0.4);
        \draw (0,0) -- (0,0.4);
        \fill (0,-0.4) ellipse (2pt and 14pt);
        \draw (0,-0.4) node[anchor=east]{$\hat{\Psi}$};
        \draw (-0.4,0.4) -- (1.2,0.4);
        \draw (-0.4,0.4) -- (-0.4,1.2);
        \draw (1.2,0.4) -- (1.2,1.2);
        \draw (-0.4,1.2) -- (1.2,1.2);
        \fill (1.2,0.8) ellipse (2pt and 14pt);
        \draw (1.2,0.8) node[anchor=west]{$\bar{\hat{\Phi}}$};
        \draw[fill=white] (0.6,0) circle (2pt);
        \draw[fill=white] (0.6,0.4) circle (2pt);
        \draw[fill=white] (0,0.4) circle (2pt);
        \draw[fill=white] (1.2,0) circle (2pt);
        \fill (0.6,0) circle (0.75pt);
        \fill (0.6,0.4) circle (0.75pt);
        \fill (0.0,0.4) circle (1.25pt);
        \fill (1.2,0) circle (1.25pt);
    \end{tikzpicture} - \begin{tikzpicture}[baseline={([yshift=-.5ex]current bounding box.center)}]
        \draw (0,0) -- (0.5,0);
        \draw (0.5,0) -- (0.5,0.4);
        \draw (0.5,0.2) node[anchor=east]{$\rho$};
        \draw (0.7,0) -- (0.7,0.4);
        \draw (0.7,0) -- (1.6,0);
        \draw (0,-0.8) -- (0,0);
        \draw (1.6,-0.8) -- (1.6,0);
        \draw (0,-0.8) -- (1.6,-0.8);
        \draw (1.2,0) -- (1.2,0.4);
        \draw (0,0) -- (0,0.4);
        \fill (0,-0.4) ellipse (2pt and 14pt);
        \draw (0,-0.4) node[anchor=east]{$\hat{\Psi}$};
        \draw (-0.4,0.4) -- (0.5,0.4);
        \draw (0.7,0.4) -- (1.2,0.4);
        \draw (-0.4,0.4) -- (-0.4,1.2);
        \draw (1.2,0.4) -- (1.2,1.2);
        \draw (-0.4,1.2) -- (1.2,1.2);
        \fill (1.2,0.8) ellipse (2pt and 14pt);
        \draw (1.2,0.8) node[anchor=west]{$\bar{\hat{\Phi}}$};
        \draw[fill=white] (1.2,0) circle (2pt);
        \draw[fill=white] (0,0.4) circle (2pt);
        \fill (0.5,0.2) circle (2pt);
        \fill (1.2,0) circle (1.25pt);
        \fill (0.0,0.4) circle (1.25pt);
    \end{tikzpicture} \right\vert + \left\vert \begin{tikzpicture}[baseline={([yshift=-.5ex]current bounding box.center)}]
        \draw (0,0) -- (0.5,0);
        \draw (0.5,0) -- (0.5,0.4);
        \draw (0.5,0.2) node[anchor=east]{$\rho$};
        \draw (0.7,0) -- (0.7,0.4);
        \draw (0.7,0) -- (1.6,0);
        \draw (0,-0.8) -- (0,0);
        \draw (1.6,-0.8) -- (1.6,0);
        \draw (0,-0.8) -- (1.6,-0.8);
        \draw (1.2,0) -- (1.2,0.4);
        \draw (0,0) -- (0,0.4);
        \fill (0,-0.4) ellipse (2pt and 14pt);
        \draw (0,-0.4) node[anchor=east]{$\hat{\Psi}$};
        \draw (-0.4,0.4) -- (0.5,0.4);
        \draw (0.7,0.4) -- (1.2,0.4);
        \draw (-0.4,0.4) -- (-0.4,1.2);
        \draw (1.2,0.4) -- (1.2,1.2);
        \draw (-0.4,1.2) -- (1.2,1.2);
        \fill (1.2,0.8) ellipse (2pt and 14pt);
        \draw (1.2,0.8) node[anchor=west]{$\bar{\hat{\Phi}}$};
        \draw[fill=white] (1.2,0) circle (2pt);
        \draw[fill=white] (0,0.4) circle (2pt);
        \fill (0.5,0.2) circle (2pt);
        \fill (1.2,0) circle (1.25pt);
        \fill (0.0,0.4) circle (1.25pt);
    \end{tikzpicture} \right\vert.
\end{equation}
Since the transfer matrix of a normal MPS in the canonical form satisfies~\cref{transferlimit}, the first expression of the right-hand side can be understood as a composition with the map $(\mathcal{T}^p-\mathcal{T}_{\infty})$, so by $|\operatorname{Tr}(XY)| \leq \|X\| \cdot \|Y\|_1$ we obtain:
\begin{equation}
    \left\vert \begin{tikzpicture}[baseline={([yshift=-.5ex]current bounding box.center)}]
        \draw (0,0) -- (1.6,0);
        \draw (0,-0.8) -- (0,0);
        \draw (1.6,-0.8) -- (1.6,0);
        \draw (0,-0.8) -- (1.6,-0.8);
        \draw (0.6,0) -- (0.6,0.4);
        \draw (1.2,0) -- (1.2,0.4);
        \draw (0,0) -- (0,0.4);
        \fill (0,-0.4) ellipse (2pt and 14pt);
        \draw (0,-0.4) node[anchor=east]{$\hat{\Psi}$};
        \draw (-0.4,0.4) -- (1.2,0.4);
        \draw (-0.4,0.4) -- (-0.4,1.2);
        \draw (1.2,0.4) -- (1.2,1.2);
        \draw (-0.4,1.2) -- (1.2,1.2);
        \fill (1.2,0.8) ellipse (2pt and 14pt);
        \draw (1.2,0.8) node[anchor=west]{$\bar{\hat{\Phi}}$};
        \draw[fill=white] (0.6,0) circle (2pt);
        \draw[fill=white] (0.6,0.4) circle (2pt);
        \draw[fill=white] (0,0.4) circle (2pt);
        \draw[fill=white] (1.2,0) circle (2pt);
        \fill (0.6,0) circle (0.75pt);
        \fill (0.6,0.4) circle (0.75pt);
        \fill (0.0,0.4) circle (1.25pt);
        \fill (1.2,0) circle (1.25pt);
    \end{tikzpicture} - \begin{tikzpicture}[baseline={([yshift=-.5ex]current bounding box.center)}]
        \draw (0,0) -- (0.5,0);
        \draw (0.5,0) -- (0.5,0.4);
        \draw (0.5,0.2) node[anchor=east]{$\rho$};
        \draw (0.7,0) -- (0.7,0.4);
        \draw (0.7,0) -- (1.6,0);
        \draw (0,-0.8) -- (0,0);
        \draw (1.6,-0.8) -- (1.6,0);
        \draw (0,-0.8) -- (1.6,-0.8);
        \draw (1.2,0) -- (1.2,0.4);
        \draw (0,0) -- (0,0.4);
        \fill (0,-0.4) ellipse (2pt and 14pt);
        \draw (0,-0.4) node[anchor=east]{$\hat{\Psi}$};
        \draw (-0.4,0.4) -- (0.5,0.4);
        \draw (0.7,0.4) -- (1.2,0.4);
        \draw (-0.4,0.4) -- (-0.4,1.2);
        \draw (1.2,0.4) -- (1.2,1.2);
        \draw (-0.4,1.2) -- (1.2,1.2);
        \fill (1.2,0.8) ellipse (2pt and 14pt);
        \draw (1.2,0.8) node[anchor=west]{$\bar{\hat{\Phi}}$};
        \draw[fill=white] (1.2,0) circle (2pt);
        \draw[fill=white] (0,0.4) circle (2pt);
        \fill (0.5,0.2) circle (2pt);
        \fill (1.2,0) circle (1.25pt);
        \fill (0.0,0.4) circle (1.25pt);
    \end{tikzpicture} \right\vert \leq \|\mathcal{T}^p-\mathcal{T}_{\infty}\| \cdot \left\| \begin{tikzpicture}[baseline={([yshift=-.5ex]current bounding box.center)}]
        \draw (0,0) -- (0.4,0);
        \draw (0.8,0) -- (1.6,0);
        \draw (0,-0.8) -- (0,0);
        \draw (1.6,-0.8) -- (1.6,0);
        \draw (0,-0.8) -- (1.6,-0.8);
        \draw (1.2,0) -- (1.2,0.4);
        \draw (0,0) -- (0,0.4);
        \fill (0,-0.4) ellipse (2pt and 14pt);
        \draw (0,-0.4) node[anchor=east]{$\hat{\Psi}$};
        \draw (-0.4,0.4) -- (0.4,0.4);
        \draw (0.8,0.4) -- (1.2,0.4);
        \draw (-0.4,0.4) -- (-0.4,1.2);
        \draw (1.2,0.4) -- (1.2,1.2);
        \draw (-0.4,1.2) -- (1.2,1.2);
        \fill (1.2,0.8) ellipse (2pt and 14pt);
        \draw (1.2,0.8) node[anchor=west]{$\bar{\hat{\Phi}}$};
        \draw[fill=white] (0,0.4) circle (2pt);
        \draw[fill=white] (1.2,0) circle (2pt);
        \fill (0.0,0.4) circle (1.25pt);
        \fill (1.2,0) circle (1.25pt);
    \end{tikzpicture}  \right\|_1 .
\end{equation}
The tensor on the right-hand side is viewed as an operator from the left indices to the right indices. By putting $\rho^{1/2} \cdot \rho^{-1/2}$ on both input indices, we can use Hölder's inequality:
\begin{equation}
    \left\| \begin{tikzpicture}[baseline={([yshift=-.5ex]current bounding box.center)}]
        \draw (0,0) -- (0.4,0);
        \draw (0.8,0) -- (1.6,0);
        \draw (0,-0.8) -- (0,0);
        \draw (1.6,-0.8) -- (1.6,0);
        \draw (0,-0.8) -- (1.6,-0.8);
        \draw (1.2,0) -- (1.2,0.4);
        \draw (0,0) -- (0,0.4);
        \fill (0,-0.4) ellipse (2pt and 14pt);
        \draw (0,-0.4) node[anchor=east]{$\hat{\Psi}$};
        \draw (-0.4,0.4) -- (0.4,0.4);
        \draw (0.8,0.4) -- (1.2,0.4);
        \draw (-0.4,0.4) -- (-0.4,1.2);
        \draw (1.2,0.4) -- (1.2,1.2);
        \draw (-0.4,1.2) -- (1.2,1.2);
        \fill (1.2,0.8) ellipse (2pt and 14pt);
        \draw (1.2,0.8) node[anchor=west]{$\bar{\hat{\Phi}}$};
        \draw[fill=white] (0,0.4) circle (2pt);
        \draw[fill=white] (1.2,0) circle (2pt);
        \fill (0.0,0.4) circle (1.25pt);
        \fill (1.2,0) circle (1.25pt);
    \end{tikzpicture}  \right\|_1 \leq \|\rho^{-1/2} \otimes \rho^{-1/2}\|_2 \cdot \left\| \begin{tikzpicture}[baseline={([yshift=-.5ex]current bounding box.center)}]
        \draw (0,0) -- (0.6,0);
        \draw (0.9,0) -- (1.6,0);
        \draw (0,-0.8) -- (0,0);
        \draw (1.6,-0.8) -- (1.6,0);
        \draw (0,-0.8) -- (1.6,-0.8);
        \draw (1.2,0) -- (1.2,0.4);
        \draw (0,0) -- (0,0.4);
        \fill (0,-0.4) ellipse (2pt and 14pt);
        \draw (0,-0.4) node[anchor=east]{$\hat{\Psi}$};
        \draw (-0.4,0.4) -- (0.6,0.4);
        \draw (0.9,0.4) -- (1.2,0.4);
        \draw (-0.4,0.4) -- (-0.4,1.2);
        \draw (1.2,0.4) -- (1.2,1.2);
        \draw (-0.4,1.2) -- (1.2,1.2);
        \fill (1.2,0.8) ellipse (2pt and 14pt);
        \draw (1.2,0.8) node[anchor=west]{$\bar{\hat{\Phi}}$};
        \fill (0.3,0) circle (2pt);
        \fill (0.3,0.4) circle (2pt);
        \draw (0.3,0) node[anchor=north] {\tiny $\sqrt{\rho}$};
        \draw (0.3,0.4) node[anchor=south] {\tiny $\sqrt{\rho}$};
        \draw[fill=white] (0,0.4) circle (2pt);
        \draw[fill=white] (1.2,0) circle (2pt);
        \fill (0.0,0.4) circle (1.25pt);
        \fill (1.2,0) circle (1.25pt);
    \end{tikzpicture}  \right\|_2.
\end{equation}
Due to the 2-norm being invariant under index permutations, we can reinterpret the tensor on the right as a linear map from the bottom two to the top two indices. Since the Schatten norms are submultiplicative, we have
\begin{equation}
    \left\| \begin{tikzpicture}[baseline={([yshift=-.5ex]current bounding box.center)}]
        \draw (0,0) -- (0.6,0);
        \draw (0.9,0) -- (1.6,0);
        \draw (0,-0.8) -- (0,0);
        \draw (1.6,-0.8) -- (1.6,0);
        \draw (0,-0.8) -- (1.6,-0.8);
        \draw (1.2,0) -- (1.2,0.4);
        \draw (0,0) -- (0,0.4);
        \fill (0,-0.4) ellipse (2pt and 14pt);
        \draw (0,-0.4) node[anchor=east]{$\hat{\Psi}$};
        \draw (-0.4,0.4) -- (0.6,0.4);
        \draw (0.9,0.4) -- (1.2,0.4);
        \draw (-0.4,0.4) -- (-0.4,1.2);
        \draw (1.2,0.4) -- (1.2,1.2);
        \draw (-0.4,1.2) -- (1.2,1.2);
        \fill (1.2,0.8) ellipse (2pt and 14pt);
        \draw (1.2,0.8) node[anchor=west]{$\bar{\hat{\Phi}}$};
        \fill (0.3,0) circle (2pt);
        \fill (0.3,0.4) circle (2pt);
        \draw (0.3,0) node[anchor=north] {\tiny $\sqrt{\rho}$};
        \draw (0.3,0.4) node[anchor=south] {\tiny $\sqrt{\rho}$};
        \draw[fill=white] (0,0.4) circle (2pt);
        \draw[fill=white] (1.2,0) circle (2pt);
        \fill (0.0,0.4) circle (1.25pt);
        \fill (1.2,0) circle (1.25pt);
    \end{tikzpicture}  \right\|_2 \leq \left\| \begin{tikzpicture}[baseline={([yshift=-.5ex]current bounding box.center)}]
        \draw (0,0) -- (0.6,0);
        \draw (0.6,0) -- (0.6,0.4);
        \draw (0.9,0) -- (1.6,0);
        \draw (0.9,0) -- (0.9,0.4);
        \draw (0,-0.8) -- (0,0);
        \draw (1.6,-0.8) -- (1.6,0);
        \draw (0,-0.8) -- (1.6,-0.8);
        \draw (1.2,0) -- (1.2,0.4);
        \draw (0,0) -- (0,0.4);
        \fill (0,-0.4) ellipse (2pt and 14pt);
        \draw (0,-0.4) node[anchor=east]{$\hat{\Psi}$};
        \fill (0.6,0.2) circle (2pt);
        \draw (0.6,0.2) node[anchor=east] {\scriptsize $\sqrt{\rho}$};
        \draw[fill=white] (1.2,0) circle (2pt);
        \fill (1.2,0) circle (1.25pt);
    \end{tikzpicture} \ \right\|_2 \cdot \left\| \ \begin{tikzpicture}[baseline={([yshift=-.5ex]current bounding box.center)}]
        \draw (1.2,0) -- (1.2,-0.4);
        \draw (0,0) -- (0,-0.4);
        \draw (-0.4,-0.4) -- (0.6,-0.4);
        \draw (0.9,-0.4) -- (1.2,-0.4);
        \draw (0.9,-0.4) -- (0.9,0);
        \draw (0.6,-0.4) -- (0.6,0);
        \draw (-0.4,-0.4) -- (-0.4,-1.2);
        \draw (1.2,-0.4) -- (1.2,-1.2);
        \draw (-0.4,-1.2) -- (1.2,-1.2);
        \fill (1.2,-0.8) ellipse (2pt and 14pt);
        \draw (1.2,-0.8) node[anchor=west]{$\hat{\Phi}$};
        \fill (0.6,-0.2) circle (2pt);
        \draw (0.6,-0.2) node[anchor=east] {\scriptsize $\sqrt{\rho}$};
        \draw[fill=white] (0,-0.4) circle (2pt);
        \fill (0.0,-0.4) circle (1.25pt);
    \end{tikzpicture} \right\|_2.
\end{equation}
Combining these inequalities and setting 
\begin{equation}
    a(p) \coloneq \|\rho^{-1/2} \otimes \rho^{-1/2}\|_2 \cdot \|\mathcal{T}^p - \mathcal{T}_\infty\| = \Tr(\rho^{-1})\cdot \|\mathcal{T}^p-\mathcal{T}_{\infty}\|,
\end{equation}
we can write
\begin{equation}
\label{fnw_subresult1}
     \left\vert \begin{tikzpicture}[baseline={([yshift=-.5ex]current bounding box.center)}]
        \draw (0,0) -- (1.6,0);
        \draw (0,-0.8) -- (0,0);
        \draw (1.6,-0.8) -- (1.6,0);
        \draw (0,-0.8) -- (1.6,-0.8);
        \draw (0.6,0) -- (0.6,0.4);
        \draw (1.2,0) -- (1.2,0.4);
        \draw (0,0) -- (0,0.4);
        \fill (0,-0.4) ellipse (2pt and 14pt);
        \draw (0,-0.4) node[anchor=east]{$\hat{\Psi}$};
        \draw (-0.4,0.4) -- (1.2,0.4);
        \draw (-0.4,0.4) -- (-0.4,1.2);
        \draw (1.2,0.4) -- (1.2,1.2);
        \draw (-0.4,1.2) -- (1.2,1.2);
        \fill (1.2,0.8) ellipse (2pt and 14pt);
        \draw (1.2,0.8) node[anchor=west]{$\bar{\hat{\Phi}}$};
        \draw[fill=white] (0.6,0) circle (2pt);
        \draw[fill=white] (0.6,0.4) circle (2pt);
        \draw[fill=white] (0,0.4) circle (2pt);
        \draw[fill=white] (1.2,0) circle (2pt);
        \fill (0.6,0) circle (0.75pt);
        \fill (0.6,0.4) circle (0.75pt);
        \fill (0.0,0.4) circle (1.25pt);
        \fill (1.2,0) circle (1.25pt);
    \end{tikzpicture} - \begin{tikzpicture}[baseline={([yshift=-.5ex]current bounding box.center)}]
        \draw (0,0) -- (0.5,0);
        \draw (0.5,0) -- (0.5,0.4);
        \draw (0.5,0.2) node[anchor=east]{$\rho$};
        \draw (0.7,0) -- (0.7,0.4);
        \draw (0.7,0) -- (1.6,0);
        \draw (0,-0.8) -- (0,0);
        \draw (1.6,-0.8) -- (1.6,0);
        \draw (0,-0.8) -- (1.6,-0.8);
        \draw (1.2,0) -- (1.2,0.4);
        \draw (0,0) -- (0,0.4);
        \fill (0,-0.4) ellipse (2pt and 14pt);
        \draw (0,-0.4) node[anchor=east]{$\hat{\Psi}$};
        \draw (-0.4,0.4) -- (0.5,0.4);
        \draw (0.7,0.4) -- (1.2,0.4);
        \draw (-0.4,0.4) -- (-0.4,1.2);
        \draw (1.2,0.4) -- (1.2,1.2);
        \draw (-0.4,1.2) -- (1.2,1.2);
        \fill (1.2,0.8) ellipse (2pt and 14pt);
        \draw (1.2,0.8) node[anchor=west]{$\bar{\hat{\Phi}}$};
        \draw[fill=white] (1.2,0) circle (2pt);
        \draw[fill=white] (0,0.4) circle (2pt);
        \fill (0.5,0.2) circle (2pt);
        \fill (1.2,0) circle (1.25pt);
        \fill (0.0,0.4) circle (1.25pt);
    \end{tikzpicture} \right\vert \leq a(p) \cdot \left\| \begin{tikzpicture}[baseline={([yshift=-.5ex]current bounding box.center)}]
        \draw (0,0) -- (0.6,0);
        \draw (0.6,0) -- (0.6,0.4);
        \draw (0.9,0) -- (1.6,0);
        \draw (0.9,0) -- (0.9,0.4);
        \draw (0,-0.8) -- (0,0);
        \draw (1.6,-0.8) -- (1.6,0);
        \draw (0,-0.8) -- (1.6,-0.8);
        \draw (1.2,0) -- (1.2,0.4);
        \draw (0,0) -- (0,0.4);
        \fill (0,-0.4) ellipse (2pt and 14pt);
        \draw (0,-0.4) node[anchor=east]{$\hat{\Psi}$};
        \fill (0.6,0.2) circle (2pt);
        \draw (0.6,0.2) node[anchor=east] {\scriptsize $\sqrt{\rho}$};
        \draw[fill=white] (1.2,0) circle (2pt);
        \fill (1.2,0) circle (1.25pt);
    \end{tikzpicture} \ \right\|_2 \cdot \left\| \ \begin{tikzpicture}[baseline={([yshift=-.5ex]current bounding box.center)}]
        \draw (1.2,0) -- (1.2,-0.4);
        \draw (0,0) -- (0,-0.4);
        \draw (-0.4,-0.4) -- (0.6,-0.4);
        \draw (0.9,-0.4) -- (1.2,-0.4);
        \draw (0.9,-0.4) -- (0.9,0);
        \draw (0.6,-0.4) -- (0.6,0);
        \draw (-0.4,-0.4) -- (-0.4,-1.2);
        \draw (1.2,-0.4) -- (1.2,-1.2);
        \draw (-0.4,-1.2) -- (1.2,-1.2);
        \fill (1.2,-0.8) ellipse (2pt and 14pt);
        \draw (1.2,-0.8) node[anchor=west]{$\hat{\Phi}$};
        \fill (0.6,-0.2) circle (2pt);
        \draw (0.6,-0.2) node[anchor=east] {\scriptsize $\sqrt{\rho}$};
        \draw[fill=white] (0,-0.4) circle (2pt);
        \fill (0.0,-0.4) circle (1.25pt);
    \end{tikzpicture} \right\|_2.
\end{equation}

Applying the same approximation procedure to $\braket{\Phi}{\Phi}$ (the calculation for $\braket{\Psi}{\Psi}$ is analogous) then yields
\begin{equation}
    \left\vert \ \begin{tikzpicture}[baseline={([yshift=-.5ex]current bounding box.center)}]
        \draw (0,0) -- (1.6,0);
        \draw (0,-0.8) -- (0,0);
        \draw (1.6,-0.8) -- (1.6,0);
        \draw (0,-0.8) -- (1.6,-0.8);
        \draw (0.4,0) -- (0.4,0.4);
        \draw (1,0) -- (1,0.4);
        \draw (1.6,0) -- (1.6,0.4);
        \fill (1.6,-0.4) ellipse (2pt and 14pt);
        \fill (1.6,0.8) ellipse (2pt and 14pt);
        \draw (1.6,-0.4) node[anchor=west]{$\hat{\Phi}$};
        \draw (1.6,0.8) node[anchor=west]{$\bar{\hat{\Phi}}$};
        \draw (0,0.4) -- (1.6,0.4);
        \draw (0,0.4) -- (0,1.2);
        \draw (0,1.2) -- (1.6,1.2);
        \draw (1.6,0.4) -- (1.6,1.2);
        \draw[fill=white] (0.4,0) circle (2pt);
        \draw[fill=white] (0.4,0.4) circle (2pt);
        \draw[fill=white] (1,0) circle (2pt);
        \draw[fill=white] (1,0.4) circle (2pt);
        \fill (0.4,0) circle (1.25pt);
        \fill (0.4,0.4) circle (1.25pt);
        \fill (1,0) circle (0.75pt);
        \fill (1,0.4) circle (0.75pt);
    \end{tikzpicture} - \begin{tikzpicture}[baseline={([yshift=-.5ex]current bounding box.center)}]
        \draw (0,0) -- (1,0);
        \draw (1.3,0) -- (1.6,0);
        \draw (1,0) -- (1,0.4);
        \draw (1.3,0) -- (1.3,0.4);
        \draw (0,-0.8) -- (0,0);
        \draw (1.6,-0.8) -- (1.6,0);
        \draw (0,-0.8) -- (1.6,-0.8);
        \draw (0.4,0) -- (0.4,0.4);
        \draw (1.6,0) -- (1.6,0.4);
        \fill (1.6,-0.4) ellipse (2pt and 14pt);
        \fill (1.6,0.8) ellipse (2pt and 14pt);
        \draw (1.6,-0.4) node[anchor=west]{$\hat{\Phi}$};
        \draw (1.6,0.8) node[anchor=west]{$\bar{\hat{\Phi}}$};
        \draw (0,0.4) -- (1,0.4);
        \draw (1.3,0.4) -- (1.6,0.4);
        \draw (0,0.4) -- (0,1.2);
        \draw (0,1.2) -- (1.6,1.2);
        \draw (1.6,0.4) -- (1.6,1.2);
        \fill (1,0.2) circle (2pt);
        \draw (1,0.2) node [anchor=east]{$\rho$};
        \draw[fill=white] (0.4,0) circle (2pt);
        \draw[fill=white] (0.4,0.4) circle (2pt);
        \fill (0.4,0) circle (1.25pt);
        \fill (0.4,0.4) circle (1.25pt);
    \end{tikzpicture}  \right\vert \leq a(p) \cdot \left\| \ \begin{tikzpicture}[baseline={([yshift=-.5ex]current bounding box.center)}]
        \draw (1.2,0) -- (1.2,-0.4);
        \draw (0,0) -- (0,-0.4);
        \draw (-0.4,-0.4) -- (0.6,-0.4);
        \draw (0.9,-0.4) -- (1.2,-0.4);
        \draw (0.9,-0.4) -- (0.9,0);
        \draw (0.6,-0.4) -- (0.6,0);
        \draw (-0.4,-0.4) -- (-0.4,-1.2);
        \draw (1.2,-0.4) -- (1.2,-1.2);
        \draw (-0.4,-1.2) -- (1.2,-1.2);
        \fill (1.2,-0.8) ellipse (2pt and 14pt);
        \draw (1.2,-0.8) node[anchor=west]{$\hat{\Phi}$};
        \fill (0.6,-0.2) circle (2pt);
        \draw (0.6,-0.2) node[anchor=east] {\scriptsize $\sqrt{\rho}$};
        \draw[fill=white] (0,-0.4) circle (2pt);
        \fill (0.0,-0.4) circle (1.25pt);
    \end{tikzpicture} \right\|^2_2 = a(p) \cdot \begin{tikzpicture}[baseline={([yshift=-.5ex]current bounding box.center)}]
        \draw (0,0) -- (1,0);
        \draw (1.3,0) -- (1.6,0);
        \draw (1,0) -- (1,0.4);
        \draw (1.3,0) -- (1.3,0.4);
        \draw (0,-0.8) -- (0,0);
        \draw (1.6,-0.8) -- (1.6,0);
        \draw (0,-0.8) -- (1.6,-0.8);
        \draw (0.4,0) -- (0.4,0.4);
        \draw (1.6,0) -- (1.6,0.4);
        \fill (1.6,-0.4) ellipse (2pt and 14pt);
        \fill (1.6,0.8) ellipse (2pt and 14pt);
        \draw (1.6,-0.4) node[anchor=west]{$\hat{\Phi}$};
        \draw (1.6,0.8) node[anchor=west]{$\bar{\hat{\Phi}}$};
        \draw (0,0.4) -- (1,0.4);
        \draw (1.3,0.4) -- (1.6,0.4);
        \draw (0,0.4) -- (0,1.2);
        \draw (0,1.2) -- (1.6,1.2);
        \draw (1.6,0.4) -- (1.6,1.2);
        \fill (1,0.2) circle (2pt);
        \draw (1,0.2) node [anchor=east]{$\rho$};
        \draw[fill=white] (0.4,0) circle (2pt);
        \draw[fill=white] (0.4,0.4) circle (2pt);
        \fill (0.4,0) circle (1.25pt);
        \fill (0.4,0.4) circle (1.25pt);
    \end{tikzpicture}.
\end{equation}
As a consequence, we can bound the 2-norms in \cref{fnw_subresult1} by the vector norms of $\ket{\Psi}$ and $\ket{\Phi}$ as
\begin{equation}
\label{fnw_subresult2}
\begin{aligned}
    \left\| \ \begin{tikzpicture}[baseline={([yshift=-.5ex]current bounding box.center)}]
        \draw (1.2,0) -- (1.2,-0.4);
        \draw (0,0) -- (0,-0.4);
        \draw (-0.4,-0.4) -- (0.6,-0.4);
        \draw (0.9,-0.4) -- (1.2,-0.4);
        \draw (0.9,-0.4) -- (0.9,0);
        \draw (0.6,-0.4) -- (0.6,0);
        \draw (-0.4,-0.4) -- (-0.4,-1.2);
        \draw (1.2,-0.4) -- (1.2,-1.2);
        \draw (-0.4,-1.2) -- (1.2,-1.2);
        \fill (1.2,-0.8) ellipse (2pt and 14pt);
        \draw (1.2,-0.8) node[anchor=west]{$\hat{\Phi}$};
        \fill (0.6,-0.2) circle (2pt);
        \draw (0.6,-0.2) node[anchor=east] {\scriptsize $\sqrt{\rho}$};
        \draw[fill=white] (0,-0.4) circle (2pt);
        \fill (0.0,-0.4) circle (1.25pt);
    \end{tikzpicture} \right\|_2 & \leq \frac{1}{\sqrt{1-a(p)}} \sqrt{\begin{tikzpicture}[baseline={([yshift=-.5ex]current bounding box.center)}]
        \draw (0,0) -- (1.6,0);
        \draw (0,-0.8) -- (0,0);
        \draw (1.6,-0.8) -- (1.6,0);
        \draw (0,-0.8) -- (1.6,-0.8);
        \draw (0.4,0) -- (0.4,0.4);
        \draw (1,0) -- (1,0.4);
        \draw (1.6,0) -- (1.6,0.4);
        \fill (1.6,-0.4) ellipse (2pt and 14pt);
        \fill (1.6,0.8) ellipse (2pt and 14pt);
        \draw (1.6,-0.4) node[anchor=west]{$\hat{\Phi}$};
        \draw (1.6,0.8) node[anchor=west]{$\bar{\hat{\Phi}}$};
        \draw (0,0.4) -- (1.6,0.4);
        \draw (0,0.4) -- (0,1.2);
        \draw (0,1.2) -- (1.6,1.2);
        \draw (1.6,0.4) -- (1.6,1.2);
        \draw[fill=white] (0.4,0) circle (2pt);
        \draw[fill=white] (0.4,0.4) circle (2pt);
        \draw[fill=white] (1,0) circle (2pt);
        \draw[fill=white] (1,0.4) circle (2pt);
        \fill (0.4,0) circle (1.25pt);
        \fill (0.4,0.4) circle (1.25pt);
        \fill (1,0) circle (0.75pt);
        \fill (1,0.4) circle (0.75pt);
    \end{tikzpicture}} = \frac{\|\ket{\Phi}\|}{\sqrt{1-a(p)}}, \\
    \left\| \begin{tikzpicture}[baseline={([yshift=-.5ex]current bounding box.center)}]
        \draw (0,0) -- (0.6,0);
        \draw (0.6,0) -- (0.6,0.4);
        \draw (0.9,0) -- (1.6,0);
        \draw (0.9,0) -- (0.9,0.4);
        \draw (0,-0.8) -- (0,0);
        \draw (1.6,-0.8) -- (1.6,0);
        \draw (0,-0.8) -- (1.6,-0.8);
        \draw (1.2,0) -- (1.2,0.4);
        \draw (0,0) -- (0,0.4);
        \fill (0,-0.4) ellipse (2pt and 14pt);
        \draw (0,-0.4) node[anchor=east]{$\hat{\Psi}$};
        \fill (0.6,0.2) circle (2pt);
        \draw (0.6,0.2) node[anchor=east] {\scriptsize $\sqrt{\rho}$};
        \draw[fill=white] (1.2,0) circle (2pt);
        \fill (1.2,0) circle (1.25pt);
    \end{tikzpicture} \ \right\|_2 & \leq \frac{\|\ket{\Psi}\|}{\sqrt{1-a(p)}}.
\end{aligned}
\end{equation}
Now, substituting this bound into \cref{fnw_subresult0,fnw_subresult1} gives
\begin{equation}
    |\braket{\Phi}{\Psi}| \leq \frac{a(p)}{1-a(p)} \|\Phi\|\cdot\|\Psi\| + \left\vert \begin{tikzpicture}[baseline={([yshift=-.5ex]current bounding box.center)}]
        \draw (0,0) -- (0.5,0);
        \draw (0.5,0) -- (0.5,0.4);
        \fill (0.5,0.2) circle (2pt);
        \draw (0.5,0.2) node[anchor=east]{$\rho$};
        \draw (0.7,0) -- (0.7,0.4);
        \draw (0.7,0) -- (1.6,0);
        \draw (0,-0.8) -- (0,0);
        \draw (1.6,-0.8) -- (1.6,0);
        \draw (0,-0.8) -- (1.6,-0.8);
        \draw (1.2,0) -- (1.2,0.4);
        \draw (0,0) -- (0,0.4);
        \fill (0,-0.4) ellipse (2pt and 14pt);
        \draw (0,-0.4) node[anchor=east]{$\hat{\Psi}$};
        \draw (-0.4,0.4) -- (0.5,0.4);
        \draw (0.7,0.4) -- (1.2,0.4);
        \draw (-0.4,0.4) -- (-0.4,1.2);
        \draw (1.2,0.4) -- (1.2,1.2);
        \draw (-0.4,1.2) -- (1.2,1.2);
        \fill (1.2,0.8) ellipse (2pt and 14pt);
        \draw (1.2,0.8) node[anchor=west]{$\bar{\hat{\Phi}}$};
        \draw[fill=white] (0,0.4) circle (2pt);
        \draw[fill=white] (1.2,0) circle (2pt);
        \fill (0.0,0.4) circle (1.25pt);
        \fill (1.2,0) circle (1.25pt);
    \end{tikzpicture} \right\vert.
\end{equation}

Now we define the $D \times D$ matrices $\Delta_\Psi$ and $\Delta_\Phi$ such that their $\rho$-scalar product is exactly the diagram on the right-hand side:
\begin{equation}
    \Delta_{\Psi} = \begin{tikzpicture}[baseline={([yshift=-.5ex]current bounding box.center)}]
        \draw (0,0) -- (0.5,0);
        \draw (0.5,0) -- (0.5,0.4);
        \fill (0.5,0.2) circle(2pt);
        \draw (0.5,0.2) node[anchor=east]{\small $\rho$};
        \draw (0.5,0.4) -- (-0.4,0.4);
        \draw (-0.4,0.4) -- (-0.4,0.8);
        \fill (-0.4,0.6) circle(2pt);
        \draw (-0.4,0.6) node[anchor=east]{\small $\rho^{-1}$};
        \draw (-0.4,0.8) -- (0.8,0.8);
        \draw (0,-0.8) -- (0,0.4);
        \draw (0,-0.8) -- (0.8,-0.8);
        \fill (0,-0.4) ellipse (2pt and 14pt);
        \draw (0,-0.4) node[anchor=east]{$\hat{\Psi}$};
        \draw[fill=white] (0,0.4) circle (2pt);
        \fill (0,0.4) circle(1.25pt);
    \end{tikzpicture} \quad , \qquad \Delta_{\Phi} = \begin{tikzpicture}[baseline={([yshift=-.5ex]current bounding box.center)}]
        \draw (0,0) -- (-0.5,0);
        \draw (-0.5,0) -- (-0.5,0.4);
        \draw (-0.5,0.4) -- (0.4,0.4);
        \draw (0.4,0.4) -- (0.4,0.8);
        \draw (0.4,0.8) -- (-0.8,0.8);
        \draw (0,-0.8) -- (0,0.4);
        \draw (0,-0.8) -- (-0.8,-0.8);
        \fill (0,-0.4) ellipse (2pt and 14pt);
        \draw (0,-0.4) node[anchor=east]{$\hat{\Phi}$};
        \draw[fill=white] (0,0.4) circle (2pt);
        \fill (0,0.4) circle(1.25pt);
    \end{tikzpicture}.
\end{equation}
For any $Z \in \mathcal{M}_D$ we have
\begin{equation}
    |\langle Z, \Delta_\Psi \rangle_\rho | = \left\vert \begin{tikzpicture}[baseline={([yshift=-.5ex]current bounding box.center)}]
        \draw (0,0) -- (0.5,0);
        \draw (0.5,0) -- (0.5,0.4);
        \fill (0.5,0.2) circle (2pt);
        \draw (0.5,0.2) node[anchor=east]{$\rho$};
        \draw (0.7,0) -- (0.7,0.4);
        \draw (0.7,0) -- (1.6,0);
        \draw (0,-0.8) -- (0,0);
        \draw (1.6,-0.8) -- (1.6,0);
        \draw (0,-0.8) -- (1.6,-0.8);
        \draw (0,0) -- (0,0.4);
        \fill (0,-0.4) ellipse (2pt and 14pt);
        \draw (0,-0.4) node[anchor=east]{$\hat{\Psi}$};
        \draw (-0.4,0.4) -- (0.5,0.4);
        \draw (0.7,0.4) -- (1.2,0.4);
        \draw (-0.4,0.4) -- (-0.4,1.2);
        \draw (1.2,0.4) -- (1.2,1.2);
        \draw (-0.4,1.2) -- (1.2,1.2);
        \fill (0.4,1.2) circle (2pt);
        \draw (0.4,1.2) node[anchor=south]{$\bar{Z}$};
        \draw[fill=white] (0,0.4) circle (2pt);
        \fill (0.0,0.4) circle (1.25pt);
    \end{tikzpicture} \right\vert.
\end{equation}

The vectors $\ket{\Psi}$ and $\ket{\Phi}$ are orthogonal to $\ket{\Psi_{3p+2q}(A,Z)}$ for any $Z \in \mathcal{M}_D$ (see \cref{orth}), thus we can use the same approximation strategy as before and write
\begin{equation}
    |\langle Z, \Delta_{\Psi}\rangle_{\rho}| = \left\vert \begin{tikzpicture}[baseline={([yshift=-.5ex]current bounding box.center)}]
        \draw (0,0) -- (1.6,0);
        \draw (1,0) -- (1,0.4);
        \draw (0.5,0) -- (0.5,0.4);
        \draw (0,-0.8) -- (0,0);
        \draw (1.6,-0.8) -- (1.6,0);
        \draw (0,-0.8) -- (1.6,-0.8);
        \draw (0,0) -- (0,0.4);
        \fill (0,-0.4) ellipse (2pt and 14pt);
        \draw (0,-0.4) node[anchor=east]{$\hat{\Psi}$};
        \draw (-0.4,0.4) -- (1.2,0.4);
        \draw (-0.4,0.4) -- (-0.4,1.2);
        \draw (1.2,0.4) -- (1.2,1.2);
        \draw (-0.4,1.2) -- (1.2,1.2);
        \fill (0.4,1.2) circle (2pt);
        \draw (0.4,1.2) node[anchor=south]{$\bar{Z}$};
        \draw[fill=white] (0,0.4) circle(2pt);
        \draw[fill=white] (0.5,0.4) circle(2pt);
        \draw[fill=white] (1,0.4) circle(2pt);
        \draw[fill=white] (0.5,0) circle(2pt);
        \draw[fill=white] (1,0) circle(2pt);
        \fill (0.0,0.4) circle (1.25pt);
        \fill (0.5,0) circle (0.75pt);
        \fill (0.5,0.4) circle (0.75pt);
        \fill (1,0) circle (1.25pt);
        \fill (1,0.4) circle (1.25pt);
    \end{tikzpicture} - \begin{tikzpicture}[baseline={([yshift=-.5ex]current bounding box.center)}]
        \draw (0,0) -- (0.5,0);
        \draw (0.5,0) -- (0.5,0.4);
        \fill (0.5,0.2) circle (2pt);
        \draw (0.5,0.2) node[anchor=east]{$\rho$};
        \draw (0.7,0) -- (0.7,0.4);
        \draw (0.7,0) -- (1.6,0);
        \draw (0,-0.8) -- (0,0);
        \draw (1.6,-0.8) -- (1.6,0);
        \draw (0,-0.8) -- (1.6,-0.8);
        \draw (0,0) -- (0,0.4);
        \fill (0,-0.4) ellipse (2pt and 14pt);
        \draw (0,-0.4) node[anchor=east]{$\hat{\Psi}$};
        \draw (-0.4,0.4) -- (0.5,0.4);
        \draw (0.7,0.4) -- (1.2,0.4);
        \draw (-0.4,0.4) -- (-0.4,1.2);
        \draw (1.2,0.4) -- (1.2,1.2);
        \draw (-0.4,1.2) -- (1.2,1.2);
        \fill (0.4,1.2) circle (2pt);
        \draw (0.4,1.2) node[anchor=south]{$\bar{Z}$};
        \draw[fill=white] (0,0.4) circle (2pt);
        \fill (0.0,0.4) circle (1.25pt);
    \end{tikzpicture} \right\vert \leq a(p) \left\| \begin{tikzpicture}[baseline={([yshift=-.5ex]current bounding box.center)}]
        \draw (0,0) -- (0.6,0);
        \draw (0.6,0) -- (0.6,0.4);
        \draw (0.9,0) -- (1.6,0);
        \draw (0.9,0) -- (0.9,0.4);
        \draw (0,-0.8) -- (0,0);
        \draw (1.6,-0.8) -- (1.6,0);
        \draw (0,-0.8) -- (1.6,-0.8);
        \draw (1.2,0) -- (1.2,0.4);
        \draw (0,0) -- (0,0.4);
        \fill (0,-0.4) ellipse (2pt and 14pt);
        \draw (0,-0.4) node[anchor=east]{$\hat{\Psi}$};
        \fill (0.6,0.2) circle (2pt);
        \draw (0.6,0.2) node[anchor=east] 
        {\scriptsize $\sqrt{\rho}$};
        \draw[fill=white] (1.2,0) circle (2pt);
        \fill (1.2,0) circle (1.25pt);
    \end{tikzpicture} \ \right\|_2 \cdot \left\| \ \begin{tikzpicture}[baseline={([yshift=-.5ex]current bounding box.center)}]
        \draw (0,0) -- (0.7,0);
        \draw (0.7,0) -- (0.7,0.4);
        \draw (0.9,0) -- (1.6,0);
        \draw (0.9,0) -- (0.9,0.4);
        \draw (0,-0.8) -- (0,0);
        \draw (1.6,-0.8) -- (1.6,0);
        \draw (0,-0.8) -- (1.6,-0.8);
        \draw (1.25,0) -- (1.25,0.4);
        \draw (0.2,0) -- (0.2,0.4);
        \fill (0.7,0.2) circle (2pt);
        \draw (0.7,0.2) node[anchor=east] {\tiny $\sqrt{\rho}$};
        \fill (0.8,-0.8) circle (2pt);
        \draw (0.8,-0.8) node[anchor=south] {$Z$};
        \draw[fill=white] (0.2,0) circle (2pt);
        \draw[fill=white] (1.25,0) circle (2pt);
        \fill (0.2,0) circle (1.25pt);
        \fill (1.25,0) circle (1.25pt);
    \end{tikzpicture} \right\|_2.
\end{equation}
The first norm on the right-hand side is upper bounded by $\|\ket{\Psi}\|/\sqrt{1-a(p)}$ (cf. \cref{fnw_subresult2}), and the second norm is $\|Z\|_\rho$, since we have assumed that the MPS is in the left canonical form. By setting $Z = \Delta_\Psi$ we have
\begin{equation}
    \|\Delta_\Psi\|_\rho \leq \frac{a(p)}{\sqrt{1-a(p)}}\|\ket{\Psi}\|.
\end{equation}

Applying the same steps to $|\langle \Delta_\Phi , Z \rangle_\rho|$ results analogously in $\|\Delta_\Phi\|_\rho \leq \frac{a(p)}{\sqrt{1-a(p)}}\|\ket{\Phi}\|$. The Cauchy-Schwarz inequality then gives
\begin{equation}
    |\langle \Delta_\Psi, \Delta_\Phi\rangle_\rho| \leq \|\Delta_\Psi\|_\rho \cdot \|\Delta_\Phi\|_\rho \leq \frac{a(p)^2}{1-a(p)}\|\ket{\Psi}\|\cdot\|\ket{\Phi}\|.
\end{equation}
By conclusion, we have
\begin{equation}
    \gamma_{p,q} = \sup_{\substack{\ket{\Psi} \in \mathcal{R} \setminus \{0\} \\ \ket{\Phi} \in \mathcal{Q} \setminus \{0\}}} \frac{|\braket{\Phi}{\Psi}|}{\|\ket{\Phi}\| \cdot \|\ket{\Psi}\|} \leq a(p)\frac{1+a(p)}{1-a(p)}.
\end{equation}
\hspace{\fill}$\square$

\subsection{Block-normal MPS and the result of Nachtergaele}\label{sec:nachtergaele}

In the case of block-normal MPS, Nachtergaele~\cite{nachtergaele_spectral_1996} derived an upper bound on $\gamma_{p,q}$, which we state here without proof. Let $\alpha=1,\dots,k$ label the blocks. Define
$a_\alpha(p):=\Tr(\rho_\alpha^{-1}) \cdot\|\mathcal{T}_\alpha^p -
(\mathcal{T}_\alpha)_{\infty}\|$, where $\mathcal T_\alpha$ and $\rho_\alpha$
denote the transfer matrix and its normalized fixed point restricted to block
$\alpha$, respectively.  Let $p\ge \ell_{\mathrm{max}} =
\max_\alpha\ell_\alpha$ (cf.\ \cref{subsec:parenthams}). Define the overlap
matrix 
\begin{equation}
\label{overlapmatrix}
    (o_p)_{\alpha \beta} := (1-\delta_{\alpha \beta}) \sqrt{\max \operatorname{spec}(P_p^{(\alpha)}P_p^{(\beta)})}\ ,
\end{equation}
where $P^{(\alpha)}$ is the projector onto the space spanned by the MPS block $\alpha$. Note that $(o_p)_{\alpha\beta}\to0$, cf.~\cref{footnote:orth}, and thus $\|o_p\|\to 0$. 

\begin{theorem}[Nachtergaele~\cite{nachtergaele_spectral_1996}]
\label{thm:martingale}
    For a block-normal MPS, we have the following bound for all $q \ge 0$ and $p\ge \ell_{\mathrm{max}}$ large enough such that $\|o_p\|<1/2$ and $a_\alpha(p) <1\ \forall \alpha$:
    \begin{equation}
        \label{martingale}
    \gamma_{p,q} \leq \frac{4\|o_p\|(1-\|o_p\|)}{(1-2\|o_p\|)^2} + \sum_{\alpha=1}^{k} a_{\alpha}(p) \frac{1+a_\alpha(p)}{1-a_\alpha(p)}\ .
    \end{equation}
\end{theorem}
Yet again, this bound implies that $\gamma_{p,q}\to0$ as $p\to\infty$, and thus that any parent Hamiltonian of a block-normal MPS is gapped.

In the same article, Nachtergaele proposed an alternative way of regrouping the Hamiltonian. In the following, we state the related spectral gap bound without proof.
\begin{theorem}[Nachtergaele~\cite{nachtergaele_spectral_1996}]\label{theorem:nachtergaele2}
    For parent Hamiltonians $H_N$ of block-normal MPS, let $r\in \mathbb{N}$ be large enough such that $\gamma_{r,0} < \tfrac{1}{\sqrt{r}}$. Then we have the following spectral gap bound for any $N \ge r$:
    \begin{equation}
        \Delta(H_N) \ge \frac{\Delta(H_r)}{r-1}(1-\sqrt{r} \cdot \gamma_{r,0})^2\ .
    \end{equation}
\end{theorem}

Finally, we note that the matrix $o_p$ given in \cref{overlapmatrix} can be efficiently calculated in a similar way as the optimal $\gamma_{p,q}$ in~\cref{theorem:gamma-exact}: the nonzero spectrum of $(P_p^{(\alpha)}\cdot P_p^{(\beta)})$ is computable by diagonalizing the following $D_\beta^2\times D_\beta^2$ matrix:
\begin{equation}
    \begin{tikzpicture}[baseline={([yshift=-.5ex]current bounding box.center)}]
        \draw (-0.75,0.25) -- (0.75,0.25);
        \draw (-0.75,0.25) -- (-0.75,0.75);
        \draw (-0.75,0.75) -- (0.75,0.75);
        \draw (0.75,0.75) -- (0.75,0.25);
        \draw (-0.5,0.25) -- (-0.5,0);
        \draw (0.5,0.25) -- (0.5,0);
        \draw (-0.5,0.75) -- (-0.5,1);
        \draw (0.5,0.75) -- (0.5,1);
        \draw (-0.5,1) -- (0.5,1);
        \draw (-0.5,3) -- (0.5,3);
        \draw (0,0.5) node[]{\small $(T_p^{-1})^{(\beta)}$};
        \draw (0,1) -- (0,1.25);
        \draw (0,3) -- (0,2.75);
        \draw (-0.5,3) -- (-0.5,3.25);
        \draw (0.5,3) -- (0.5,3.25);
        \draw[blue] (-0.75,1.75) -- (0.75,1.75);
        \draw[blue] (-0.75,1.75) -- (-0.75,2.25);
        \draw[blue] (-0.75,2.25) -- (0.75,2.25);
        \draw[blue] (0.75,2.25) -- (0.75,1.75);
        \draw[blue] (-0.5,1.75) -- (-0.5,1.5);
        \draw[blue] (0.5,1.75) -- (0.5,1.5);
        \draw[blue] (-0.5,2.25) -- (-0.5,2.5);
        \draw[blue] (0.5,2.25) -- (0.5,2.5);
        \draw[blue] (-0.5,2.5) -- (0.5,2.5);
        \draw[blue] (-0.5,1.5) -- (0.5,1.5);
        \draw[blue] (0,2) node[]{\small $(T_p^{-1})^{(\alpha)}$};
        \draw[blue] (0,2.5) -- (0,2.75);
        \draw[blue] (0,1.5) -- (0,1.25);
        \draw[fill=white] (0,1) circle (2pt);
        \draw[fill=white] (0,3) circle (2pt);
        \draw[blue,fill=white] (0,2.5) circle (2pt);
        \draw[blue,fill=white] (0,1.5) circle (2pt);
        \fill (0,1) circle (0.75pt);
        \fill (0,3) circle (0.75pt);
        \fill[blue] (0,2.5) circle (0.75pt);
        \fill[blue] (0,1.5) circle (0.75pt);
        \end{tikzpicture},
\end{equation}
as we have
\begin{equation}
\label{projoverlap}
    \operatorname{spec}\left( \begin{tikzpicture}[baseline={([yshift=-.5ex]current bounding box.center)}]
        \draw (-0.75,0.25) -- (0.75,0.25);
        \draw (-0.75,0.25) -- (-0.75,0.75);
        \draw (-0.75,0.75) -- (0.75,0.75);
        \draw (0.75,0.75) -- (0.75,0.25);
        \draw (-0.5,0.25) -- (-0.5,0);
        \draw (0.5,0.25) -- (0.5,0);
        \draw (-0.5,0.75) -- (-0.5,1);
        \draw (0.5,0.75) -- (0.5,1);
        \draw (-0.5,1) -- (0.5,1);
        \draw (-0.5,0) -- (0.5,0);
        \draw (0,0.5) node[]{\small $(T_p^{-1})^{(\beta)}$};
        \draw (0,1) -- (0,1.25);
        \draw (0,0) -- (0,-0.25);
        \draw[blue] (-0.75,1.75) -- (0.75,1.75);
        \draw[blue] (-0.75,1.75) -- (-0.75,2.25);
        \draw[blue] (-0.75,2.25) -- (0.75,2.25);
        \draw[blue] (0.75,2.25) -- (0.75,1.75);
        \draw[blue] (-0.5,1.75) -- (-0.5,1.5);
        \draw[blue] (0.5,1.75) -- (0.5,1.5);
        \draw[blue] (-0.5,2.25) -- (-0.5,2.5);
        \draw[blue] (0.5,2.25) -- (0.5,2.5);
        \draw[blue] (-0.5,2.5) -- (0.5,2.5);
        \draw[blue] (-0.5,1.5) -- (0.5,1.5);
        \draw[blue] (0,2) node[]{\small $(T_p^{-1})^{(\alpha)}$};
        \draw[blue] (0,2.5) -- (0,2.75);
        \draw[blue] (0,1.5) -- (0,1.25);
        \draw[fill=white] (0,0) circle (2pt);
        \draw[fill=white] (0,1) circle (2pt);
        \draw[blue,fill=white] (0,1.5) circle (2pt);
        \draw[blue,fill=white] (0,2.5) circle (2pt);
        \fill (0,1) circle (0.75pt);
        \fill (0,0) circle (0.75pt);
        \fill[blue] (0,1.5) circle (0.75pt);
        \fill[blue] (0,2.5) circle (0.75pt);
        \end{tikzpicture} \right) \setminus \{0\} = \operatorname{spec} \left( \begin{tikzpicture}[baseline={([yshift=-.5ex]current bounding box.center)}]
        \draw (-0.75,0.25) -- (0.75,0.25);
        \draw (-0.75,0.25) -- (-0.75,0.75);
        \draw (-0.75,0.75) -- (0.75,0.75);
        \draw (0.75,0.75) -- (0.75,0.25);
        \draw (-0.5,0.25) -- (-0.5,0);
        \draw (0.5,0.25) -- (0.5,0);
        \draw (-0.5,0.75) -- (-0.5,1);
        \draw (0.5,0.75) -- (0.5,1);
        \draw (-0.5,1) -- (0.5,1);
        \draw (-0.5,3) -- (0.5,3);
        \draw (0,0.5) node[]{\small $(T_p^{-1})^{(\beta)}$};
        \draw (0,1) -- (0,1.25);
        \draw (0,3) -- (0,2.75);
        \draw (-0.5,3) -- (-0.5,3.25);
        \draw (0.5,3) -- (0.5,3.25);
        \draw[blue] (-0.75,1.75) -- (0.75,1.75);
        \draw[blue] (-0.75,1.75) -- (-0.75,2.25);
        \draw[blue] (-0.75,2.25) -- (0.75,2.25);
        \draw[blue] (0.75,2.25) -- (0.75,1.75);
        \draw[blue] (-0.5,1.75) -- (-0.5,1.5);
        \draw[blue] (0.5,1.75) -- (0.5,1.5);
        \draw[blue] (-0.5,2.25) -- (-0.5,2.5);
        \draw[blue] (0.5,2.25) -- (0.5,2.5);
        \draw[blue] (-0.5,2.5) -- (0.5,2.5);
        \draw[blue] (-0.5,1.5) -- (0.5,1.5);
        \draw[blue] (0,2) node[]{\small $(T_p^{-1})^{(\alpha)}$};
        \draw[blue] (0,2.5) -- (0,2.75);
        \draw[blue] (0,1.5) -- (0,1.25);
        \draw[fill=white] (0,1) circle (2pt);
        \draw[fill=white] (0,3) circle (2pt);
        \draw[blue,fill=white] (0,2.5) circle (2pt);
        \draw[blue,fill=white] (0,1.5) circle (2pt);
        \fill (0,1) circle (0.75pt);
        \fill (0,3) circle (0.75pt);
        \fill[blue] (0,2.5) circle (0.75pt);
        \fill[blue] (0,1.5) circle (0.75pt);
        \end{tikzpicture}
        \right) \setminus \{0\}.
\end{equation}
The parts in blue represent the normal MPS and its corresponding inverse transfer matrix generated by the tensor labeled by $\alpha$.

\section{The Knabe and the Gosset-Mozgunov bounds}\label{appendix:knabe_and_gosset}
In this section, we will sketch the results of the Knabe~\cite{knabe_energy_1988} and Gosset-Mozgunov~\cite{gosset_local_2016} which find lower bounds on the spectral gap of frustration-free 2-local Hamiltonians. Collectively, these methods are often called as finite-size criteria, since they prove that an $N$-independent nonzero lower bound on the spectral gap exists if the spectral gap of the restriction of the Hamiltonian to a relatively small subsystem is above a certain threshold, and provide relations between the respective values of the gap.

We recall our notation for a frustration-free OBC or PBC Hamiltonian consisting of translations of the local term $h$:
\begin{equation}
    H_N = \sum_{i=1}^{N-1} h_i\ , \quad H_N^\circ = \sum_{i=1}^N h_i\ .
\end{equation}
Both the Knabe and Gosset-Mozgunov methods were introduced for Hamiltonians with projection-valued local terms. In order to apply them to arbitrary frustration-free Hamiltonians, one has to flatten the spectrum of the local terms: That is, we bound the local terms as $h \ge \Delta(h) \Pi^\perp$, where $\Pi^\perp$ denotes the projector onto the orthogonal complement of the ground state space of $h$. Then the projector-valued Hamiltonians
\begin{equation}
    K_N = \sum_{i=1}^{N-1} \Pi^\perp_i\ , \quad K_N^\circ = \sum_{i=1}^N \Pi^\perp_i
\end{equation}
have the same ground state spaces as $H_N$ and $H_N^\circ$, respectively, and $H_N \ge \Delta(h) K_N$. Due to frustration-freeness, this implies 
\begin{equation}
    \Delta(H_N) \geq \Delta(h) \Delta(K_N) .
\end{equation}
Finally, the two methods obtain a bound on $\Delta(K_N)$ by finding $\delta>0$ such that the inequality
\begin{equation}\label{eq:delta_proj}
    K_N^2- \delta K_N \ge 0
\end{equation}
holds: This inequality implies that $K_N$ does not have any eigenvalues in the interval $(0,\delta)$, thus $\delta \le \Delta(K_N)$.

We furthermore note that this bound holds for any frustration-free local Hamiltonian (with ground state energy being set to 0) independently of the boundary conditions.

\subsection{The Knabe bound}\label{appendix:knabe}
The Knabe bound introduced in Ref.~\cite{knabe_energy_1988} relates the spectral gap of the PBC Hamiltonian $K_N^\circ$ to the gap of the OBC Hamiltonian $K_r=\sum_{i=1}^{r-1} \Pi_i^\perp$ for some small $r\in \mathbb{N},\ 2<r<N$. The general idea is to expand both $K_r^2$ and $(K_N^\circ)^2$ and to sum up translations of $K_r^2$ along the entire ring to find that
\begin{equation}\label{eq:Knabe_subresult1}
    (K_N^\circ)^2 \ge \frac{1}{r-2}\left( \sum_{i=0}^{N-1} \tau_i(K_r^2) - K_N^\circ \right).
\end{equation}
Then, by using the inequalities $K_r^2-\Delta(K_r)K_r \ge 0$ and $H_N^\circ \ge \Delta(h) K_N^\circ$, the Knabe bound is obtained:
\begin{equation}\label{eq:Knabe_pbc}
     \Delta(H_N^\circ) \ge \Delta(h) \frac{r-1}{r-2} \left( \Delta(K_r) - \frac{1}{r-1} \right)
\end{equation}
for any $N>r$. As a consequence, $H_N^\circ$ is gapped if there exists an $r\in \mathbb{N},\ 2 <  r < N$ such that $\Delta(K_r) > \tfrac{1}{r-1}$.

The extension of Knabe's method to OBC Hamiltonians was introduced in Ref.~\cite{wouters_interrelations_2021}. The underlying idea is similar to the PBC case, but instead of translating the subsystem Hamiltonian $K_r$ along a ring, we translate between the two endpoints of the $N$-body open chain, and add the missing terms at the edges that constitute Hamiltonians on a support smaller than $r$ to arrive at the same prefactor as in~\cref{eq:Knabe_subresult1}:
\begin{equation}
    K_N^2 \ge \frac{1}{r-2} \left( \sum_{i=0}^{N-r} \tau_i(K_r^2) + \sum_{r' = 2}^{r-1} (K_{r'}^2 + \tau_{N-r'}(K_{r'}^2)) - K_N \right).
\end{equation}
Finally, we note that $K_N$ can be decomposed in the following way:
\begin{equation}
    K_N = \frac{1}{r-1} \left( \sum_{i=0}^{N-r} \tau_i(K_r) + \sum_{r' = 2}^{r-1}(K_{r'}+\tau_{N-r'}(K_{r'})) \right).
\end{equation}
By using the finite-size gap condition $K_{r'} - \Delta(K_{r'}) K_{r'} \ge 0$ for any $r' \in \{2,\dots,r\}$, we obtain the bound
\begin{equation}
     \Delta(H_N) \ge \Delta(h) \frac{r-1}{r-2} \left( \min_{r' = 2,\dots,r} \Delta(K_{r'}) - \frac{1}{r-1} \right).
\end{equation}

This expression implies that the condition for a frustration-free OBC Hamiltonian to be gapped is more strict, as the spectral gap of the subsystem Hamiltonian $K_{r'}$ has to be above the $\tfrac{1}{r-1}$ threshold for any $r' \in \{2,\dots,r\}$, not only $\Delta(K_r)$ as in the PBC case.

\subsection{The Gosset-Mozgunov bound}\label{appendix:gosset}
The Gosset-Mozgunov method introduced in Ref.~\cite{gosset_local_2016} can be viewed as a refinement of Kna\-be's method for frustration-free PBC Hamiltonians. Instead of working directly with the subsystem Hamiltonians $K_r$, they propose a Hamiltonian with a positive weight belonging to each projector-valued term:
\begin{equation}
    \hat{K}_r = \sum_{i=1}^{r-1} c_i \Pi_i^\perp,
\end{equation}
where $c_i = i(r-i)$. This choice of weights is optimal in the sense that it minimizes the threshold $G(r)$ in the bound of the form $\Delta(K_N^\circ) \ge F(r) (\Delta(K_r) - G(r))$, where $F(r)$ and $G(r)$ are obtained by similar decomposition strategies as the one used for Knabe's method. Moreover, this choice of $c_i$ provides the bound
\begin{equation}
    \Delta(H_N^{\circ}) \geq \frac{5\Delta(h)}{6}\left(\frac{r^2+r}{r^2-4}\right) \left( \Delta(K_r) - \frac{6}{r(r+1)} \right)
\end{equation}
for any $N > 2r$.

With the Gosset-Mozgunov method, the gap detection threshold is improved to $\tfrac{6}{r(r+1)}$ compared to the $\tfrac{1}{r-1}$ threshold obtained by Knabe's method. We note that for $r=3$, the bounds obtained by the Knabe and Gosset-Mozgunov methods are identical.

\section{Supplementary material for the SU(3) VBS models}
\subsection{Gell-Mann matrices}
\label{gm}
\begin{equation}
\begin{aligned}
    \Lambda^1 & = \begin{pmatrix}
        0 & 1 & 0 \\
        1 & 0 & 0 \\
        0 & 0 & 0 
    \end{pmatrix}, \quad \Lambda^2 = \begin{pmatrix}
        0 & -i & 0 \\
        i & 0 & 0 \\
        0 & 0 & 0 
    \end{pmatrix}, \quad \Lambda^3 = \begin{pmatrix}
        1 & 0 & 0 \\
        0 & -1 & 0 \\
        0 & 0 & 0 
    \end{pmatrix}, \\
    \Lambda^4 & = \begin{pmatrix}
        0 & 0 & 1 \\
        0 & 0 & 0 \\
        1 & 0 & 0 
    \end{pmatrix}, \quad 
    \Lambda^5 = \begin{pmatrix}
        0 & 0 & -i \\
        0 & 0 & 0 \\
        i & 0 & 0 
    \end{pmatrix}, \quad 
    \Lambda^6 = \begin{pmatrix}
        0 & 0 & 0 \\
        0 & 0 & 1 \\
        0 & 1 & 0 
    \end{pmatrix}, \quad \\
    \Lambda^7 & = \begin{pmatrix}
        0 & 0 & 0 \\
        0 & 0 & -i \\
        0 & i & 0 
    \end{pmatrix}, \quad 
    \Lambda^8 = \frac{1}{\sqrt{3}} \begin{pmatrix}
        1 & 0 & 0 \\
        0 & 1 & 0 \\
        0 & 0 & -2 
    \end{pmatrix}.
\end{aligned}
\end{equation}
\subsection{The partial isometry $W$}
\label{W}
We denote the basis of $\mathbb{C}^3$ by $\{\ket{r}, \ket{b}, \ket{g}\}$ and the basis of $\mathbb{C}^{10}$ by $\{\ket{I_3, Y}\}_{I_3,Y}$ according to its weight diagram \cite{cornwell_group_1997}.

The partial isometry $W$ exchanges the basis vectors between $\mathbb{C}^3 \otimes \mathbb{C}^3 \otimes \mathbb{C}^3$ and $\mathbb{C}^{10}$ with the symmetry considerations of representation $\mathbf{10}$.
\begin{equation}
    \begin{aligned}
        W =& \ket{-\frac{3}{2},1}\bra{rrr} + \frac{1}{\sqrt{3}} \ket{-\frac{1}{2},1} (\bra{rrb} + \bra{rbr} + \bra{brr}) + \frac{1}{\sqrt{3}} \ket{\frac{1}{2}, 1} (\bra{rbb}+\bra{brb}+\bra{bbr}) \\ & + \ket{\frac{3}{2},1}\bra{bbb} + \frac{1}{\sqrt{3}} \ket{-1,0}(\bra{rrg} + \bra{rgr} + \bra{grr}) \\ & + \frac{1}{\sqrt{6}}\ket{0,0}(\bra{rgb} + \bra{gbr} + \bra{brg} + \bra{rbg} + \bra{bgr} + \bra{grb}) \\ & + \frac{1}{\sqrt{3}} \ket{1,0}(\bra{bbg} + \bra{bgb} + \bra{gbb}) + \frac{1}{\sqrt{3}}\ket{-\frac{1}{2}, -1}(\bra{rgg} + \bra{grg} + \bra{ggr}) \\ & + \frac{1}{\sqrt{3}} \ket{\frac{1}{2},-1}(\bra{bgg} + \bra{gbg} + \bra{ggb}) + \ket{0,-2}\bra{ggg}
    \end{aligned}
\end{equation}

\end{document}